\documentclass[11pt]{article}
\usepackage{macros}

\title{Fast Algorithms for Sparse PCA and Robust Sparse Estimation \blfootnote{Authors are listed in alphabetical order.}}

\author{
Giannis Iakovidis\\
University of Wisconsin-Madison\\
{\tt iakovidis@wisc.edu}
\and
Ankit Pensia\\
Carnegie Mellon University\\
{\tt ankitp@cmu.edu}
}

\begin{document}
\maketitle

\begin{abstract}
We study fast algorithms for sparse-PCA certification.  Given a positive
semidefinite matrix $M$, the problem asks either to rule out a large
$k$-sparse quadratic form or to return a high-value (relaxed) witness.  The standard
semidefinite relaxation provides such certificates, but existing
general-purpose solvers require $\Omega(d^4)$ time.  We give a bicriteria
algorithm running in $O(d^2+d k^{O(\log k)})$ time: if some $k$-sparse unit
vector has quadratic form greater than $2$, it returns either an
$O(k^2)$-sparse unit vector or an SDP-feasible matrix of value at least $1$.
For $k\leq\exp(O(\sqrt{\log d}))$, this running time is $O(d^2)$.  We also go
below the quadratic barrier in the sample-access model:  
Given
$n=d^{o(1)}$ samples, our algorithm obtains a related one-sided certificate
in $d^{2 - \Omega(1)}$ time for $k=\polylog(d)$, without forming the empirical
covariance matrix.
As an application, these certificate routines yield the first quadratic and
subquadratic-time algorithms for robust sparse estimation for broad families
of distributions.  
Our sparse-PCA algorithm reduces a high-value sparse
direction to a bounded-radius set in the graph of large correlations and
searches the resulting candidate supports.  The subquadratic implementation
constructs this graph using fast correlation detection.
\end{abstract}
\tableofcontents
\section{Introduction} %
\label{sec:introduction}
 
\subsection{Motivation}

Principal component analysis (PCA) is a fundamental problem in
high-dimensional data analysis and is widely used for dimensionality
reduction.
In its classical form, PCA is defined as follows: Given a positive semidefinite matrix $M \in \R^{d\times d}$, find a unit vector $v \in \R^d$ that (approximately) maximizes the quadratic form $v^\top M v$. When $M$ is the (empirical) covariance matrix of data, such directions are considered  ``interesting'' because they capture the variability in the data.

Many high-dimensional distributions  possess additional structure: the directions of interest are rather \emph{sparse}, meaning that they depend on  only a few coordinates; see, for example, the textbooks ~\cite{EldKut12,HasTW15,vandeGeer16}.
This naturally leads to the problem of \emph{Sparse PCA}, whose goal
is to find a $k$-sparse unit vector that (approximately) maximizes the quadratic form; we say that a vector
$x$ is $k$-sparse if at most $k$ of its coordinates are
nonzero~\cite{ZouHT06,dAsGJL2007}.
Leveraging this sparsity structure can yield substantial savings in sample complexity \cite{johnstone2009consistency,amini2008high} and improve interpretability \cite{ZouHT06}. Sparse PCA has also found applications in genomics, neuroimaging, and large-scale text analysis~\cite{lee2010super,ulfarsson2007sparse,sparsepcatextdata}.
To formalize this objective,
let $M\in\R^{d\times d}$ be a positive semidefinite matrix, and let $k \in [d]$ be a sparsity
parameter.  Define the maximum quadratic form (variance) over $k$-sparse unit
vectors by
\begin{align}
\label{eq:sparse-op-def}
    \|M\|_{\mathrm{op},k}
    &:=
    \max_{v: \|v\|_2 = 1;  \|v\|_0 \leq k} v^\top M v.
\end{align}
Equivalently, lifting $v$ to the rank-one matrix $vv^\top$ gives
\begin{align}
    \|M\|_{\mathrm{op},k}
    &=
    \max_{V\in\cX_{\mathrm{sparse},k}}\langle M,V\rangle, 
\label{eq:sparse-set-def}
\quad     \text{ where  }\quad {\cX_{\mathrm{sparse},k}}
   & {:=
    \left\{vv^\top:
    v\in\R^d,\ \|v\|_2=1,\ \|v\|_0\le k\right\}\,.}
\end{align}
We study this objective both at the population level, where
$M=\Sigma$ is the covariance matrix of the underlying distribution on $\R^d$,
and at the sample level, where $M=\widehat\Sigma$ is the empirical covariance of samples $x_1,\dots,x_n$.
Our focus is the highly sparse regime $k=\polylog(d)$.

Unlike classical PCA, optimizing the sparse quadratic form
in \Cref{eq:sparse-op-def} is computationally hard in both worst-case and
average-case models.
For worst-case positive semidefinite
inputs, approximating the optimum within a factor
$1-\varepsilon_0$ is NP-hard for some constant $\varepsilon_0>0$ \cite{chan2016approximability}.  Even for
Gaussian average-case models, lower bounds are known in the statistical-query
(SQ) and sum-of-squares (SoS) frameworks and through average-case reductions,
yielding statistical--computational gaps~\cite{BerRig13,PotRaj22,HopKPRSS17,BreBre20}. 
These computational barriers motivate tractable
relaxations of $\|M\|_{\mathrm{op},k}$.
The standard semidefinite relaxation
enlarges the sparse rank-one class $\cX_{\sparse,k}$ in \Cref{eq:sparse-set-def} to
\begin{align}
\label{eq:sdp-set-def}
    \cX_{\sdp,k}
    &:=
    \left\{X\in\R^{d\times d}:
    X\succeq0,\ \trace(X)=1,\ \|X\|_1\le k\right\},
    \end{align}
   and  the associated SDP value is
    \begin{align}
    \label{eq:sdp-op-def}
    \|M\|_{\sdp,k}
    &:=
    \max_{X\in\cX_{\sdp,k}}\langle M,X\rangle.
\end{align}
Here and below, $\|X\|_1$ denotes the entrywise $\ell_1$ norm.  Since
$\cX_{\mathrm{sparse},k}\subseteq\cX_{\sdp,k}$, the SDP value upper
bounds $\|M\|_{\mathrm{op},k}$.
Proposed by \cite{dAsGJL2007}, this polynomial-time relaxation
retains the sample-complexity benefits of the sparse formulation and is also
used in robust statistics~\cite{BalDLS17}.  However, existing solvers for the
SDP require $\Omega(d^4)$ time; see the discussion in~\cite{KSTZ26}.

This high computational cost motivates faster algorithms. Partial
progress is known under additional structure on the covariance matrix.
The simplest such structure is the spiked identity model
$\Sigma=I+\theta vv^\top$, where $v$ is $k$-sparse. 
The algebraic identity  $\Sigma_{ii}=1+\theta v_i^2$ suggests diagonal thresholding: estimate the
support of $v$ from the largest empirical marginal variances.  This approach
can fail for general covariance matrices, whose diagonal entries need not
identify a high-value sparse direction.  \cite{KSTZ26} go
beyond the spiked identity model and give a $d^2$-time algorithm when the
population covariance $\Sigma$ has a sparse leading eigenvector.
Although this algorithm extends beyond the spiked identity
model, its sparse-leading-eigenvector assumption does not cover the settings
considered here, for two reasons.
First, contamination can create a dominant dense
direction at either the population or sample level, even when the clean
covariance has a sparse leading eigenvector.  Second, exact sparsity is not
stable under sampling: the empirical covariance of $n= o(d)$ dense samples
generically does not have a sparse leading eigenvector.  These two obstructions arise, respectively, for covariance matrices
encountered in robust filtering and for ordinary high-dimensional empirical
covariances.
In fact, we show that their algorithm
can fail on such inputs (see \Cref{lem:rtpm-counterexample});
consequently, no fast
algorithm is currently known for the general problem.

To avoid imposing additional structure on $M$, we instead relax the
form of the witness.  Fix a sufficiently large universal constant $C_0$ and
define
\begin{align*}
    \cX_{\mathrm{relaxed},k} \quad
    &:=\quad
    \cX_{\mathrm{sparse},C_0k^2}
    \cup\cX_{\sdp,k},\\[0.2cm]
    \|M\|_{\mathrm{relaxed},k}
    \quad&:= \max_{X\in\cX_{\mathrm{relaxed},k}}\langle M,X\rangle
     \quad = \quad\max\left\{
    \|M\|_{\mathrm{op},C_0k^2},
    \|M\|_{\sdp,k}
    \right\}.
    \numberthis
\end{align*}
Thus, an admissible witness is either a rank-one projector onto a
$C_0k^2$-sparse unit vector or a matrix feasible for the standard SDP.

This relaxation preserves two properties needed below.
First,
$\cX_{\mathrm{sparse},k}\subseteq\cX_{\mathrm{relaxed},k}$, so an upper
bound on $\|M\|_{\mathrm{relaxed},k}$ also certifies an upper bound on
$\|M\|_{\mathrm{op},k}$.  
Second, it has comparable sample complexity:
for light-tailed distributions such as Gaussians, $O(k^2\log d)$ i.i.d.\ samples
suffice for the empirical relaxed value to approximate
$\|\Sigma\|_{\mathrm{relaxed},k}$ to constant accuracy.
Thus, in statistical context, the sample complexity for $\|\cdot\|_{\relaxed,k}$ is similar to that for $\|\cdot\|_{\sdp,k}$; see \Cref{sec:sample-complexity}.

We do not seek to approximate $\|M\|_{\mathrm{op},k}$ or
$\|M\|_{\mathrm{relaxed},k}$ on every input.  Instead, we consider the
following one-sided certification task.

\begin{definition}[Sparse-PCA certification problem]
\label{def:certification-problem}
Given $M$, either certify that
$\|M\|_{\op,k}\leq2$, or return a matrix
$A\in\cX_{\mathrm{relaxed},k}$ such that
$\langle M,A\rangle\geq 1$.
\end{definition}

The corresponding decision problem is weaker: it asks only to distinguish
between $\|M\|_{\op,k}>2$ and
$\|M\|_{\mathrm{relaxed},k}<1$, without requiring an explicit witness.
The standard SDP solves both problems, but existing general-purpose solvers
incur the $\Omega(d^4)$ running time discussed above.  Moreover, the faster
methods described earlier require additional structure on $M$ and do not
solve even this weaker decision problem for general inputs.

\begin{question}[Quadratic-time sparse-PCA certification]
\label{q:quadratic}
Can the sparse-PCA certification problem be solved in $O(d^2)$ time?
\end{question}

The stronger certificate output, as opposed to the decision problem, is important beyond sparse PCA.  In iterative
filtering for robust sparse estimation~\cite{DiaKan22-book}, one repeatedly tests whether the
empirical covariance has a large sparse quadratic form.  When it does, a
high-value witness defines scores used to remove outliers.  Certificate
computation is therefore the main computational bottleneck.  For several
distribution families, existing filtering algorithms use the standard SDP
and inherit its $\Omega(d^4)$ running time.  This motivates the following question, variants
of which have been raised in
\cite{Diakonikolas2019,Cheng2021,Pensia24-subquad}.

\begin{question}[Fast robust sparse estimation]
\label{q:fast-robust}
    Are there fast algorithms for robust sparse estimation for general distribution families?
\end{question}

In high dimensions, even $\Theta(d^2)$ time can be
prohibitive in many applications.
Suppose that $M=\widehat\Sigma$ is the empirical covariance of
$n$ samples $x_1,\dots,x_n\in\R^d$.   
This implicit representation of $M$ has
size $O(nd)$, and in sparse-estimation settings one typically has $n\ll d$.
Consequently, the input size is $o(d^2)$, raising the question of whether the
certificate problem can be solved in $o(d^2)$ time in this sample-access
setting.  For example, in the special case of the spiked identity model,
diagonal thresholding runs in $O(nd)$ time, although it does not solve the
general certificate problem.  This motivates our second question.

\begin{question}[Subquadratic certificates from samples]\label{q:subquadratic}
Given direct access to $n=d^{o(1)}$ samples whose empirical covariance
is $M$, can one obtain an analogous one-sided certificate guarantee in time $o(d^2)$?
\end{question}

We answer all three questions affirmatively in the 
regime $k=\polylog(d)$.

\subsection{Our results}
\label{sec:our-results}

We first answer \Cref{q:quadratic} with oracle access to the matrix
entries.  We then apply the resulting certificate to robust sparse mean
estimation before returning to \Cref{q:subquadratic} in the sample-access
model.

\paragraph{Quadratic-time Sparse PCA.}
Our first result concerns \Cref{def:certification-problem} and answers \Cref{q:quadratic}.

\begin{theorem}[Informal; quadratic-time Sparse PCA]
\label{thm:informal-quadratic-sdp}
There is a deterministic algorithm that, given $M\succeq0$ and a sparsity parameter $k$, runs in time $O\!\left(d^2+d k^{O(\log k)}\right)$ and, if $\|M\|_{\op,k}>2$, returns a matrix $X\in\cX_{\mathrm{relaxed},k}$ such that $\langle X,M\rangle>1$.
\end{theorem}
The formal statement appears in
\Cref{cor:parameterized-radius-sparse-pca}.

If the algorithm does not return a relaxed witness of value
greater than $1$, then $\lVert M\rVert_{\mathrm{op},k}\leq 2$.  Thus,
\Cref{thm:informal-quadratic-sdp} solves the certification problem in
\Cref{def:certification-problem} and hence the corresponding decision problem. For $k=\exp(\sqrt{\log d}/C)=d^{o(1)}$, where $C>0$ is
a sufficiently large universal constant, the running time remains $O(d^2)$.
This sparsity range is asymptotically larger than every polylogarithmic function
of $d$, and therefore includes the regime emphasized above.  The result
improves on the stated $\Omega(d^4)$ running time of the standard SDP approach
and answers \Cref{q:quadratic}.
\begin{remark}[The approximation--runtime tradeoff]
\label{rem:informal-tradeoff}
The approximation factor $2$ is not essential: for any fixed
constant multiplicative gap, the algorithm achieves the same asymptotic
running time.  The general tradeoff is as follows.\begin{itemize}
    \item Suppose that $M$ has a $k$-sparse unit vector of value at least $1$. 
For any integer $L\ge1$, the algorithm runs in time
    $O(d^2+d\,k^{O(L)})$ and returns a certificate of value at least
    $k^{-1/(2L)}$; thus, the approximation factor is at most
    $k^{1/(2L)}$.  See \Cref{cor:parameterized-radius-sparse-pca} for the
    formal tradeoff.
    \item The informal theorem above follows by taking $L = \Theta(\log k)$.

    \item Taking $L$ to be a sufficiently large constant yields a
    $d^2\poly(k)$-time algorithm with a $k^{\eta}$ approximation factor for a small constant $\eta$,
    improving on the trivial $\sqrt{k}$ approximation factor.

\end{itemize}

\end{remark}

The routine can instead return a high-value sparse vector,
retaining the classical output form at the cost of allowing a larger
support.
\begin{remark}[Finding a high-value sparse vector]
\label{rem:informal-power-iteration}
For every $L\in\Z_+$, if $M$ has a $k$-sparse unit vector of value at least $1$, then in time $O(d^2+dk^{O(L)})$ our algorithm returns a $k^{O(L)}$-sparse unit vector of value at least $k^{-1/(2L)}$. 
We keep the semidefinite result primary because its $\ell_1$ budget $k$ gives a stronger sample complexity guarantee for downstream applications.
We refer the reader to \Cref{cor:parameterized-radius-sparse-pca} for the formal version of this guarantee.
\end{remark}

\paragraph{Subquadratic time Sparse PCA.}
We now turn to \Cref{q:subquadratic}.
When the covariance matrix is
given implicitly through samples, one need not form all $d^2$ entries.  We
combine our certificate routine with a fast correlation-detection
subroutine~\cite{Valiant15}, following a strategy proposed by
\cite{Pensia24-subquad}.
\begin{theorem}[Informal; subquadratic Sparse PCA]
\label{thm:informal-subquadratic-sdp}
Fix a sufficiently large constant $\Gamma>0$. There is a randomized algorithm that, given direct access to $n=d^{o(1)}$ samples with empirical covariance matrix $M$ and sparsity parameter $k$, runs in time $d^{1.72}k^{O(\log k)}$ and, if $\|M\|_{\op,k}>2$, with high probability returns one of the following:
\begin{enumerate}[label=(\roman*)]
    \item a unit vector $u$ with $\|u\|_0\le k^{O(\Gamma)}$ and $u^\top M u>1$;
    \item a matrix $X\in\cX_{\sdp,k}$ with $\langle M,X\rangle\ge1$.
\end{enumerate}
\end{theorem}

We refer the reader to \Cref{cor:subquadratic-sdp} for the formal version of this result.

This answers \Cref{q:subquadratic} affirmatively. The $k^{O(\Gamma)}$ sparsity in the direct-vector branch is the price of making the correlation-detection step subquadratic, whereas the semidefinite branch retains the original $\ell_1$ budget $k$.

Sparse-PCA certification is also the main computational primitive in
iterative filtering for robust sparse mean estimation.  We next use our
certificate routines to obtain fast estimators in two distributional settings.

\paragraph{Applications: Robust sparse mean estimation.}
Robust mean estimation asks for an accurate estimate of the mean of a distribution from possibly contaminated samples \cite{DiaKan22-book}. 
We formally define the contamination model below:
\begin{definition}[Contamination Model]
Let $S=(x_1,\ldots,x_n)$ be a dataset.
A dataset $T=(y_1,\ldots,y_n)$ is an $\eps$-corruption of $S$ if it can be obtained by replacing at most $\eps n$ elements of $S$ with arbitrary points, possibly chosen by an adversary after observing $S$.
If $S$ is a set of $n$ i.i.d.\ samples from a distribution $P$, we call such a dataset $T$ an $\eps$-corrupted set of samples from $P$.
\end{definition}
We focus on robust sparse mean estimation.  Given
$\epsilon$-corrupted samples from a distribution $P$ with $k$-sparse mean
$\mu$, the goal is to estimate $\mu$ accurately using few samples.  Naive
methods such as coordinatewise estimation have asymptotic error that scales as
$\sqrt{k}f(\eps)$, where $f$ depends on the distribution family.  In contrast,
the information-theoretic error $f(\epsilon)$ is independent of $k$, although
algorithms attaining it can be computationally expensive.  We therefore seek
three guarantees simultaneously: (i) $\|\widehat{\mu}-\mu\|_2$ independent of
$k$ (or a slowly growing function of $k$), (ii) sample complexity $\poly(k,\log d)$, and (iii) mild runtime dependence
on the ambient dimension $d$.

As noted above, existing polynomial-time algorithms achieving all
three guarantees require a tractable relaxation of the sparse-PCA
certification problem.  For many distribution families, only the SDP
relaxation is known to achieve the desired error and sample complexity in
polynomial time; with existing solvers, these algorithms therefore run in
$\Omega(d^4)$ time.
Our fast sparse-PCA certificate routine from \Cref{thm:informal-quadratic-sdp} overcomes this computational bottleneck for these distribution families.

\begin{theorem}[Informal; robust sparse mean estimation]
\label{thm:informal-robust-mean}
Let $P$ be a distribution on $\R^d$ with $k$-sparse mean $\mu$.
There are two algorithms which, given an $\eps$-corrupted set  of $n$
 samples from $P$, return an estimate $\widehat\mu$ with
probability at least $0.99$:
\begin{enumerate}[label=(\roman*)]
    \item \label{thm:informal-bounded-covariance-mean}

    If the covariance matrix $\Sigma$ satisfies $\|\Sigma\|_{\op}\le1$
    and $\E_P[(x_j-\mu_j)^4]=O(1)$ for every $j\in[d]$, then for
    $n=O(k^{4}\log(d)/\eps)$ the first algorithm runs in time
    $\widetilde O(d^2k^{O(\log k)}/\eps^2)$ and satisfies
    $\|\widehat\mu-\mu\|_2=O(\sqrt{\eps})$.

    \item \label{thm:informal-subgaussian-mean}
   
    If $P$ is subgaussian with covariance matrix $I$, then for
    $n=O(k^2\log(d)/\eps^2)$ the second algorithm runs in time
    $\widetilde O(d^2k^{O(\log k)}/\eps^{O(1)})$ and satisfies
    $\|\widehat\mu-\mu\|_2=\widetilde O(\eps)$.
\end{enumerate}
\end{theorem}
See
\Cref{cor:robust-sparse-mean-sdp-tradeoff,cor:identity-subgaussian-sparse-mean-tradeoff}
for the formal versions of the two guarantees.

The error guarantees in the two parts of the preceding theorem are information-theoretically optimal for both subgaussian distributions~\cite{DiaKan22-book} and heavy-tailed distributions (for large enough $k$)~\cite{DiaKLP22}. Moreover, the $k^2$ dependence in the sample complexity for subgaussian distributions is conjectured to be optimal among polynomial-time algorithms~\cite{DiaKS17,BreBre20}.
Thus, in the regime of interest for $k$, the above result gives the first algorithms for broad families of distributions for robust sparse estimation with sparsity-respecting sample complexities and quadratic dependence on the ambient dimension in their running times.

In fact, this is the \emph{first polynomial-time} guarantee for
isotropic subgaussian distributions that achieves $\widetilde{O}(\epsilon)$ error
with $\poly(k,\log d,1/\eps)$ samples.  The existing SDP approach is indeed polynomial-time,
but its correctness requires that the inlier samples from the subgaussian distribution satisfy a deterministic condition (termed stability). 
Establishing this stability requires a uniform bound on a
quadratic empirical process indexed by $\cX_{\sdp,k}$.  
Such a bound was known
for Gaussian data, but not for general subgaussian data. 
{Prior works incorrectly applied {the Hanson--Wright inequality to} general subgaussian distributions, invalidating the claims.}~\cite{CheDKGGS21,SasFuj24-arxiv}
{We prove the required bound
for this index set using results from probability theory
\cite{hua2026talagrand}.
}

The  guarantees in \Cref{thm:informal-robust-mean} are one point on the following more general
approximation--runtime tradeoff.
\begin{remark}[The approximation--runtime tradeoff]
\label{rem:informal-tradeoff-robust}
As in \Cref{rem:informal-tradeoff}, the same parameter $L$ gives the
following robust-estimation guarantees (with the same number of samples as in \Cref{thm:informal-robust-mean}) with running time $\widetilde{O}(d^2k^{O(L)}/\eps^{O(1)})$:
\begin{enumerate}[label=(\roman*)]
    \item error $O(k^{1/(2L)}\sqrt\eps)$ for bounded-covariance mean estimation; 
    \item error $\widetilde O( k^{1/(4L)}\eps)$ for identity-covariance subgaussian mean estimation.
\end{enumerate}
Taking $L=\Theta(\log k)$ recovers the near-optimal error guarantees of \Cref{thm:informal-robust-mean}, while smaller values of $L$ reduce the dependence on $k$ in the running time at the cost of the errors displayed above.
For any $L>1$, the sparsity-dependent factor in the error is better than the naive $\sqrt{k}$ factor achieved by the coordinate-wise estimators, while retaining polynomial-time and near-optimal sample complexity.
\end{remark}

The estimator also applies when the mean $\mu$ is not sparse. In this
case, the output $\widehat{\mu}$ achieves error comparable to that of the
best $k$-sparse approximation to $\mu$, up to an
estimation error; see \Cref{rem:sparse-direction-error}.

\begin{remark}[Robust sparse estimation in subquadratic time]
    {Replacing the quadratic-time sparse-PCA algorithm from \Cref{thm:informal-quadratic-sdp} with the subquadratic-time subroutine from \Cref{thm:informal-subquadratic-sdp}} gives subquadratic variants of the
robust sparse mean estimators in \Cref{thm:informal-robust-mean}, albeit with increased sample complexity;
see \Cref{cor:subquadratic-identity-subgaussian-sparse-mean}.  These variants
make further progress on \Cref{q:fast-robust} and partially answer questions raised
in \cite{Diakonikolas2019,Cheng2021,Pensia24-subquad}.
\end{remark}

Although we do not pursue
these directions here, we expect similar runtime improvements for other
robust sparse-estimation tasks, such as sparse linear regression and sparse functional estimation.

\paragraph{Organization.}
The remainder of the paper is structured as follows. In \Cref{sec: sparse-pca}, we develop our fast sparse-PCA algorithm, establish its quadratic-time and subquadratic correlation-detection implementations, and present the general approximation--runtime tradeoff under both oracle and empirical access models.
In \Cref{sec:mean-estimation}, we reduce robust sparse mean estimation to the sparse-PCA certification and instantiate the reduction for bounded-covariance and identity-covariance subgaussian distributions. Finally, \Cref{app:related-work} presents the related work, \Cref{app:omitted,app:shifted-covariance-detection,app:raw-covariance-detection,app:subgaussian-stability} prove omitted technical details, 
and \Cref{app:rtpm-counterexample} gives a counterexample to the restarted truncated power method.

\paragraph{Notation and preliminaries.}
We write $\Z_+$ for the positive integers and $[d]=\{1,\ldots,d\}$. For quantities $a,b\in\R$, we write $a\gg b$ to mean that $a\ge Cb$ for a sufficiently large absolute constant $C$. For a vector $x\in\R^d$, we use $\|x\|_p$ for its $\ell_p$ norm, $\supp(x)$ for its support, and $\|x\|_0=|\supp(x)|$ for its sparsity. For $s\in[d]$, $\topp_s(x)$ denotes the vector obtained by retaining the $s$ coordinates of $x$ with largest absolute value and setting the remaining coordinates to zero, with ties broken arbitrarily. For $s\in[d]$, we write $\|x\|_{2,s}:=\sup_{\|v\|_2=1,\ \|v\|_0\le s}\langle v,x\rangle=\|\topp_s(x)\|_2$. For a matrix $M$, we write $\|M\|_{\op}$ for its operator norm, $\|M\|_1=\sum_{i,j}|M_{ij}|$ for its entrywise $\ell_1$ norm, and $\langle A,B\rangle=\operatorname{tr}(A^\top B)$ for the Frobenius inner product; we write $M\succeq0$ when $M$ is positive semidefinite and use $I_d$, or simply $I$ when the dimension is clear, for the identity matrix. For a dataset $S \subset \R^d$ and a vector $\mu \in \R^d$, define the empirical second moment of $S$ with respect to $\mu$ as $Q_{S}(\mu)\eqdef\frac1{|S|}\sum_{x\in S}(x-\mu)(x-\mu)^\top$. For a symmetric matrix $A\in\R^{m\times m}$, we write $\operatorname{SDP}_k(A):=\max\{\langle A,X\rangle:X\succeq0,\ \operatorname{tr}(X)=1,\ \|X\|_1\le k\}$; a feasible matrix $X$ is an additive-$\eta$ approximate solution if $\langle A,X\rangle\ge\operatorname{SDP}_k(A)-\eta$.
Our results do not require the matrix to be positive semidefinite in all directions; it suffices that it be positive semidefinite over sparse directions. This motivates the following relaxed definition. 
For $s\in\Z_+$, we say that a symmetric matrix $\Sigma\in\R^{d\times d}$ is $s$-PSD if $\Sigma_{U,U}\succeq0$ for every $U\subseteq[d]$ with $|U|\le s$; in particular, every positive semidefinite matrix is $s$-PSD for every $s$.
We write $x\sim P$ when $x$ is distributed according to $P$, and use $\E_P[x]$ and $\cov_P(x)$ for its mean and covariance. Throughout, a subgaussian distribution means that its centered random vector $x-\E x$ has subgaussian norm $\|x-\E x\|_{\psi_2}:=\sup_{\|u\|_2=1}\|\langle u,x-\E x\rangle\|_{\psi_2}$ bounded by a universal constant. The corresponding guarantees for subgaussian norm bounded by a general parameter $K$ follow by applying the normalized statements to $x/K$ and rescaling the resulting estimates and covariance parameters to the original scale.

\section{Fast algorithm for Sparse PCA}
\label{sec: sparse-pca}
Let $d \in \mathbb{Z}_+$, let $k \in \mathbb{Z}_+$ be a sparsity parameter with $k \ll d$, and let $\Sigma \in \mathbb{R}^{d \times d}$ be a positive semidefinite matrix.
We say that $\Sigma$ has a big $k$-sparse direction if there exists a vector $v \in \R^d$ with $\|v\|_2 = 1$ and $\|v\|_0 \le k$ such that $
v^\top \Sigma v \ge c$,
where $c>0$.
Our goal is, if such a vector exists, find a unit vector $u\in \R^d$ with larger allowed sparsity  $\|u\|_0 \le r(k)$, such that $u^\top \Sigma u > c'$ for some $c' > 0$. 
Note that in the above definition $v$ need not be an eigenvector of $\Sigma$.

In this section, we present an algorithm for sparse PCA with
running time $C_kd^2$, where $C_k$ depends only on $k$.  
We explain the basic premise of our algorithm in \Cref{sec: fpt-algorithm}, which will lead to a $d^2k^{O(k)}$-time algorithm.
In \Cref{sec:small-radius}, we show how to improve the runtime to $d^2 k^{O(\log k)}$, proving \Cref{thm:informal-quadratic-sdp}.

In
\Cref{sec:sub-quadratic}, we give a faster algorithm whose dependence on $d$
has exponent smaller than $2$, proving \Cref{thm:informal-subquadratic-sdp}.
The algorithm is based on the observation that if there exists a vector $v$ attaining a large quadratic form, then there should exist a sparse vector $v'$ that attains a similar quadratic form and satisfies an additional connectivity property. 
This connectivity lets us discover one support coordinate from
another.
After brute-force recovery of $C_k d$ candidate supports, each in $C_kd $ time, we run a constrained eigenvalue method on each support and return the candidate with the largest quadratic form.

\subsection{Graph Connectivity and One-Step Truncated Power Method}
\label{sec: fpt-algorithm}
We first present an algorithm under an additional structural assumption of connectivity (\Cref{def:connected-vector}), which we later show is implied by the existence of a target vector with sufficiently large quadratic form. 

The key ingredient of our algorithm is the following lemma, which shows that for sufficiently large truncation parameter $r$, one truncated power-method step either produces a valid output vector or discovers all coordinates that are sufficiently correlated with the current one. More precisely,  if we aim for a target quadratic form $\beta$, then choosing $r\geq \beta^2/\tau^2$ discovers all coordinates correlated by at least $\tau$ with our initial vector or we can immediately output a vector with quadratic form at least $\beta$. For example, suppose we start from $e_i$ with $i \in \supp(v)$. If neither $e_i$ nor $\topp_r(\Sigma e_i)/\|\topp_r(\Sigma e_i)\|$ has quadratic form larger than $\beta$, then $\topp_r(\Sigma e_i)/\|\topp_r(\Sigma e_i)\|$ contains in its support every index $j$ such that $|\Sigma_{i,j}| \ge \tau$. Consequently, if a sparse target vector $v$ has $i,j \in \supp(v)$ with $|\Sigma_{i,j}| \ge \tau$, then starting from $i$, either the algorithm has already found a valid output vector, or among the $r$ coordinates retained by the truncated step it discovers another element $j$ of $\supp(v)$.
\begin{lemma}[Discovery Lemma]
\label{lem:discover}
Let $\Sigma\in\R^{d\times d}$ be symmetric. Let $x\in\R^d$ be a unit vector and let
\[
    y=\Sigma x,
    \qquad
    z=\topp_r(y),
    \qquad
    u=\frac{z}{\|z\|_2}.
\]
Set $S:=\supp(x)\cup\supp(z)$ and assume that $\Sigma_{S,S}\succeq 0$.
Assume $z\neq 0$.
If a coordinate $a\in[d]$ is not kept by $\topp_r(y)$, then
$$    u^\top\Sigma u
    \ge
    \frac{r y_a^2}{x^\top\Sigma x}\;.$$
\end{lemma}

\begin{proof}
Since $u=z/\|z\|_2$, we have
\[
    u^\top\Sigma x
    =u^\top y
    =\frac{z^\top y}{\|z\|_2}
    =\frac{\|z\|_2^2}{\|z\|_2}
    =\|z\|_2.
\]
Both $x$ and $u$ are supported on $S$. Therefore, by Cauchy--Schwarz in the seminorm induced by the positive semidefinite matrix $\Sigma_{S,S}$,
$(u^\top\Sigma x)^2\le (u^\top\Sigma u)(x^\top\Sigma x)$.
We next verify that the denominator is positive. If $x^\top\Sigma x=0$, then $x_S^\top\Sigma_{S,S}x_S=0$. Since $\Sigma_{S,S}\succeq0$, this implies $\Sigma_{S,S}x_S=0$. Because $z$ is supported on $S$ and $x$ is supported on $S$, we get
\[
    \|z\|_2^2
    =z^\top y
    =z^\top\Sigma x
    =z_S^\top\Sigma_{S,S}x_S
    =0,
\]
contradicting $z\neq0$. Thus $x^\top\Sigma x>0$.
Therefore, $u^\top\Sigma u\ge \|z\|_2^2/(x^\top\Sigma x)$.
If coordinate $a$ is dropped, then all $r$ retained coordinates have magnitude at least $|y_a|$. Hence
$\|z\|_2^2\ge r y_a^2$.
This concludes the proof of \Cref{lem:discover}.
\end{proof}
We now define a graph on the coordinates, where edges correspond to pairs of coordinates with nontrivial correlation. We call a vector connected if its support is connected in this graph. Intuitively, the truncated-power discovery lemma above shows that, as long as the algorithm has not already found a vector with quadratic form larger than $\beta$, one truncated step from a coordinate $i$ discovers all coordinates $j$ with $|\Sigma_{ij}| \ge \tau$, provided $r>\beta^2/\tau^2$. Repeating this argument allows the algorithm to explore the support of a target vector, provided that this support is connected in the large-correlation graph.

\begin{definition}[$\tau$-connected support] \label{def:connected-vector}
Let $\Sigma\in\mathbb{R}^{d\times d}$ be symmetric and let $T\subseteq[d]$.
The $\tau$-large-correlation graph on $T$, denoted $G_\tau(T)$, is the graph with vertex set $T$ and an edge $\{i,j\}$, $i\neq j$, whenever $|\Sigma_{ij}|\ge \tau$.
We say that $T$ is $\tau$-connected if $G_\tau(T)$ is connected.
Similarly, we say that a vector $x\in\mathbb{R}^d$ is $\tau$-connected if $\supp(x)$ is $\tau$-connected.
We say that a set $E\subseteq[d]$ has $\tau$-correlation radius at most $L$ if the graph $G_\tau(E)$ has radius at most $L$, that is, if there exists a center $i_\star\in E$ such that every $j\in E$ has graph distance at most $L$ from $i_\star$ in $G_\tau(E)$.
\end{definition}

We are now ready to present our main structural algorithm (\Cref{alg:branching-truncated-power}) and prove its fixed-parameter tractable guarantee.

\begin{algorithm}[h]
    \centering
    \fbox{\parbox{6in}{
        {\bf Input:}
        A symmetric matrix $\Sigma\in\R^{d\times d}$, sparsity parameter $k\in\Z_+$, search-radius parameter $L$, output quadratic form threshold $\beta>0$, correlation threshold $\tau>0$, and truncation parameter $r\in\Z_+$.\\
        {\bf Output:}
        Either a unit vector $u\in\R^d$, or a family of candidate supports $\mathcal{B}$.

        \begin{enumerate}[leftmargin=*]

            \item\label{line:compute-neighborhoods}
            For every coordinate $i\in[d]$:
            \begin{enumerate}[leftmargin=*, nosep]
                \item If $\Sigma_{ii}>\beta$, return $u=e_i$.
                \item Compute
                $
                    z_i\gets \topp_r(\Sigma e_i).
                $
                \item If $z_i\neq 0$, set $u_i\gets z_i/\|z_i\|_2$. If $u_i^\top\Sigma u_i>\beta$, return $u=u_i$.
                \item Set
                $
                    N(i)\eqdef \supp(z_i).
                $
            \end{enumerate}

            \item\label{line:init-candidates}
            Initialize an empty family of candidate supports $\mathcal{B}\gets \emptyset$.

            \item\label{line:branch-from-all-coordinates}
            For every starting coordinate $i_0\in[d]$:
            \begin{enumerate}[leftmargin=*, nosep]
                \item Set
                \[
                    B^{(0)}_{i_0}\gets \{i_0\},
                \]
                and add $B^{(0)}_{i_0}$ to $\mathcal{B}$.

                \item For $t=0,1,\ldots,L-1$, define
                \[
                    B^{(t+1)}_{i_0}
                    \gets
                    B^{(t)}_{i_0}
                    \cup
                    \bigcup_{i\in B^{(t)}_{i_0}} N(i).
                \]
                Add $B^{(t+1)}_{i_0}$ to $\mathcal{B}$.
            \end{enumerate}

            \item\label{line:return-candidates}
            Return $\mathcal{B}$.
        \end{enumerate}
    }}
    \caption{Branching Truncated-Power Support Search.}
    \label{alg:branching-truncated-power}
\end{algorithm}

{Our algorithm performs a brute-force search in $G_{\tau}([d])$
for the support of a vector with large quadratic form, under the assumption
that this support is contained in a set of correlation radius at most $L$.}
{For each $i$, the algorithm first computes the truncated-power
update $u_i=\topp_r(\Sigma e_i)/\|\topp_r(\Sigma e_i)\|$ and checks the
quadratic forms of both $e_i$ and $u_i$.}
If one of these vectors has quadratic form larger than $\beta$, the algorithm returns it immediately. 
Otherwise, the algorithm constructs, from each starting coordinate $i$, the radius-$L$ branching region, of size at most $(r+1)^L$.
In detail, it computes, for every coordinate $i$, the set $N(i)$ of coordinates reached by one truncated-power step from $e_i$. 
It then initiates a branching search from each coordinate $i_0$, repeatedly expanding the current support by adding all neighborhoods $N(i)$ of the coordinates already discovered, for $L$ rounds. 
This produces a family of at most $d(L+1)$ candidate supports, each of size at most $(r+1)^L$.

Next, we prove the structural guarantee of our search method. The main statement is that if there exists a vector $v$ with large quadratic form whose support is contained in a set with correlation radius at most $L$, then the algorithm either returns an $r$-sparse vector with large quadratic form, or returns a family containing a support that contains $\supp(v)$.

\begin{theorem}[Branching support recovery with bounded correlation radius]
\label{thm:bounded-diameter-branching}
Let $\Sigma\in\R^{d\times d}$ be symmetric and $(r+1)$-PSD. Fix $\beta>0$, $\tau>0$, and $r\in\Z_+$ such that
$
    r>\beta^2/\tau^2.
$
Assume that there exist a unit vector $v\in\R^d$ and a set $E\subseteq[d]$ such that $\|v\|_0\le k$, $v^\top\Sigma v\geq \beta$, $\supp(v)\subseteq E$, and $E$ has $\tau$-correlation radius at most $L$. Then \Cref{alg:branching-truncated-power}, run with parameters $\beta,\tau,r$ and search-radius parameter $L$, returns one of the following:
\begin{enumerate}[(i)]
    \item a unit vector $u$ such that $\|u\|_0\le r$ and
    $
        u^\top\Sigma u>\beta;
    $
    \item a family $\mathcal{B}$ of at most $d(L+1)$ candidate supports such that $E\subseteq B$ for some $B\in\mathcal{B}$.
\end{enumerate}
Moreover, every $B\in\mathcal{B}$ satisfies $|B|\le (r+1)^L$, and 
the running time is
$O\!\left(d^2+dr^2+d(r+1)^L\right)$.
\end{theorem}

\begin{proof}
Let $C=E$. By assumption, $C$ has $\tau$-correlation radius at most $L$. Let $i_\star\in C$ be a radius center, so every vertex of $C$ is within graph distance at most $L$ from $i_\star$ in $G_\tau(C)$.

If the algorithm returns a vector, then by construction it satisfies the first conclusion. It remains to consider the case where the algorithm reaches Line \ref{line:return-candidates} and returns the family $\mathcal{B}$. In this case, for every $i\in[d]$,
$
    \Sigma_{ii}\le \beta
$
and, whenever $z_i=\topp_r(\Sigma e_i)\neq 0$, the vector $u_i=z_i/\|z_i\|_2$ satisfies
$
    u_i^\top\Sigma u_i\le \beta.
$

Consider the branch initialized at $i_\star$. For $t\ge 0$, let $B^{(t)}=B^{(t)}_{i_\star}$ be the supports constructed by the algorithm. We prove by induction that $B^{(t)}$ contains every vertex of $C$ at graph distance at most $t$ from $i_\star$ in $G_\tau(C)$.

The base case $t=0$ is immediate because $B^{(0)}=\{i_\star\}$.

Now suppose the claim holds at time $t$. Let $j\in C$ be a vertex at graph distance at most $t+1$ from $i_\star$. If $j$ already has distance at most $t$, then $j\in B^{(t)}\subseteq B^{(t+1)}$.

Otherwise, $j$ has distance exactly $t+1$. Then there exists a vertex $i\in C$ at distance $t$ from $i_\star$ such that $\{i,j\}\in E(G_\tau(C))$. By the induction hypothesis, $i\in B^{(t)}$. Also, by definition of $G_\tau(C)$, $|\Sigma_{ij}|\ge \tau$.

We claim that $j\in N(i)$. Suppose, for contradiction, that $j\notin N(i)$. Since $|\Sigma_{ij}|\ge\tau$, the vector $z_i=\topp_r(\Sigma e_i)$ is nonzero. In this instantiation of \Cref{lem:discover}, the set $\supp(e_i)\cup\supp(z_i)$ has size at most $r+1$, so the $(r+1)$-PSD assumption gives the positivity hypothesis of the discovery lemma. Applying \Cref{lem:discover} with $x=e_i$ and $a=j$, we get
$$
    u_i^\top\Sigma u_i
    \ge
    \frac{r\Sigma_{ij}^2}{\Sigma_{ii}}.
$$
Using $\Sigma_{ii}\le \beta$ and $|\Sigma_{ij}|\ge \tau$, this gives
$$
    u_i^\top\Sigma u_i
    \ge
    \frac{r\tau^2}{\beta}
    >\beta,
$$
where the final inequality is exactly $r\tau^2>\beta^2$. This contradicts the fact that the algorithm did not return $u_i$. Therefore $j\in N(i)$. Since $i\in B^{(t)}$, the construction of $B^{(t+1)}$ gives $N(i)\subseteq B^{(t+1)}$, and hence $j\in B^{(t+1)}$. This completes the induction.

Now every vertex of $C$ is at graph distance at most $L$ from the radius center $i_\star$. Therefore $C\subseteq B^{(L)}_{i_\star}$, and the second conclusion follows.

It remains to bound the number and sizes of the candidate supports and the running time. The algorithm constructs $L+1$ supports for each of the $d$ starting coordinates, so $|\mathcal{B}|\le d(L+1)$. By construction, $|B^{(0)}|=1$ and
$$
    |B^{(t+1)}|
    \le
    |B^{(t)}|+\sum_{i\in B^{(t)}} |N(i)|
    \le
    (r+1)|B^{(t)}|.
$$
Hence $|B^{(t)}|\le (r+1)^t$. In particular, every candidate support has size at most $(r+1)^L$.

Finally, the sets $N(i)$ can be computed for all $i\in[d]$ by scanning the columns of $\Sigma$ and keeping the $r$ largest entries in each column. The basis-vector and truncated-power quadratic form checks take an additional $d r^{2}$ time after these supports are formed. Constructing all candidate supports takes $d(r+1)^L$ time. 
Therefore, the total running time is
$O\!\left(d^2+dr^2+d(r+1)^L\right)$.
\end{proof}
The preceding theorem separates the combinatorial support-recovery step from the optimization procedure used on the resulting candidate supports. In the first outcome, the truncated-power search already produces an $r$-sparse unit vector with large quadratic form. In the second outcome, it produces a family of at most $d(L+1)$ supports of size at most $(r+1)^L$, one of which contains the support of the original $k$-sparse witness. This structural conclusion can be converted into different sparse PCA guarantees depending on how the candidate supports are subsequently processed.

We record two such consequences. The first applies the ordinary power method on each candidate support and returns a vector with large quadratic form. The second applies the standard sparse semidefinite relaxation on each candidate support and returns either the direct sparse vector found by the search or a positive semidefinite matrix with bounded entrywise $\ell_1$ norm and large objective value.
The former solves the sparse PCA task considered here by returning
a sparse vector with large quadratic form, although its sparsity may be
substantially larger than the target sparsity.  The latter returns a
certificate with effective sparsity comparable to the target sparsity, which
is better suited to downstream applications such as the robust sparse mean
estimation problem studied in \Cref{sec:mean-estimation}.

\begin{corollary}[Sparse PCA via the constrained power method]
\label{cor:bounded-diameter-constrained-power}
Fix $\delta\in(0,1)$, an output quadratic form threshold $\beta>0$, a power-method slack $\eta>0$, a correlation threshold $\tau>0$, and $r\in\Z_+$ satisfying $r\tau^2>\beta^2$. Let $\Sigma\in\R^{d\times d}$ be symmetric and $(r+1)^L$-PSD. Assume that there exist a unit vector $v\in\R^d$ and a set $E\subseteq[d]$ such that $\|v\|_0\le k$, $v^\top\Sigma v\geq \beta+\eta$, $\supp(v)\subseteq E$, and $E$ has $\tau$-correlation radius at most $L$. Run \Cref{alg:branching-truncated-power} with parameters $\beta,\tau,r$ and search-radius parameter $L$. If it returns a vector, return that vector. Otherwise, for every $B\in\mathcal{B}$ such that $\Sigma_{B,B}\neq 0$, run the random-start ordinary power method on $\Sigma_{B,B}$ long enough to return a vector of value at least $\lambda_{\max}(\Sigma_{B,B})-\eta$, with failure probability $\delta$, pad the resulting vector with zeros outside $B$, and return the vector with the largest quadratic form. Then, with probability at least $1-\delta$, the returned unit vector $u$ satisfies
$$
    u^\top\Sigma u\geq \beta.
$$
Moreover, $\|u\|_0\le \max\{r,(r+1)^L\}$, and the running time is
$O\!\left(
d^2+dr^2+d(r+1)^{O(L)}T_{\rm pm}
\right)$, where $T_{\mathrm{pm}}$ denotes the runtime of the power-method subroutine.
\end{corollary}

\begin{proof}
If \Cref{alg:branching-truncated-power} returns a vector, then its quadratic form is larger than $\beta$. Otherwise, by \Cref{thm:bounded-diameter-branching}, some $B\in\mathcal{B}$ contains $\supp(v)$, and hence
$$
    \lambda_{\max}(\Sigma_{B,B})\ge v^\top\Sigma v\ge \beta+\eta.
$$
By the guarantee of the power-method subroutine on this support, the returned vector has quadratic form at least $(\beta+\eta)-\eta=\beta$ with probability at least $1-\delta$. Since the procedure returns the vector with the largest quadratic form among all candidate supports, the returned vector also has quadratic form at least $\beta$ with probability at least $1-\delta$. The sparsity and running-time bounds follow from \Cref{thm:bounded-diameter-branching}.
\end{proof}

\begin{corollary}[Sparse PCA via the semidefinite relaxation]
\label{cor:bounded-diameter-sdp}
Fix an output quadratic form threshold $\beta>0$, an SDP additive-approximation slack $\eta>0$, a correlation threshold $\tau>0$, and $r\in\Z_+$ satisfying $r\tau^2>\beta^2$. Let $\Sigma\in\R^{d\times d}$ be symmetric and $(r+1)$-PSD. Assume that there exist a unit vector $v\in\R^d$ and a set $E\subseteq[d]$ such that $\|v\|_0\le k$, $v^\top\Sigma v\ge \beta+\eta$, $\supp(v)\subseteq E$, and $E$ has $\tau$-correlation radius at most $L$. Run \Cref{alg:branching-truncated-power} with parameters $\beta,\tau,r$ and search-radius parameter $L$.
If it returns a vector, return that vector. Otherwise, for every $B\in\mathcal{B}$ compute a feasible additive-$\eta$ approximate solution to $\operatorname{SDP}_k(\Sigma_{B,B})$.
Pad each approximate solution with zeros outside $B\times B$ and return the padded solution with the largest objective value.
Then the procedure returns one of the following:
\begin{enumerate}[(i)]
    \item a unit vector $u$ such that $\|u\|_0\le r$ and $u^\top\Sigma u>\beta$;
    \item a matrix $X\in\cX_{\sdp,k}$ such that
    $
        \langle\Sigma,X\rangle\ge \beta.
    $
\end{enumerate}
Using a polynomial-time SDP solver, the running time is
$O\!\left(
d^2+dr^2+
\frac{d(r+1)^{O(L)}}{\eta^{O(1)}}
\right)$
\end{corollary}

\begin{proof}
The direct-vector case follows from \Cref{thm:bounded-diameter-branching}. Otherwise, some $B\in\mathcal{B}$ contains $\supp(v)$. The matrix $vv^\top$, restricted to $B\times B$, is feasible for the SDP because
$$
    vv^\top\succeq0,
    \qquad
    \operatorname{tr}(vv^\top)=1,
    \qquad
    \|vv^\top\|_1=\|v\|_1^2\le k.
$$
Therefore
$$
    \operatorname{SDP}_k(\Sigma_{B,B})
    \ge
    \langle\Sigma,vv^\top\rangle
    =
    v^\top\Sigma v
    \ge
    \beta+\eta.
$$
Thus, the additive-$\eta$ approximate solution corresponding to this support has objective value at least $\beta$. Padding it with zeros preserves positive semidefiniteness, trace, entrywise $\ell_1$ norm, and objective value. Since the procedure returns the padded solution with the largest objective value, the returned matrix has objective value at least $\beta$. The running-time bound follows from the number and sizes of the candidate supports in \Cref{thm:bounded-diameter-branching}.
\end{proof}

We now remove the bounded-radius assumption. Specifically, we first show that for any sparse vector with large quadratic form, there exists a subset of coordinates that both induces a vector of comparable quadratic form and is connected in the large-correlation graph. Combining this reduction with \Cref{cor:bounded-diameter-constrained-power} gives the worst-case fixed-parameter sparse PCA guarantee below.

\begin{lemma}[Connectivity Lemma]
\label{lem:component-reduction}
Let $\tau>0$, let $v$ be a unit vector with $\supp(v)= S$, with $|S|\leq k$, and suppose
$
    v^\top\Sigma v\ge R.
$

Then there exists a unit vector $v'$ such that
$
    \supp(v')\subseteq \supp(v),
$
$\supp(v')$ is $\tau$-connected, and
$
    (v')^\top\Sigma v'
    \ge
    R-\tau k.
$
\end{lemma}

\begin{proof}
Let $C_1,\ldots,C_m$ be the connected components of the graph $G_\tau(S)$ defined as in \Cref{def:connected-vector}.
For each component $C_p$, let $v_{C_p}$ denote the restriction of $v$ to $C_p$.

We decompose the quadratic form as
\[
    v^\top\Sigma v
    =
    \sum_{p=1}^m v_{C_p}^\top\Sigma v_{C_p}
    +
    2\sum_{p<q} v_{C_p}^\top\Sigma v_{C_q}.
\]
If $p\neq q$, then there is no edge of $G_\tau(S)$ between $C_p$ and $C_q$.
Therefore, $|\Sigma_{ij}|<\tau$ for every $i\in C_p$ and
$j\in C_q$.
Hence the total contribution of the cross-component terms is bounded by
\[
    \left|
    2\sum_{p<q} v_{C_p}^\top\Sigma v_{C_q}
    \right|
    \le
    2\tau\sum_{p<q}\|v_{C_p}\|_1\|v_{C_q}\|_1.
\]
Using $2\sum_{p<q} a_p a_q \le \left(\sum_p a_p\right)^2$
with $a_p=\|v_{C_p}\|_1$, we get
\[
    \left|
    2\sum_{p<q} v_{C_p}^\top\Sigma v_{C_q}
    \right|
    \le
    \tau\left(\sum_p\|v_{C_p}\|_1\right)^2
    =
    \tau\|v\|_1^2.
\]
Since $v$ is supported on at most $k$ coordinates and $\|v\|_2=1$,
$\|v\|_1^2\le k$.
Therefore
\[
    \left|
    2\sum_{p<q} v_{C_p}^\top\Sigma v_{C_q}
    \right|
    \le
    \tau k.
\]
It follows that
\[
    \sum_{p=1}^m v_{C_p}^\top\Sigma v_{C_p}
    \ge
    v^\top\Sigma v-\tau k
    \ge
    R-\tau k.
\]

Now write $a_p=\|v_{C_p}\|_2^2$.  Then $\sum_p a_p=1$. For
every component with $a_p>0$, define
$w_p=v_{C_p}/\|v_{C_p}\|_2$.
Then $v_{C_p}^\top\Sigma v_{C_p}=a_p\, w_p^\top\Sigma w_p$.
Thus $\sum_{p:a_p>0} a_p\, w_p^\top\Sigma w_p\ge R-\tau k$.
Since the left-hand side is a weighted average of the values
$w_p^\top\Sigma w_p$, there must exist some component $C_\ell$ with $a_\ell>0$
such that $w_\ell^\top\Sigma w_\ell\ge R-\tau k$.
Set $v'=w_\ell=v_{C_\ell}/\|v_{C_\ell}\|_2$.
Then $v'$ is a unit vector, $\supp(v')=C_\ell\cap\supp(v)\subseteq\supp(v)$,
and $\supp(v')$ is $\tau$-connected by definition of $C_\ell$ as a connected
component of $G_\tau(S)$. 
\end{proof}

\begin{corollary}[Worst-case fixed-parameter guarantee via the constrained power method]
\label{cor:worst-case-fpt-sparse-pca}
Let $\Sigma\in\R^{d\times d}$ be symmetric and $(r+1)^k$-PSD. Fix $\delta\in(0,1)$, $\beta>0$, $\eta>0$, $\tau>0$, and $r\in\Z_+$ satisfying $r\tau^2>\beta^2$. Assume that there exists a unit vector
$v\in\R^d$ such that $\|v\|_0\le k$ and
$
    v^\top\Sigma v\ge R,
$
where $R-\tau k\ge \beta+\eta$.
Then the procedure of \Cref{cor:bounded-diameter-constrained-power}, run with search-radius parameter $L=k$ and parameters $\beta,\tau,r,\eta$, returns a unit vector $u$ such that with probability $1-\delta$
$$
    u^\top\Sigma u\geq \beta.
$$
Moreover, $\|u\|_0\le \max\{r,(r+1)^k\}$, and the running time is
$O\!\left(
d^2+dr^2+d(r+1)^{O(k)}T_{\rm pm}
\right)$,
where $T_{\mathrm{pm}}$ denotes the runtime of the power-method subroutine.
\end{corollary}

\begin{proof}
Let $S=\supp(v)$. By \Cref{lem:component-reduction}, there exists a unit vector $v'$ such that $\supp(v')\subseteq S$, $\supp(v')$ is $\tau$-connected, and
$$
    (v')^\top\Sigma v'
    \ge
    R-\tau k
    \ge
    \beta+\eta.
$$
Let $C=\supp(v')$. Since $C$ is $\tau$-connected and $|C|\le k$, the graph $G_\tau(C)$ has radius at most $k-1$, and hence at most $k$.

Therefore $v'$ satisfies the assumptions of \Cref{cor:bounded-diameter-constrained-power} with witness set $E=C$ and search-radius parameter $L=k$. Applying \Cref{cor:bounded-diameter-constrained-power} returns a unit vector $u$ satisfying
$
    u^\top\Sigma u\geq \beta.
$
The sparsity and running-time guarantees follow from \Cref{cor:bounded-diameter-constrained-power}.
\end{proof}

\subsection{Finding a small-radius witness}
\label{sec:small-radius}
In this section, we will prove \Cref{thm:informal-quadratic-sdp} by improving \Cref{cor:worst-case-fpt-sparse-pca}.
As proved in \Cref{thm:bounded-diameter-branching} and \Cref{lem:component-reduction}, the dependence on \(k\) in the branching search is controlled by the radius of the connected component that must be explored.

This suggests a more refined structural viewpoint: if one can trim
the original sparse witness to a smaller $\tau$-connected witness of radius
$O(\log k)$ while preserving a large quadratic form, then the branching
search needs only $O(\log k)$ rounds.  The resulting candidate supports have
size $(r+1)^{O(\log k)}$, giving a quasi-polynomial dependence on $k$ when
$r=k^{O(1)}$, rather than the current $k^{O(k)}$ bound.

Thus, the main question for improving the running time is the following:
For a $k$-sparse unit vector $v$, how large must $v^\top\Sigma v$
be to guarantee a sparse unit vector $u$ with quadratic form larger than
$\beta$ whose support has radius $O(\log k)$ in the large-correlation graph?
The next lemma shows that a sufficiently large sparse quadratic form contains a low-radius certificate after losing only a factor of $k^{1/(2L)}$ in value. This is the main structural input behind the improved running-time guarantees. Indeed, if the original vector has value $R$, then for any radius parameter $L$, we can find a subset $T$ contained in the original support whose large-correlation graph has radius at most $L$ and whose top eigenvalue is at least roughly $R/k^{1/(2L)}$, up to the error introduced by thresholding small correlations.

The intuition is to first discard all off-diagonal entries of $\Sigma_{S,S}$ with magnitude smaller than $\tau$. This changes every quadratic form on $S$ by at most $\delta=\tau k$, so the thresholded matrix $H$ still has a large top eigenvalue. Since $H$ is supported on the large-correlation graph, the vector $H^L e_i$ is supported only on the radius-$L$ ball around $i$. If every radius-$L$ ball had small spectral norm, then all columns of $H^L$ would be small, forcing $\|H^L\|_{\mathrm{op}}$ to be small, contradicting the large eigenvalue of $H$. Therefore some radius-$L$ ball must already contain a large quadratic form. Passing from $H$ back to $\Sigma$ loses another additive $\delta$, giving the claimed bound.
\begin{lemma}[Quantitative bounded-radius reduction]
\label{lem:quantitative-bounded-radius-reduction}
Let $\Sigma\in\mathbb{R}^{d\times d}$ be symmetric, and let $v$ be a unit vector with
$\supp(v)=S$ and $|S|=k$. Assume that $\Sigma_{S,S}\succeq 0$. Suppose
$
    v^\top \Sigma v \ge R .
$
Fix $\tau>0$, and let $G_\tau(S)$ be the graph on $S$ with an edge
$\{i,j\}$, $i\neq j$, whenever
$
    |\Sigma_{ij}|\ge \tau .
$
Let $L\ge 1$ be an integer and set
$
    \delta := \tau k .
$
Then there exists a set $T\subseteq S$ such that $G_\tau(T)$ 
has radius at most $L$, and
\[
    \lambda_{\max}(\Sigma_{T,T})
    \ge
    \frac{R-\delta}{k^{1/(2L)}}-\delta .
\]

\end{lemma}

\begin{proof}[Proof of \Cref{lem:quantitative-bounded-radius-reduction}]
We work on the coordinate set $S$. Define the thresholded matrix $H\in\mathbb{R}^{S\times S}$ by
\[
    H_{ij}
    =
    \begin{cases}
        \Sigma_{ij}, & i=j,\\
        \Sigma_{ij}, & i\neq j \text{ and } |\Sigma_{ij}|\ge \tau,\\
        0, & i\neq j \text{ and } |\Sigma_{ij}|<\tau.
    \end{cases}
\]
Thus $H$ is supported on the graph $G_\tau(S)$, together with the diagonal.

Let $E:=\Sigma_{S,S}-H$.
Every nonzero off-diagonal entry of $E$ has magnitude strictly smaller than $\tau$,
and $E$ has zero diagonal. Therefore, by the row-sum bound,
\[
    \|E\|_{\mathrm{op}}
    \le
    \tau k
    =
    \delta .
\]
Equivalently, every vector $x$ supported on $S$ satisfies
$|x^\top E x|\le \delta \|x\|_2^2$.
Since $v$ is supported on $S$ and $\|v\|_2=1$, we get
\[
    v^\top H v
    =
    v^\top \Sigma v - v^\top E v
    \ge
    R-\delta .
\]
Hence $\lambda_{\max}(H)\ge R-\delta$.
If $R\le \delta$, then the claimed lower bound is nonpositive and is trivial since
$T\subseteq S$ implies $\Sigma_{T,T}\succeq 0$ by the assumption $\Sigma_{S,S}\succeq 0$. Thus assume $R>\delta$, and set
$\Lambda:=\lambda_{\max}(H)$.
Then $\Lambda\ge R-\delta>0$.

Because $H$ is supported on $G_\tau(S)$, the vector $H^L e_i$ is supported inside
the radius-$L$ ball around $i$ in $G_\tau(S)$. Denote this ball by
$B(i,L):=\{j\in S:\operatorname{dist}_{G_\tau(S)}(i,j)\le L\}$.
We now show that one such ball has large operator norm. Since $H$ is symmetric,
$\|H^L\|_{\mathrm{op}} \ge \Lambda^L$.
On the other hand,
\[
    \|H^L\|_{\mathrm{op}}
    \le
    \|H^L\|_F
    =
    \left(\sum_{i\in S}\|H^L e_i\|_2^2\right)^{1/2}
    \le
    \sqrt{k}\max_{i\in S}\|H^L e_i\|_2 .
\]
Therefore there exists some $i\in S$ such that
$\|H^L e_i\|_2\ge \Lambda^L/\sqrt{k}$.
Set $T:=B(i,L)$.
Then $G_\tau(T)$ is connected and has radius at most $L$. Moreover, every walk of
length at most $L$ starting at $i$ stays inside $T$, so
$H^L e_i = H_{T,T}^L e_i$.
Consequently,
\[
    \|H_{T,T}^L\|_{\mathrm{op}}
    \ge
    \|H_{T,T}^L e_i\|_2
    =
    \|H^L e_i\|_2
    \ge
    \frac{\Lambda^L}{\sqrt{k}} .
\]
Since 
$
    \|H_{T,T}^L\|_{\mathrm{op}}
    \le
    \|H_{T,T}\|_{\mathrm{op}}^L
$, we have
\[
    \|H_{T,T}\|_{\mathrm{op}}
    \ge
    \frac{\Lambda}{k^{1/(2L)}}
    \ge
    \frac{R-\delta}{k^{1/(2L)}} .
\]
It remains to pass from $H_{T,T}$ back to $\Sigma_{T,T}$. Since
$\Sigma_{T,T}=H_{T,T}+E_{T,T}$ and
$\|E_{T,T}\|_{\mathrm{op}}\le\delta$, we have
\[
    \lambda_{\max}(\Sigma_{T,T})
    \ge
    \lambda_{\max}(H_{T,T})-\delta .
\]
It remains to show that
$\lambda_{\max}(H_{T,T})=\|H_{T,T}\|_{\mathrm{op}}$ in the nontrivial
regime.
If $(R-\delta)/k^{1/(2L)}\le \delta$, then the claimed lower bound is
at most zero, and the result is again trivial because $T\subseteq S$ implies
$\Sigma_{T,T}\succeq 0$. Otherwise, $\|H_{T,T}\|_{\mathrm{op}}>\delta$.
Since $H_{T,T}=\Sigma_{T,T}-E_{T,T}$, for every unit vector $x$ supported on $T$,
the assumption $\Sigma_{S,S}\succeq 0$ and the inclusion $T\subseteq S$ give $x^\top\Sigma_{T,T}x\ge 0$, and hence
\[
    x^\top H_{T,T}x
    = x^\top\Sigma_{T,T}x - x^\top E_{T,T}x
    \ge
    -\delta .
\]
Thus every negative eigenvalue of $H_{T,T}$ has magnitude at most $\delta$.
Therefore the operator norm of $H_{T,T}$ must be attained by a positive eigenvalue, and hence
$\lambda_{\max}(H_{T,T})=\|H_{T,T}\|_{\mathrm{op}}$.
Combining the above inequalities gives
\[
    \lambda_{\max}(\Sigma_{T,T})
    \ge
    \frac{R-\delta}{k^{1/(2L)}}-\delta \;,
\]
which concludes the proof of \Cref{lem:quantitative-bounded-radius-reduction}.
\end{proof}
The bounded-radius reduction and the radius-search procedures from the previous subsection yield the general runtime--target-value tradeoff. If the radius-reduction lower bound exceeds the certificate level $\beta$ by the approximation slack, then the branching search detects the witness using candidate supports of size $(r+1)^{O(L)}$.

\Cref{fig:high-correlation-localization} summarizes how the preceding
structural lemmas are instantiated by the branching algorithm.
\Cref{lem:component-reduction} first replaces the original witness $v$ by a
unit vector $u$ supported on a $\tau$-connected component, while losing at
most $\tau k$ in quadratic form.  Applying
\Cref{lem:quantitative-bounded-radius-reduction} to $u$ and taking $u_L$ to be
a unit leading eigenvector on the resulting radius-$L$ set produces a low radius witness.  Finally,
\Cref{thm:bounded-diameter-branching} starts from every coordinate and grows
the sets $B^{(t)}(i)$, so the branch initialized at a radius center
$i^\star\in\operatorname{supp}(u_L)$ contains
$\operatorname{supp}(u_L)$ after $L$ rounds.

The bounded-degree alternative is essential for controlling the search.  If
a coordinate $i$ has more than $r$ neighbors $j$ satisfying
$|\Sigma_{ij}|\ge\tau$, then, after the basis-vector check, at least one such
neighbor is omitted by $\topp_r(\Sigma e_i)$ and
\Cref{lem:discover} yields the simpler $r$-sparse vector with large quadratic
form.  Otherwise every explored coordinate contributes at most $r$ new
neighbors, and hence
$
    |B^{(t)}(i)|\le (r+1)^t,
$
which is precisely the size bound used by the branching algorithm.

\begin{figure}[t]
\centering
\begin{tikzpicture}[
  x=.78cm,y=.78cm,scale=1.15,
  every node/.style={font=\scriptsize},
  edge/.style={draw=black!32,line width=.55pt},
  vnode/.style={circle,draw=blue!65!black,fill=white,
                minimum size=4mm,inner sep=0pt},
  faded node/.style={circle,draw=black!32,fill=white,
                     minimum size=4mm,inner sep=0pt},
  u edge/.style={draw=orange!85!black,line width=1.35pt},
  unode/.style={circle,draw=orange!85!black,fill=orange!18,
                minimum size=4mm,inner sep=0pt},
  ul edge/.style={draw=green!50!black,line width=1.8pt},
  ulnode/.style={circle,draw=green!50!black,fill=green!15,
                 minimum size=4mm,inner sep=0pt},
  ball/.style={draw=blue!55!black,dashed,line width=.65pt},
  ball label/.style={fill=white,inner sep=1pt,text=blue!55!black},
  panel label/.style={font=\small\bfseries,anchor=west}
]

\coordinate (a1) at (.65,2.85);
\coordinate (a2) at (1.65,3.70);
\coordinate (a3) at (2.85,3.05);
\coordinate (a4) at (3.95,2.05);
\coordinate (a5) at (5.10,2.72);
\coordinate (a6) at (6.22,1.62);
\coordinate (a7) at (7.35,2.48);
\coordinate (a8) at (8.27,1.08);
\coordinate (ab1) at (1.92,1.55);
\coordinate (ab2) at (2.85,4.55);
\coordinate (ab3) at (5.32,4.38);
\coordinate (ab4) at (7.10,.62);
\coordinate (ai1) at (.66,.72);
\coordinate (ai2) at (4.36,4.72);
\coordinate (ai3) at (8.20,4.25);

\coordinate (b1) at (10.65,2.85);
\coordinate (b2) at (11.65,3.70);
\coordinate (b3) at (12.85,3.05);
\coordinate (b4) at (13.95,2.05);
\coordinate (b5) at (15.10,2.72);
\coordinate (b6) at (16.22,1.62);
\coordinate (b7) at (17.35,2.48);
\coordinate (b8) at (18.27,1.08);
\coordinate (bb1) at (11.92,1.55);
\coordinate (bb2) at (12.85,4.55);
\coordinate (bb3) at (15.32,4.38);
\coordinate (bb4) at (17.10,.62);
\coordinate (bi1) at (10.66,.72);
\coordinate (bi2) at (14.36,4.72);
\coordinate (bi3) at (18.20,4.25);

\foreach \x/\y in {a1/a2,a2/a3,a3/a4,a4/a5,a5/a6,a6/a7,a7/a8,
                    a2/ab1,a3/ab2,a5/ab3,a6/ab4}
  \draw[edge] (\x)--(\y);
\foreach \x in {a1,a2,a3,a4,a5,a6,a7,a8,ab1,ab2,ab3,ab4,ai1,ai2,ai3}
  \node[vnode] at (\x) {};

\foreach \x/\y in {a1/a2,a2/a3,a3/a4,a4/a5,a5/a6,a6/a7,a7/a8}
  \draw[u edge] (\x)--(\y);
\foreach \x in {a1,a2,a3,a4,a5,a6,a7,a8}
  \node[unode] at (\x) {};

\foreach \x/\y in {a3/a4,a4/a5,a5/a6}
  \draw[ul edge] (\x)--(\y);
\foreach \x in {a3,a4,a5,a6}
  \node[ulnode] at (\x) {};

\draw[black!14] (9.45,.28)--(9.45,5.35);

\foreach \x/\y in {b1/b2,b2/b3,b3/b4,b4/b5,b5/b6,b6/b7,b7/b8,
                    b2/bb1,b3/bb2,b5/bb3,b6/bb4}
  \draw[edge] (\x)--(\y);
\foreach \x in {b1,b2,b3,b4,b5,b6,b7,b8,bb1,bb2,bb3,bb4,bi1,bi2,bi3}
  \node[faded node] at (\x) {};

\draw[ball] (b4) circle (.42);
\node[ball label,anchor=south east] at ($(b4)+(.63,.48)$) {$B^{(0)}(i^{\star})$};
\draw[ball] ($(b4)!.50!(b5)$) ellipse (1.45 and 1.02);
\node[ball label,anchor=south] at ($(b4)!.50!(b5)+(0,1.02)$) {$B^{(1)}(i^{\star})$};
\draw[ball] ($(b4)!.48!(b5)$) ellipse (2.58 and 1.70);
\node[ball label,anchor=south] at ($(b4)!.48!(b5)+(0,1.70)$)
  {$B^{(L)}(i^{\star})$};

\foreach \x/\y in {b3/b4,b4/b5,b5/b6}
  \draw[ul edge] (\x)--(\y);
\foreach \x in {b3,b4,b5,b6}
  \node[ulnode] at (\x) {};
\node[ulnode] at (b4) {$i^{\star}$};

\node[anchor=south,font=\scriptsize] at (14.55,-0.18)
  {$B^{(0)}(i^{\star})\subseteq B^{(1)}(i^{\star})\subseteq\cdots
    \subseteq B^{(L)}(i^{\star})$};

\end{tikzpicture}\caption{The high-correlation graph $G_\tau(v)$ has vertex set
$\operatorname{supp}(v)$ and contains an edge $ij$ whenever
$|\Sigma_{ij}|\ge\tau$. In the left figure, $\operatorname{supp}(v)$ consists
of all the vertices, $\operatorname{supp}(u)$ consists of the orange and green
vertices, and $\operatorname{supp}(u_L)$ consists only of the green vertices.
The vector $u$ is supported on the high-correlation part of the graph and
satisfies
$u^{\top}\Sigma u\ge v^{\top}\Sigma v-\tau k$; the vector $u_L$ is
supported on a subset of $\operatorname{supp}(u)$ and satisfies
$
u_L^{\top}\Sigma u_L
\ge {u^{\top}\Sigma u}/{k^{1/(2L)}}.
$
Moreover, its support has graph radius at most $L$. In the right figure,
starting from each coordinate $s$, the algorithm constructs the nested sets
$B^{(i)}(s)$. For an appropriate coordinate $i^{\star}\in u_L$, the set
$B^{(L)}(i^{\star})$ contains $\operatorname{supp}(u_L)$.}
\label{fig:high-correlation-localization}
\end{figure}
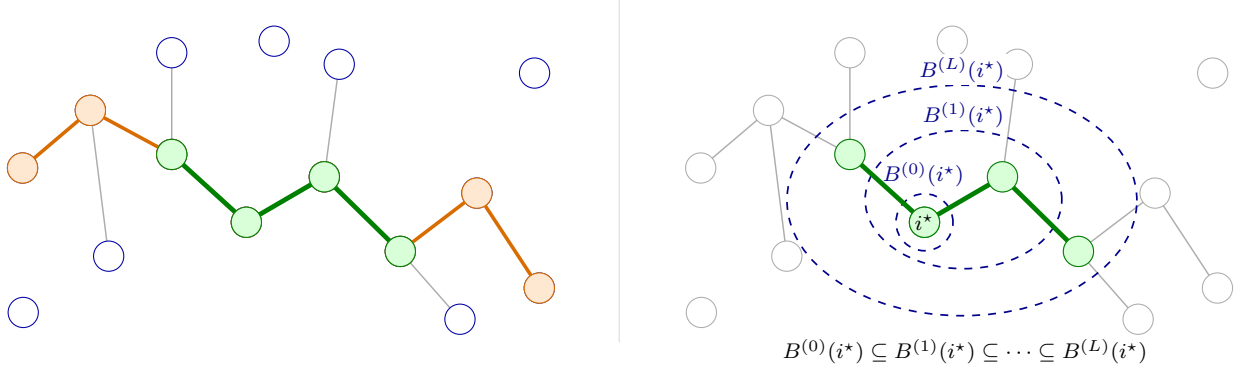

We state the resulting guarantee simultaneously for the constrained power method and for the semidefinite relaxation; the latter proves \Cref{thm:informal-quadratic-sdp}.

\begin{corollary}[Sparse PCA tradeoff]
\label{cor:parameterized-radius-sparse-pca}
\label{cor:quasipolynomial-sparse-pca}
\label{cor:polynomial-root-k-sparse-pca}
Fix $\delta\in(0,1)$, a certificate threshold $\beta>0$, an approximation slack $\eta>0$, a correlation threshold $\tau>0$, and $r\in\Z_+$ satisfying $r\tau^2>\beta^2$. Let $L\ge 1$ be an integer. Let $\Sigma\in\R^{d\times d}$ be symmetric. Suppose there exists a unit vector $v\in\R^d$ such that $\|v\|_0\le k$ and
\[
    v^\top\Sigma v\ge R,
\]
where
\[
    \frac{R-\tau k}{k^{1/(2L)}}-\tau k
    \ge
    \beta+\eta.
\]
Run \Cref{alg:branching-truncated-power} with search-radius parameter $L$ and parameters $\beta,\tau,r$. Then:
\begin{enumerate}[(i)]
    \item If $\Sigma$ is $\max\{k,(r+1)^L\}$-PSD and the candidate supports are processed using the constrained power method as in \Cref{cor:bounded-diameter-constrained-power} with slack $\eta$, then with probability at least $1-\delta$ the procedure returns a unit vector $u$ with $\|u\|_0\le \max\{r,(r+1)^{O(L)}\}$  satisfying
    $
        u^\top\Sigma u\ge \beta.
    $
    Moreover, the running time is
$O\!\left(
d^2+dr^2+d(r+1)^{O(L)}T_{\rm pm}
\right)$,  
    where $T_{\mathrm{pm}}$ denotes the runtime of the power-method subroutine.

    \item If $\Sigma$ is $\max\{k,r+1\}$-PSD and the candidate supports are processed using the semidefinite relaxation as in \Cref{cor:bounded-diameter-sdp} with additive error $\eta$, then the procedure returns either a unit vector $u$ such that $\|u\|_0\le r$ and $u^\top\Sigma u>\beta$, or a matrix $X\in\cX_{\sdp,k}$ such that
    $
        \langle\Sigma,X\rangle\ge \beta.
    $
    Using a polynomial-time SDP solver, the running time is $O\!\left(d^2+d r^{2}+d(r+1)^{O(L)}/\eta^{O(1)}\right)$.
\end{enumerate}
\end{corollary}

\begin{proof}
For either conclusion, the corresponding PSD assumption implies $\Sigma_{\supp(v),\supp(v)}\succeq0$, so \Cref{lem:quantitative-bounded-radius-reduction} applies. Apply \Cref{lem:quantitative-bounded-radius-reduction} with the supplied threshold $\tau$ and radius parameter $L$. Then the additive thresholding loss is $\delta_0=\tau k$, and there exists a set $T\subseteq\supp(v)$ whose large-correlation graph has radius at most $L$ and such that
\[
    \lambda_{\max}(\Sigma_{T,T})
    \ge
    \frac{R-\tau k}{k^{1/(2L)}}-\tau k
    \ge
    \beta+\eta.
\]
Let $w$ be a unit vector supported on $T$ attaining $\lambda_{\max}(\Sigma_{T,T})$. Then $w^\top\Sigma w\ge\beta+\eta$, and $T$ is a set containing $\supp(w)$ with $\tau$-correlation radius at most $L$. Therefore, \Cref{thm:bounded-diameter-branching}, applied with target vector $w$ and witness set $E=T$, shows that the branching procedure either returns an $r$-sparse vector with quadratic form larger than $\beta$, or produces a candidate support $B$ containing $T$.

In the latter case,
\[
    \lambda_{\max}(\Sigma_{B,B})
    \ge
    \lambda_{\max}(\Sigma_{T,T})
    \ge
    \beta+\eta.
\]
The constrained-power conclusion now follows from \Cref{cor:bounded-diameter-constrained-power}. For the SDP conclusion, since $w$ is supported on $T$ and $|T|\le k$, the matrix $ww^\top$ is feasible for the $k$-budget SDP on $B$, because
\[
    ww^\top\succeq0,
    \qquad
    \operatorname{tr}(ww^\top)=1,
    \qquad
    \|ww^\top\|_1=\|w\|_1^2\le k.
\]
Hence the SDP optimum is at least $\beta+\eta$. Therefore, the additive-$\eta$ approximate solution has objective value at least $\beta$. The sparsity and running-time bounds follow from the corresponding radius-$L$ bounds in \Cref{cor:bounded-diameter-constrained-power,cor:bounded-diameter-sdp}.
\end{proof}

This corollary captures the full runtime--target-value tradeoff, up to constant factors in $L$: the condition
\[
    R
    \ge
    \tau k+(\beta+\eta+\tau k)k^{1/(2L)}
\]
suffices to obtain a certificate of value $\beta$ in time
$O\!\left(
d^2+dr^{O(1)}+d(r+1)^{O(L)}
\right)$.  Now assume that $r=\poly(k)$. 
For constant $L$, this gives polynomial dependence on k for root-k target values. Taking $L=C\log k$ for a sufficiently large constant $C$ makes $k^{1/L}$ arbitrarily close to $1$, and therefore recovers the quasi-polynomial $k^{O(\log k)}$ guarantee for a constant approximation to the target quadratic form.

\subsection{Faster algorithm through Correlation Detection}\label{sec:sub-quadratic}
In this section, we consider the case in which $\Sigma$ is the empirical covariance of the dataset $x_1,\dots,x_n$
and show how to run our algorithm in subquadratic time in the dimension, thereby proving \Cref{thm:informal-subquadratic-sdp}.
Since each entry of $\Sigma$ can be computed in $O(n)$ time,
we assume oracle access to the entries of $\Sigma$, incurring a multiplicative $\poly(n)$ factor in the running time.

At a high level, the proposed  algorithm can be divided into two procedures. The first constructs a sparse graph with at most $d r$ edges; if the graph has more than this many edges, then we are done by \Cref{lem:discover}. The second computes the candidate supports from the resulting graph and runs in time linear in the size of that graph. Thus we aim to speed up the first part of the algorithm.

The intuition is as follows.
We first perform a fast randomized density check.  If the graph is
dense, the check immediately returns a sparse vector with large quadratic
form.  Otherwise, we use correlation detection (\Cref{fact:cordetect}) to
construct the graph without computing the full covariance matrix and then run
the branching algorithm on this sparse graph.

First, in the following lemma we prove that we can reduce to the sparse case by randomly querying  the matrix entries.

\begin{algorithm}[h]
\centering
\fbox{\parbox{6in}{
{\bf Input:}
Entry-query access to a symmetric matrix $\Sigma\in\R^{d\times d}$,
a threshold $\tau>0$, a degree parameter $r\in\Z_+$, output quadratic form threshold $\beta>0$, an edge budget $B$,
and failure probability $\delta\in(0,1)$.\\
{\bf Output:}
Either a unit vector $u\in\R^d$ with $u^\top\Sigma u>\beta$, or the declaration
``sparse case.''

    \begin{enumerate}[leftmargin=*]
        \item\label{line:dense-set-samples}
        Set $M \gets C(d^2/B)\log(1/\delta)$ and
        $L \gets C'\log(1/\delta)$, for sufficiently large universal
        constants $C,C'>0$.

        \item\label{line:dense-sample-pairs}
        Sample $M$ unordered pairs $\{i,j\}$ uniformly at random from
        $\binom{[d]}{2}$, with replacement.

        \item\label{line:dense-init-counter}
        Initialize a counter $q\gets 0$.

        \item\label{line:dense-process-samples}
        For each sampled pair $\{i,j\}$:
        \begin{enumerate}[leftmargin=*, nosep]
            \item Query $\Sigma_{ij}$.

            \item If $|\Sigma_{ij}|<\tau$, continue to the next sampled pair.

            \item Otherwise, increment $q\gets q+1$.

            \item For each endpoint $a\in\{i,j\}$, scan row $a$ and compute
            $\Gamma_\tau(a)\gets\{b\neq a:|\Sigma_{ab}|\ge \tau\}$.

            \item If, during one of these scans, a coordinate $a$ satisfies
            $\Sigma_{aa}>\beta$, return $u=e_a$.

            \item If, during one of these scans, a coordinate $a$ satisfies
            $|\Gamma_\tau(a)|\ge r$, choose any subset
            $S\subseteq\Gamma_\tau(a)$ with $|S|=r$. Define $z\in\R^d$ by
            \[
                z_b
                \gets
                \begin{cases}
                \Sigma_{ab}, & b\in S,\\
                0, & b\notin S.
                \end{cases}
            \]
            Return $u\gets z/\|z\|_2$.

            \item If $q\ge L$, terminate the loop.
        \end{enumerate}

        \item\label{line:dense-return-sparse}
        Return ``sparse case.''
    \end{enumerate}
}}
\caption{Dense-Case Random Row Search.}
\label{alg:dense-case-random-row-search}
\end{algorithm}

\begin{lemma}[Density Check]\label{lem:densesearch}
Let $\Sigma\in\R^{d\times d}$ be symmetric and $(r+1)$-PSD.
Fix a threshold $\tau > 0$, an output quadratic form threshold $\beta>0$, and define the large-correlation edge set
$
E_\tau
=
\bigl\{\{i,j\}: i\neq j,\ |\Sigma_{ij}|\ge \tau \bigr\}.
$
Let
$
r\tau^2>\beta^2
$
and let $B$ be an edge threshold satisfying
$
B \ge 2dr.
$
There is a randomized algorithm running in time
\[
O\left(\frac{d^2}{B}\log\frac{1}{\delta}
+
d\log\frac{1}{\delta}\right)
\]
such that with probability at least $1-\delta$: either
$
|E_\tau| \ge B,
$
and the algorithm outputs an
$r$-sparse unit vector $u$ such that
$
u^\top \Sigma u > \beta,
$
or if the algorithm fails to output such a vector 
$
|E_\tau| < B.
$
\end{lemma}

\begin{proof} 

We will prove that \Cref{alg:dense-case-random-row-search} satisfies the  guarantees of the statement.
We first show that any row with at least $r$ entries of magnitude at
least $\tau$ yields a vector with large quadratic form.

\begin{claim}[A heavy row gives a sparse certificate]\label{cl:disc2}
Let $\Sigma\in\R^{d\times d}$ be symmetric and $(r+1)$-PSD, let $\tau>0$, let $\beta>0$, and for $a\in[d]$ define
$\Gamma_\tau(a)\eqdef\{b\neq a:|\Sigma_{ab}|\ge \tau\}$.
Suppose $|\Gamma_\tau(a)|\ge r$, where $r\tau^2>\beta^2$.
Then one can construct, in $O(d)$ time, an $r$-sparse unit vector $u$ such that
$u^\top\Sigma u>\beta$.
\end{claim}

\begin{proof}
If $\Sigma_{aa}>\beta$, then $e_a^\top\Sigma e_a>\beta$, so we may output
$u=e_a$. Thus assume $\Sigma_{aa}\le \beta$.
Choose any subset $S\subseteq\Gamma_\tau(a)$ with $|S|=r$, and define
$u=\Sigma_{S,a}/\|\Sigma_{S,a}\|_2$, viewed as a vector supported on $S$.

Since $u$ is supported on $S$ and $|\{a\}\cup S|\le r+1$, the $(r+1)$-PSD assumption gives $\Sigma_{\{a\}\cup S,\{a\}\cup S}\succeq0$.  Thus Cauchy--Schwarz in this local PSD seminorm gives
$(u^\top \Sigma e_a)^2\le (u^\top \Sigma u)(e_a^\top \Sigma e_a)$.
Therefore, $u^\top\Sigma u\ge (u^\top\Sigma e_a)^2/\Sigma_{aa}$.
By the definition of $u$, we have $u^\top\Sigma e_a=\|\Sigma_{S,a}\|_2$.
Since every $b\in S$ satisfies $|\Sigma_{ab}|\ge\tau$, we have
$\|\Sigma_{S,a}\|_2^2\ge r\tau^2$. Hence
\[
    u^\top\Sigma u
    \ge
    \frac{r\tau^2}{\Sigma_{aa}}
    \ge
    \frac{r\tau^2}{\beta}
    >\beta.
\]
This concludes the proof of \Cref{cl:disc2}.
\end{proof}
Thus scanning any row with at least $r$ large-correlation entries lets us
output a valid sparse certificate.

It remains to show that, when $|E_\tau|\ge B$, random sampling finds such a
row with high probability. Let
$
H
=
\{a\in[d]: |\Gamma_\tau(a)|\ge r\}
$
be the set of heavy vertices. Call an edge useful if it is incident to a
vertex in $H$.

The number of edges with no endpoint in $H$ is at most
$
\frac{dr}{2},
$
because every non-heavy vertex has degree strictly less than $r$. Therefore,
if $|E_\tau|\ge B$ and $B\ge 2dr$, then the number of useful edges is at least
\[
|E_\tau|-\frac{dr}{2}
\ge
|E_\tau|-\frac{B}{4}
\ge
\frac34 |E_\tau|.
\]
Consequently, conditioned on a uniformly sampled pair being an edge of
$E_\tau$, it is useful with probability at least $3/4$.

Also, a uniformly sampled unordered pair belongs to $E_\tau$ with probability
$
\frac{|E_\tau|}{d^2}
\ge
\frac{B}{d^2}.
$
Thus, after
$
M=C\frac{d^2}{B}\log\frac{1}{\delta}
$
samples, the number of sampled large-correlation edges is at least
$
L=C'\log\frac{1}{\delta}
$
with probability at least $1-\delta/2$, provided $C$ is sufficiently large
relative to $C'$.

Condition on this event. Among the first $L$ sampled large-correlation edges,
each is useful with conditional probability at least $3/4$. Therefore the
probability that none of them is useful is at most
$
\left(\frac14\right)^L
\le
\frac{\delta}{2},
$
for $C'$ sufficiently large. Hence, with probability at least $1-\delta$,
the algorithm samples a useful edge. Scanning the endpoints of a useful edge
finds a heavy vertex $a\in H$, and by \Cref{cl:disc2}   the
algorithm outputs an $r$-sparse unit vector $u$ satisfying
$
u^\top \Sigma u > \beta.
$

Finally, we bound the running time. The algorithm makes
$
O\left(\frac{d^2}{B}\log\frac{1}{\delta}\right)
$
entry queries during random sampling. It scans at most
$O(\log(1/\delta))$ rows, and each row scan costs $O(d)$ entry queries.
Thus the total runtime is
$O\left(\frac{d^2}{B}\log\frac{1}{\delta}+d\log\frac{1}{\delta}\right)$.
This proves the lemma.
\end{proof}

Therefore, it suffices to handle the sparse case efficiently. 
The next lemma analyzes the case in which the large-correlation
graph is given explicitly.  The graph search either returns a sparse vector
with large quadratic form or a family of candidate supports, one of which
contains the target support.

\begin{algorithm}[h]
\centering
\fbox{\parbox{6in}{
{\bf Input:}
A symmetric matrix $\Sigma\in\R^{d\times d}$, a graph
$G=([d],E)$ whose edges are pairs ${i,j}$ with $|\Sigma_{ij}|\ge \alpha$,
a threshold $\alpha>0$, a certificate threshold $\beta>0$, a degree parameter $r$, and a
search-radius parameter $L$.\\
{\bf Output:}
Either an $r$-sparse unit vector $u\in\R^d$ with $u^\top\Sigma u>\beta$, or a family $\mathcal B$ of candidate supports.

\begin{enumerate}[leftmargin=*]
    \item For every coordinate $i\in[d]$, let $\Gamma(i)$ denote the
    neighborhood of $i$ in $G$.

    \item If there exists $a\in[d]$ with $|\Gamma(a)|\ge r$, choose any
    subset $S\subseteq\Gamma(a)$ with $|S|=r$. If $\Sigma_{aa}>\beta$,
    return $u=e_a$. Otherwise define $z\in\R^d$ by
    \[
        z_b
        \gets
        \begin{cases}
        \Sigma_{ab}, & b\in S,\\
        0, & b\notin S,
        \end{cases}
    \]
    and return $u\gets z/\|z\|_2$.

    \item Initialize an empty family of candidate supports
    $\mathcal B\gets\emptyset$.

    \item For every starting coordinate $i_0\in[d]$:
    \begin{enumerate}[leftmargin=*, nosep]
        \item Set $B^{(0)}_{i_0}\gets\{i_0\}$ and add
        $B^{(0)}_{i_0}$ to $\mathcal B$.

        \item For $t=0,1,\ldots,L-1$, define
        \[
            B^{(t+1)}_{i_0}
            \gets
            B^{(t)}_{i_0}
            \cup
            \bigcup_{i\in B^{(t)}_{i_0}}\Gamma(i).
        \]
        Add $B^{(t+1)}_{i_0}$ to $\mathcal B$.
    \end{enumerate}

    \item Return the family $\mathcal B$.
\end{enumerate}

}}
\caption{Graph Branching Search for Sparse PCA.}
\label{alg:graph-branching-search}
\end{algorithm}

\begin{lemma}[Sparse graph search]\label{lem:sparse-search}
Let $\Sigma\in\R^{d\times d}$ be symmetric and $(r+1)$-PSD.
Fix a strong correlation threshold $\rho>0$, a certificate threshold $\beta>0$, and a degree parameter $r\in\Z_+$ satisfying $r\rho^2>\beta^2$. Suppose we are given the edge
list
$
    E_\rho=\bigl\{\{i,j\}:i\neq j,\ |\Sigma_{ij}|\ge \rho\bigr\}
$
and assume $|E_\rho|\le B$. Let $G_\rho=([d],E_\rho)$.

Assume there exist a unit vector $v\in\R^d$ and a set $E\subseteq[d]$ such that $\|v\|_0\le k$,
$v^\top\Sigma v\ge\beta$, $\supp(v)\subseteq E$, and $E$ has $\rho$-correlation radius at most $L$. Then \Cref{alg:graph-branching-search} returns one of the
following:
\begin{enumerate}[(i)]
    \item an $r$-sparse unit vector $u$ such that $u^\top\Sigma u>\beta$;
    \item a family $\mathcal B$ of at most $d(L+1)$ candidate supports,
    each of size at most $(r+1)^L$, such that
    $E\subseteq B$ for some $B\in\mathcal B$.
\end{enumerate}
The running time is $O(B+d r(r+1)^L)$.
\end{lemma}

\begin{proof}
Building the adjacency lists $\Gamma_\rho(i)$ of $G_\rho$ from the given edge
list costs $O(B)$ time. Next we check whether some row has large-correlation
degree at least $r$, namely $|\Gamma_\rho(a)|\ge r$. If this happens,
\Cref{cl:disc2} applied with threshold $\rho$ gives an $r$-sparse unit vector $u$ with
$u^\top\Sigma u>\beta$, so the algorithm terminates in $O(B+d)$ time. Thus, in
the remaining case, every vertex satisfies $|\Gamma_\rho(i)|<r$.

We now run the branching search using the given graph neighborhoods. For each
starting coordinate $i_0\in[d]$, set $B^{(0)}_{i_0}=\{i_0\}$ and for
$t=0,1,\ldots,L-1$ define
\[
    B^{(t+1)}_{i_0}
    =
    B^{(t)}_{i_0}
    \cup
    \bigcup_{i\in B^{(t)}_{i_0}}\Gamma_\rho(i).
\]
Since every vertex has degree less than $r$, we have
$|B^{(t+1)}_{i_0}|\le (r+1)|B^{(t)}_{i_0}|$, and therefore
$|B^{(t)}_{i_0}|\le (r+1)^t$. In particular every final support has size at
most $s\eqdef (r+1)^L$.

It remains to prove correctness. Let $C=E$. Let $i_\star\in C$ be a radius center of $C$, and consider the branch initialized at $i_\star$. We
claim by induction that $B^{(t)}_{i_\star}$ contains every vertex of $C$ at graph
distance at most $t$ from $i_\star$ inside $G_\rho(C)$. The base case $t=0$ is
immediate. For the induction step, suppose $j\in C$ has distance $t+1$ from
$i_\star$. Then there exists $i\in C$ at distance $t$ from $i_\star$ with
$\{i,j\}\in E_\rho$. By the induction hypothesis, $i\in B^{(t)}_{i_\star}$.
Since the algorithm has not terminated in the heavy-row case, all graph
neighbors of $i$ are included in $\Gamma_\rho(i)$, and in particular
$j\in\Gamma_\rho(i)$. Hence $j\in B^{(t+1)}_{i_\star}$. Because $i_\star$ is a radius-$L$ center of $C$ in $G_\rho(C)$, the above induction implies
$C\subseteq B^{(L)}_{i_\star}$.

The algorithm constructs $L+1$ supports for each of the $d$ starting
coordinates, so $|\mathcal B|\le d(L+1)$. The preceding size bound gives
$|B|\le (r+1)^L$ for every $B\in\mathcal B$, and the induction
shows that one of these supports contains $\supp(v)$.

Finally, we bound the running time. Building the graph costs $O(B)$. For each
of the $d$ starting coordinates, the graph search explores at most
$s=(r+1)^L$ vertices, each with degree less than $r$, so the total
branching cost is $O(dsr)$. Therefore the total running time is
$O(B+d r(r+1)^L)$.
\end{proof}

We use the following modification of fast correlation detection to
construct the  non-trivial correlation graph without computing every entry of the
empirical covariance matrix.
\begin{restatable}[Fast  covariance detection]{lemma}{rawcovariancedetection}
\label{fact:cordetect}
Let $T=\{z_1,\ldots,z_n\}\subseteq\R^d$ have empirical covariance
matrix $\Sigma$.  Let $\delta\in(0,1)$ and let $\rho,\tau>0$ satisfy
$\rho>12\tau$.  Define
$
    D\eqdef\max\left\{2\rho,\max_{i\in[d]}\Sigma_{ii}\right\}.
$
Suppose that there are at most $s$ off-diagonal coordinate pairs
$\{i,j\}$ satisfying $|\Sigma_{ij}|\ge\tau$.  Then there is a randomized
algorithm that, with probability at least $1-\delta$, outputs the pairs
satisfying $|\Sigma_{ij}|\ge\rho$.  Its running time is
\[
\left(
    s d^{0.62}
    +
    d^{1.62+2.4\frac{\log(4D/\rho)}{\log(D/(3\tau))}}
\right)
\operatorname{poly}\left(n,\log d,\frac{D}{\tau}\right)
\log(1/\delta).
\]
\end{restatable}
This result is a modification of the correlation-detection result
of \cite{Valiant15}, using the real-to-Boolean reduction recorded by
\cite{Pensia24-subquad}.  Those results threshold according to the normalized correlation
$|\Sigma_{ij}|/\sqrt{\Sigma_{ii}\Sigma_{jj}}$, whereas the graph used in this work thresholds according to the
raw covariance entry $|\Sigma_{ij}|$.  We prove the stated raw-covariance version
in \Cref{app:raw-covariance-detection}.

Finally we combine those ingredients in the following \Cref{thm:subquadratic-correlation-detection-pca} and \Cref{alg:cd-branching-search}.

\begin{algorithm}[h]
\centering
\fbox{\parbox{6in}{
{\bf Input:}
A data set $T\subseteq\R^d$ with empirical covariance matrix $\Sigma$,
entry-query access to $\Sigma$, sparsity parameter $k$,
search-radius parameter $L$, certificate threshold $\beta>0$, strong correlation threshold $\rho>0$, margin threshold $\tau>0$, and failure probability $\delta$.\\
{\bf Output:}
Either a sparse unit vector $u\in\R^d$ with $u^\top\Sigma u>\beta$, or a family $\mathcal B$ of candidate supports.

\begin{enumerate}[leftmargin=*]
    \item Query all diagonal entries of $\Sigma$.  If
    $\Sigma_{ii}>\beta$ for some $i\in[d]$, return $e_i$.  Otherwise, set
    $D\gets\max\left\{2\rho,\max_{i\in[d]}\Sigma_{ii}\right\}$.

    \item Set
    $r_{\mathrm{dense}}\gets \lceil 4\beta^2/\tau^2\rceil$ and
    $r_{\mathrm{sparse}}\gets \lceil4\beta^2/\rho^2\rceil$.

    \item Run \Cref{alg:dense-case-random-row-search} with threshold $\tau$, degree parameter $r_{\mathrm{dense}}$, certificate threshold $\beta$, edge threshold $d^{1.1}$, and failure probability $\delta/2$.

    \item If the density check outputs a vector $u$, return $u$.

    \item Run the covariance detection algorithm from
    \Cref{fact:cordetect} on $T$ with parameters $(\delta/2,\rho,\tau,D)$.
    Let $E_\rho$ be the returned list of coordinate pairs satisfying
    $|\Sigma_{ij}|\ge\rho$, and set $G_\rho\gets([d],E_\rho)$.

    \item \label{line:call-sparse-alg} Run \Cref{alg:graph-branching-search} on $(\Sigma,G_\rho,\rho,\beta,r_{\mathrm{sparse}},L)$ and return its output.
\end{enumerate}
}}
\caption{Correlation-Detection Branching Search for Sparse PCA.}
\label{alg:cd-branching-search}
\end{algorithm}

\begin{theorem}[Subquadratic support recovery from correlation detection]
\label{thm:subquadratic-correlation-detection-pca}
Fix $\delta\in(0,1)$, a certificate threshold $\beta>0$, an approximation slack $\eta>0$, a strong correlation threshold $\rho\in(0,1)$, and a margin threshold $\tau\in(0,1)$. Set $r_{\mathrm{dense}}=\left\lceil4\beta^2/\tau^2\right\rceil$, $r_{\mathrm{sparse}}=\left\lceil4\beta^2/\rho^2\right\rceil$, and $\overline{D}=\max\{2\rho,\beta\}$.
Assume $\rho>12\tau$ and $d^{1.1}\ge 2d r_{\mathrm{dense}}$.

Let $T\subseteq\R^d$ be a data set of size $n$ with empirical covariance matrix $\Sigma$. Assume that $\Sigma$ is $\max\{r_{\mathrm{dense}}+1,r_{\mathrm{sparse}}+1\}$-PSD. 

Suppose there exist a unit vector $v\in\R^d$ and a set $E\subseteq[d]$ such that $\|v\|_0\le k$, $v^\top\Sigma v\ge\beta+\eta$, $\supp(v)\subseteq E$, and $E$ has $\rho$-correlation radius at most $L$. Then, with probability at least $1-\delta$, \Cref{alg:cd-branching-search} returns one of the following:
\begin{enumerate}[(i)]
    \item a unit vector $u$ such that $\|u\|_0\le r_{\mathrm{dense}}$ and $u^\top\Sigma u>\beta$;
    \item a family $\mathcal B$ of at most $d(L+1)$ candidate supports, each of size at most $(r_{\mathrm{sparse}}+1)^L$, such that $E\subseteq B$ for some $B\in\mathcal B$.
\end{enumerate}
The running time is
\[
O\left(
\left(
\left(
    d^{1.72}
    +
    d^{1.62+2.4\frac{\log(4\overline{D}/\rho)}
                                  {\log(\overline{D}/(3\tau))}}
\right)
\poly(n,\log d,\overline{D}/\tau)
+
 d r_{\mathrm{sparse}}(r_{\mathrm{sparse}}+1)^L
\right)
\log(1/\delta)
\right).
\]
\end{theorem}

\begin{proof}
The algorithm first queries all diagonal entries.  If it finds
$\Sigma_{ii}>\beta$, then the $1$-sparse unit vector $e_i$ satisfies
$e_i^\top\Sigma e_i>\beta$ and
$\|e_i\|_0=1\le r_{\mathrm{dense}}$, giving the first conclusion.  Otherwise,
$\max_i\Sigma_{ii}\le\beta$.  Hence the  parameter used in
\Cref{fact:cordetect} satisfies
$D=\max\left\{2\rho,\max_i\Sigma_{ii}\right\}
\le\max\{2\rho,\beta\}=\overline{D}$.

For $\theta\in\{\tau,\rho\}$, let
$E_\theta=\{\{i,j\}:i\neq j,\ |\Sigma_{ij}|\ge\theta\}$.
The algorithm next runs the density check at threshold $\tau$ with edge threshold $d^{1.1}$, degree parameter $r_{\mathrm{dense}}$, and certificate threshold $\beta$. Since $r_{\mathrm{dense}}\tau^2>\beta^2$ and $d^{1.1}\ge2d r_{\mathrm{dense}}$, \Cref{lem:densesearch} implies that if $|E_\tau|\ge d^{1.1}$, then the density check returns an $r_{\mathrm{dense}}$-sparse vector with quadratic form greater than $\beta$ with probability $1-\delta/2$. This satisfies the first conclusion. If the algorithm does not terminate in the density-check stage, we may assume $|E_\tau|<d^{1.1}$ with probability $1-\delta/2$.

Condition on this event. Since $\rho>12\tau$ by assumption, the covariance detection algorithm from \Cref{fact:cordetect} returns all pairs in $E_\rho$ with probability at least $1-\delta/2$, using the sparsity bound $|E_\tau|<d^{1.1}$. Moreover, since $\rho>\tau$, we have $E_\rho\subseteq E_\tau$, and hence $|E_\rho|\le |E_\tau|<d^{1.1}$. Thus the graph passed to the sparse graph search has at most $d^{1.1}$ edges.

We now apply \Cref{lem:sparse-search} with strong threshold $\rho$, certificate threshold $\beta$, and degree parameter $r_{\mathrm{sparse}}$. If the graph search returns a vector, it is $r_{\mathrm{sparse}}$-sparse and has quadratic form larger than $\beta$. Since $\tau<\rho$, we have $r_{\mathrm{sparse}}\le r_{\mathrm{dense}}$, so this again satisfies the first conclusion. Otherwise, \Cref{lem:sparse-search} returns a family $\mathcal B$ satisfying the second conclusion. A union bound over the density-check and correlation-detection events shows that the conclusions hold with probability at least $1-\delta$.

Finally, we bound the runtime. The diagonal scan makes $d$ entry
queries and costs $O(nd)$ time when the entries are computed from the
samples; this is absorbed by the bound below. The density check
(\Cref{alg:dense-case-random-row-search}) costs
$\widetilde O(d^2/d^{1.1}+d)=\widetilde O(d)$.
Correlation detection (\Cref{fact:cordetect}) costs
\[
\left(
    d^{1.1}d^{0.62}
    +
    d^{1.62+2.4\frac{\log(4D/\rho)}{\log(D/(3\tau))}}
\right)\poly(n,\log d,D/\tau)\log(1/\delta).
\]
The exponent function
$\frac{\log(4x/\rho)}{\log(x/(3\tau))}$ is increasing for $x\ge2\rho$,
since its derivative
$\frac{\log(\rho/(12\tau))}{x\bigl(\log(x/(3\tau))\bigr)^2}$ is positive
when $\rho>12\tau$.  Thus $D\le\overline{D}$ allows us to replace $D$ by
$\overline{D}$ in both the exponent and the polynomial factor.
Consequently, the correlation-detection cost is at most
\[
\left(
    d^{1.72}
    +
    d^{1.62+2.4\frac{\log(4\overline{D}/\rho)}
                                  {\log(\overline{D}/(3\tau))}}
\right)
\poly(n,\log d,\overline{D}/\tau)\log(1/\delta).
\]
The final graph search costs $O(d^{1.1}+d r_{\mathrm{sparse}}(r_{\mathrm{sparse}}+1)^L)$. The $d^{1.1}$ term is absorbed by the displayed correlation-detection bound. Combining the three bounds gives the claimed runtime.
\end{proof}

The same construction also applies when the sparse PCA routines
operate on a diagonal shift $A=\Sigma-\lambda I$, when $A$ is PSD over sparse directions.  In that case, covariance
detection is applied to the empirical covariance matrix $\Sigma$, while the
density check, branching search, and optimization routines are applied to
$A$.  This is valid because $A$ and $\Sigma$ have identical off-diagonal
entries.  Such a shift is useful for distinguishing small additive gaps,
provided that $A$ satisfies the required PSD over sparse directions assumption close to the additive threshold.  We state and
prove this shifted variant in \Cref{app:shifted-covariance-detection} and use
it in \Cref{cor:subquadratic-identity-subgaussian-sparse-mean}.

The preceding structural theorem again gives two natural algorithmic consequences.

\begin{corollary}[Subquadratic sparse PCA via the constrained power method]
\label{cor:subquadratic-constrained-power}
Under the assumptions of \Cref{thm:subquadratic-correlation-detection-pca}, and assuming in addition that $\Sigma$ is $\max\{r_{\mathrm{dense}}+1,(r_{\mathrm{sparse}}+1)^L\}$-PSD, run \Cref{alg:cd-branching-search} with failure probability $\delta/2$. If it returns a vector, return that vector. Otherwise, run the constrained power method on every $B\in\mathcal B$ long enough to return a vector of value at least $\lambda_{\max}(\Sigma_{B,B})-\eta$, with failure probability $\delta/2$, and return the vector with the largest quadratic form. Then, with probability at least $1-\delta$, the returned unit vector $u$ satisfies $u^\top\Sigma u\ge\beta$ and
\[
\|u\|_0\le\max\left\{r_{\mathrm{dense}},(r_{\mathrm{sparse}}+1)^L\right\}.
\]
The running time is
\[
O\left(
\left(
    d^{1.72}
    +
    d^{1.62+2.4\frac{\log(4\overline{D}/\rho)}
                                  {\log(\overline{D}/(3\tau))}}
\right)
\poly(n,\log d,\overline{D}/\tau)\log(1/\delta)
+
 d(r_{\mathrm{sparse}}+1)^{O(L)}T_{\mathrm{pm}}
\right),
\]
where $T_{\mathrm{pm}}$ denotes the runtime of the power-method subroutine.
\end{corollary}

\begin{proof}
With probability at least $1-\delta/2$, \Cref{thm:subquadratic-correlation-detection-pca} either returns a vector with quadratic form larger than $\beta$, or gives a candidate $B\in\mathcal B$ containing $\supp(v)$. In the latter case, $\lambda_{\max}(\Sigma_{B,B})\ge v^\top\Sigma v\ge\beta+\eta$. Conditional on this event, the power-method subroutine on this support returns a vector with quadratic form at least $\beta$ with probability at least $1-\delta/2$. Since the procedure returns the vector with the largest quadratic form among all candidate supports, a union bound over these two events proves the claim.
\end{proof}

The next corollary, under the constant-gap specialization described in \Cref{rem:constant-gap-subquadratic-sdp}, proves \Cref{thm:informal-subquadratic-sdp}.
\begin{corollary}[Subquadratic sparse PCA via the semidefinite relaxation]
\label{cor:subquadratic-sdp}
Under the assumptions of \Cref{thm:subquadratic-correlation-detection-pca}, run \Cref{alg:cd-branching-search}. If it returns a vector, return that vector. Otherwise, compute a feasible additive-$\eta$ approximate solution to $\operatorname{SDP}_k(\Sigma_{B,B})$ for every $B\in\mathcal B$. Pad each approximate solution with zeros outside $B\times B$ and return the padded solution with the largest objective value. Then, with probability at least $1-\delta$, the procedure returns one of the following:
\begin{enumerate}[(i)]
    \item a unit vector $u$ such that $\|u\|_0\le r_{\mathrm{dense}}$ and $u^\top\Sigma u>\beta$;
    \item a matrix $X\in\cX_{\sdp,k}$ such that $\langle\Sigma,X\rangle\ge\beta$.
\end{enumerate}
Using a polynomial-time SDP solver, the running time is
\[
O\left(
\left(
    d^{1.72}
    +
    d^{1.62+2.4\frac{\log(4\overline{D}/\rho)}
                                  {\log(\overline{D}/(3\tau))}}
\right)
\poly(n,\log d,\overline{D}/\tau)\log(1/\delta)
+
 d(r_{\mathrm{sparse}}+1)^{O(L)}/\eta^{O(1)}
\right),
\]
where the constant hidden in the final exponent depends on the SDP solver.
\end{corollary}

\begin{proof}
The direct-vector case follows from \Cref{thm:subquadratic-correlation-detection-pca}. Otherwise, some candidate $B\in\mathcal B$ contains $\supp(v)$, so $vv^\top$ is feasible for $\operatorname{SDP}_k(\Sigma_{B,B})$ and has objective value $v^\top\Sigma v\ge\beta+\eta$. Thus, the additive-$\eta$ approximate solution corresponding to this support has objective value at least $\beta$. Padding it with zeros preserves positive semidefiniteness, trace, entrywise $\ell_1$ norm, and objective value. Since the procedure returns the padded solution with the largest objective value, the returned matrix has objective value at least $\beta$.
\end{proof}

\begin{remark}[The constant-gap specialization]
\label{rem:constant-gap-subquadratic-sdp}
We explain how \Cref{cor:subquadratic-sdp} gives the guarantee in \Cref{thm:informal-subquadratic-sdp}.  Set $\beta=1$, $\eta=1/10$, $\rho=1/(4k)$, $\tau=(100k)^{-\gamma}$, and $L=\max\{1,\lceil C\log k\rceil\}$, where $C>0$ and $\gamma>0$ are sufficiently large universal constants.  Suppose that $\Sigma$ is an empirical covariance matrix satisfying $\|\Sigma\|_{\op,k}>2$.  By \Cref{lem:quantitative-bounded-radius-reduction}, there is a set $T$ whose $\rho$-correlation graph has radius at most $L$ and such that $\lambda_{\max}(\Sigma_{T,T})\ge (2-\rho k)/k^{1/(2L)}-\rho k\ge1+\eta$, where the final inequality holds for sufficiently large $C$.  A unit leading eigenvector of $\Sigma_{T,T}$ therefore supplies the bounded-radius witness required by \Cref{cor:subquadratic-sdp}.  For these parameters, $r_{\mathrm{sparse}}=O(k^2)$, $r_{\mathrm{dense}}=k^{O(\gamma)}$, and $(r_{\mathrm{sparse}}+1)^L=k^{O(\log k)}$.  Assume in addition that $d^{1.1}\ge2d r_{\mathrm{dense}}$, as required by \Cref{thm:subquadratic-correlation-detection-pca}; for fixed $\gamma$ this condition holds for all sufficiently large $d$ in the regime $k=\polylog(d)$.  Moreover, $\overline D=\max\{2\rho,\beta\}=1$, and hence the correlation-detection exponent is $1.62+O(1/\gamma)<2$ for sufficiently large $\gamma$.  Substituting these parameter bounds into \Cref{cor:subquadratic-sdp} gives the constant-gap subquadratic guarantee stated in \Cref{thm:informal-subquadratic-sdp}.
\end{remark}

\subsection{Empirical Access and Sample Complexity}
\label{sec:sample-complexity}
In this section, we discuss a statistical setup. 
We are interested in solving the sparse PCA objective with the covariance matrix $\Sigma$, but instead of observing this matrix directly, 
we only have sample access to a mean-zero subgaussian distribution $P$ with covariance $\Sigma$.
That is, we have access only to $x_1,\ldots,x_n\in\R^d$, which are i.i.d.\ samples from $P$.
Our goal is to show that $\widetilde{O}(k^2)$ samples suffice to find a relaxed witness for $\Sigma$ in the semidefinite relaxation.

Let $x_1,\ldots,x_n\in\R^d$ be i.i.d.\ samples from $P$ and observe that $\Sigma=\E[xx^\top]$. Given these samples, we form the empirical covariance matrix
$
    \widehat\Sigma
    =
    \frac1n\sum_{i=1}^n x_i x_i^\top .
$
The deterministic guarantees above transfer to the empirical
setting through two concentration events.  The input event ensures that the
population witness retains sufficiently large quadratic form for
$\widehat\Sigma$, while the output event transfers the certificate produced
for $\widehat\Sigma$ back to the population covariance $\Sigma$.

\begin{corollary}[Empirical sparse PCA]
\label{cor:empirical-sparse-pca-tradeoff}
\label{thm:empirical-bounded-diameter-sparse-pca}
\label{cor:empirical-quasipolynomial-sparse-pca}
\label{cor:empirical-polynomial-root-k-sparse-pca}
There exists a sufficiently large constant $C>0$ such that the following holds. Let $x_1,\ldots,x_n\in\R^d$ be i.i.d.\ subgaussian random vectors with zero mean and covariance matrix $\Sigma$.
Let
$
    \widehat\Sigma=\frac1n\sum_{i=1}^n x_i x_i^\top.
$
Fix $\eta,\delta\in(0,1)$, a certificate threshold $\beta>0$, a correlation threshold $\tau>0$, and $r\in\Z_+$ satisfying $r\tau^2>\beta^2$. Let $L\ge1$ be an integer. Suppose there exists a unit vector $v\in\R^d$ such that $\|v\|_0\le k$ and
$    v^\top\Sigma v=R$,
    $\frac{R-\tau k}{k^{1/(2L)}}-\tau k
    \ge
    \beta+2\eta$.
Then 
\begin{enumerate}[(i)]
    \item Let
    $
        s_{\mathrm{pm}}=\max\{k,r,(r+1)^{O(L)}\}.
    $
    If the candidate supports are processed using the constrained power method, then for
    $
        n\ge
        C\,s_{\mathrm{pm}}\frac{\log(d/\delta)}{\eta^2},
    $
    with probability at least $1-\delta$ the procedure returns a unit vector $u$ with  $\|u\|_0\le s_{\mathrm{pm}}$ satisfying
    $
        u^\top\Sigma u\ge\beta-\eta.
    $
    Moreover,  
    the running time is
$O\!\left(
nd^2+d(r+1)^{O(L)}T_{\rm pm}
\right)$,
    where $T_{\mathrm{pm}}$ denotes the runtime of the power-method subroutine.

    \item Using the semidefinite relaxation on the candidate supports, if
    $
        n\ge
        C\,\max\{r,k^2\}{\log(d/\delta)}/{\eta^2},
    $
    with probability at least $1-\delta$ the procedure returns either a unit vector $u$ such that $\|u\|_0\le r$ and
    $
        u^\top\Sigma u>\beta-\eta,
    $
    or a matrix $X\in\cX_{\sdp,k}$ such that
    $
        \langle\Sigma,X\rangle\ge\beta-\eta.
    $
    The running time is $O\!\left(nd^2+d(r+1)^{O(L)}/\eta^{O(1)}\right)$, where the constant hidden in the exponent depends on the SDP solver.
\end{enumerate}
\end{corollary}

Before we move on to the proof we remark the above corollary consequences for constant approximation factor. 
\begin{remark}[Empirical sparse PCA constant factor]
Let $\beta=1$, $\eta=1/10, \delta=2/3, L=C_0\log(k)$, $\tau=1/(C_0k)$, and $r=C_1k^2$, where $C_0,C_1>0$ are sufficiently large universal constants. 
Then if  $\|\Sigma\|_{\op,k}\ge 2$ we have that 
${(2-\tau k)}/{k^{1/(2L)}}-\tau k
\ge 1+2\eta.
$
Hence the assumptions of \Cref{cor:empirical-sparse-pca-tradeoff} are satisfied. 

As a result, there exists an algorithm that uses at most $O(k^2\log(d))$ samples, $O(d^2k^2\log(d)+dk^{O(\log(k))})$ time and with probability at least $2/3$  returns either an $O(k^2)$-sparse unit vector $u$ with $u^\top\Sigma u>9/10$, or a matrix $X\in\cX_{\sdp,k}$ with $
\langle\Sigma,X\rangle\ge9/10.$

Thus, in the constant-gap sparse PCA setting, $\widetilde O(k^2)$ samples suffice to obtain a constant-value relaxed certificate for the population covariance. A similar result holds for returning a sparse vector but with higher dependence on $k$ and larger output sparsity. 
\end{remark}

\begin{proof}[Proof of \Cref{cor:empirical-sparse-pca-tradeoff}]
For the constrained-power conclusion, by
\Cref{fact:sparse-quadratic-concentration}, with the stated sample size, with
probability at least $1-\delta/2$,
\[
    \sup_{\|x\|_2=1,\ \|x\|_0\le s_{\mathrm{pm}}}
    \left|x^\top(\widehat\Sigma-\Sigma)x\right|
    \le\eta.
\]
Since $s_{\mathrm{pm}}\geq k$ for the  $k$-sparse  $v$, the concentration event gives
$v^\top\widehat\Sigma v\ge R-\eta$.
Therefore, because $k^{1/(2L)}\ge1$,
\[
    \frac{v^\top\widehat\Sigma v-\tau k}{k^{1/(2L)}}-\tau k
    \ge
    \frac{R-\tau k}{k^{1/(2L)}}-\tau k-\eta
    \ge
    \beta+\eta.
\]
Since $\widehat\Sigma\succeq0$, the  PSD assumptions required
by \Cref{cor:parameterized-radius-sparse-pca} are automatic.  Conditional on
the concentration event, \Cref{cor:parameterized-radius-sparse-pca}, applied
to $\widehat\Sigma$ with certificate threshold $\beta$, approximation slack
$\eta$, threshold $\tau$, truncation parameter $r$, and failure probability
$\delta/2$, returns an $s$-sparse unit vector $u$ satisfying
$u^\top\widehat\Sigma u\ge\beta$ with probability at least $1-\delta/2$.
On the same concentration event, $u^\top\Sigma u\ge\beta-\eta$.  A union
bound over the concentration event and the constrained-power event gives
overall success probability at least $1-\delta$.

For the semidefinite conclusion, apply
\Cref{fact:sparse-quadratic-concentration} at sparsity $\max\{k,r\}$ with
failure probability $\delta/2$, and apply
\Cref{fact:l1-matrix-concentration} with $K=k$ and failure probability
$\delta/2$.  By a union bound, both events hold with probability at least
$1-\delta$ under the stated sample bound. In particular, the input vector $v$
satisfies the same lower bound for $\widehat\Sigma$ as above, so
\Cref{cor:parameterized-radius-sparse-pca} applies to $\widehat\Sigma$ with its
 PSD assumptions again automatic from $\widehat\Sigma\succeq0$. If it
returns the direct vector $u$, sparse concentration gives
$u^\top\Sigma u\ge u^\top\widehat\Sigma u-\eta>\beta-\eta$.
Otherwise it returns a matrix $X$ with $\|X\|_1\le k$ and $\langle\widehat\Sigma,X\rangle\ge\beta$. By \Cref{fact:l1-matrix-concentration},
$\langle\Sigma,X\rangle\ge\langle\widehat\Sigma,X\rangle-\eta\ge\beta-\eta$.
Forming $\widehat\Sigma$ takes $nd^2$ time. The remaining conclusions are inherited from \Cref{cor:parameterized-radius-sparse-pca}.
\end{proof}

\section{Robust Sparse Mean Estimation}\label{sec:mean-estimation}
In this section, we apply the sparse PCA routines from the previous section to the problem of robust sparse mean estimation.
In \Cref{sec:reduction-sparse-pca}, we show that the sparse PCA subroutine fits naturally in the existing filtering-based framework for robust estimation.
We then prove concrete 
guarantees for heavy-tailed distributions in \Cref{sec:bdd-cov} and for isotropic subgaussian distributions in \Cref{sec:identity-subgaussian}.

\subsection{Reduction of Robust Mean Estimation to Sparse PCA}
\label{sec:reduction-sparse-pca}
We first reduce robust mean estimation in the sparse norm to a
sparse PCA certificate routine.
We will show that the proposed sparse PCA routine fits into the filtering-based framework as long as the datasets satisfy appropriate stability conditions.

Recall the sparse norm $\|\cdot\|_{2,k}$ from the notation
paragraph, the maximum sparse quadratic-form  $\|\cdot\|_{\op,k}$ from
\Cref{eq:sparse-op-def}, and the SDP feasible set $\cX_{\sdp,s}$ from
\Cref{eq:sdp-set-def}.
Also recall from the notation paragraph that $Q_S(\mu)$ denotes the empirical second moment of $S$ with respect to $\mu$.

The following stability condition provides a unified way to analyze both the SDP relaxation (the parameter $q\geq k$) and the $\ell$-sparse relaxation (the parameter $\ell \geq k$); we
will use only one at a time.
\begin{definition}[Stability]\label{def:stability}Let $G\subseteq\R^d$ be a dataset, let $\mu\in\R^d$ be a candidate mean, and let $k,q,\ell\in\N$ be sparsity parameters with $\ell\ge k$.
Let $\eps\in(0,1/2)$ be the contamination rate, and let $\gamma,\Gamma\ge0$ be the mean- and covariance-stability parameters.
We say that $G$ is $(\eps,\gamma,\Gamma,k,q,\ell)$-stable with respect to $\mu$ if, for every $G'\subseteq G$ with $|G'|\ge(1-10\eps)|G|$, we have that
\begin{enumerate}
    \item (Closeness in mean) $    \|\mu_{G'}-\mu\|_{2,k}\le\gamma,
$
\item (Closeness for the SDP relaxation) 
$    \sup_{W\in\cX_{\sdp,q}}
    \left|\left\langle Q_{G'}(\mu)-I,W\right\rangle\right|
    \le \Gamma,
$
\item (Closeness for the vector relaxation) $  \sup_{\|u\|_2=1,\ \|u\|_0\le\ell}
    \left|u^\top\bigl(Q_{G'}(\mu)-I\bigr)u\right|
    \le \Gamma$.
\end{enumerate}

\end{definition}

The next definition captures what we need from the sparse PCA subroutine, which has exactly the same bicriteria form discussed in the introduction.
\begin{definition}[Sparse PCA certificate routine]
\label{def:sparse-pca-certificate-routine}
Fix parameters $k,q,\ell,s\in\Z_+$ and thresholds $R>\beta>0$.  We say that
$\mathcal A_{\mathrm{cert}}$ is a $(k,q,\ell,s,R,\beta)$-sparse PCA certificate
routine with runtime $T_{\mathrm{cert}}$ if the following holds.  For every
 multiset $T\subseteq\R^d$ of size $m$, let $\Sigma\succeq0$ denote its
empirical covariance.  For every failure probability
$\rho\in(0,1)$, one call to $\mathcal A_{\mathrm{cert}}$ on $T$ runs in time
$T_{\mathrm{cert}}(m,d,k,q,\ell,s,\rho)$ and, with probability at least
$1-\rho$, returns either \textsc{No}, in which case $\|\Sigma\|_{\op,k}<R$, or one of the following certificates:
\begin{enumerate}[leftmargin=*,nosep]
        \item a matrix certificate $W\in\cX_{\sdp,q}$, supported on at most
        $s$ coordinates, with $\langle \Sigma,W\rangle>\beta$; or
        \item a vector certificate $u\in\R^d$ with $\|u\|_2=1$,
        $\|u\|_0\le\ell$, and $u^\top \Sigma u>\beta$.
\end{enumerate}
\end{definition}

We now show that if the (uncontaminated) dataset satisfies the stability condition above, and we have access to a sparse PCA certification subroutine,
then the filtering-based algorithms robustly estimate the mean in the $\|\cdot\|_{2,k}$ norm.
Here, we will use the standard filtering-based framework~\cite{DiaKan22-book}.
In this framework, we iteratively remove points that are deemed to be outliers, and a point is deemed an outlier if it has a large score.
These scores are defined using the certificates from the sparse PCA subroutine above.
\paragraph{Filtering-based algorithms.}

Given a multiset $T$, its empirical mean $m$, and either a matrix certificate $W\in\cX_{\sdp,q}$ or a unit $\ell$-sparse vector certificate $u$, first form the raw quadratic scores
\[
    g(x)=
    \begin{cases}
        (x-m)^\top W(x-m), & \text{for a matrix certificate},\\[1mm]
        \langle u,x-m\rangle^2, & \text{for a vector certificate}.
    \end{cases}
\]
Let $L$ consist of the $\lfloor2\eps|T|\rfloor$ largest raw scores and define
the cutoff scores $\tau(x)\eqdef g(x)\Ind\{x\in L\}$.
Only these top scores participate in the filter.  If $\tau_{\max}\eqdef\max_{x\in T}\tau(x)>0$, delete every $x\in T$ independently with probability $\tau(x)/\tau_{\max}$.  This is the universal-filter cutoff of \cite[Proposition~2.19]{DiaKan22-book}, followed by the randomized filtering rule formalized in \Cref{lem:randomized-score-filter}.  A maximum cutoff-score point is deleted with probability one, so every nonterminal filtering round removes at least one sample.
The next statement records the technical statement that will suffice for our applications.

\begin{algorithm}[H]
    \centering
    \fbox{\parbox{6in}{
        {\bf Input:}
        An $\eps$-corrupted multiset $T\subseteq\R^d$, target sparsity $k$,
        a detection threshold $R$, a covariance-stability error $\Gamma$,
        a matrix certificate sparsity parameter $q$, a vector certificate
        sparsity parameter $\ell$, a matrix-certificate support bound $s$,
        a certificate failure parameter $\delta_{\rm cert}\in(0,1)$, and
        a sparse PCA certificate routine $\mathcal A_{\mathrm{cert}}$.\\
        {\bf Output:}
        An estimate $\widehat\mu\in\R^d$.

        \begin{enumerate}[leftmargin=*]
            \item Initialize $T^{(0)}\gets T$ and $t\gets0$.

            \item Compute the empirical mean $m^{(t)}$. Let
            $M^{(t)}$ denote the empirical covariance of $T^{(t)}$ around
            $m^{(t)}$, used only to specify the guarantee of the certificate
            routine; the algorithm does not explicitly form $M^{(t)}$.

            \item Set
            $
                \beta_{\Gamma}
                \eqdef
                1+48\Gamma .
            $
            Run $\mathcal A_{\mathrm{cert}}(T^{(t)},k,q,\ell,R,\beta_{\Gamma},s,
            \delta_{\rm cert}/n)$, where $n=|T|$, on the current set $T^{(t)}$.
            The routine returns either \label{line:run}
            \begin{enumerate}[leftmargin=*,nosep]
                \item \textsc{No}, certifying $\|M^{(t)}\|_{\op,k}<R$; or
                \item a matrix certificate $W\in\cX_{\sdp,q}$, supported
                on at most $s$ coordinates, with
                $\langle M^{(t)},W\rangle>\beta_{\Gamma}$; or
                \item a vector certificate $u\in\R^d$ with $\|u\|_2=1$,
                $\|u\|_0\le\ell$, and
                $(u^\top M^{(t)}u)>\beta_{\Gamma}$.
            \end{enumerate}

            \item If the output is \textsc{No}, return
            $\widehat\mu\gets m^{(t)}$.

            \item Otherwise, compute the raw scores from the returned certificate:
            \[
                g^{(t)}(x)=
                \begin{cases}
                    (x-m^{(t)})^\top W(x-m^{(t)}),
                    & \text{if the certificate is } W,\\[2mm]
                    \langle u,x-m^{(t)}\rangle^2,
                    & \text{if the certificate is } u,
                \end{cases}
                \qquad x\in T^{(t)}.
            \]
            Let $L_t$ contain the $\lfloor2\eps|T^{(t)}|\rfloor$ points with largest raw scores, set
            $\tau^{(t)}(x)=g^{(t)}(x)\Ind\{x\in L_t\}$,
            and let $\tau_{\max}^{(t)}=\max_x\tau^{(t)}(x)$.  In this branch $\tau_{\max}^{(t)}>0$.  Independently for every
            $x\in T^{(t)}$, delete $x$ with probability
            $\tau^{(t)}(x)/\tau_{\max}^{(t)}$, and let $T^{(t+1)}$ be the
            multiset of undeleted samples.

            \item Set $t\gets t+1$ and repeat.
        \end{enumerate}
    }}
    \caption{Filtering from a sparse PCA certificate.}
    \label{alg:robust-sparse-mean}
\end{algorithm}

\begin{proposition}[Sparse-norm mean estimation from stability and certificates]
\label{thm:robust-sparse-mean-from-pca}
Let $T, n\eqdef |T|$ be an $\eps$-corruption of a multiset $G$ with reference mean $\mu$, where $\eps\le1/18$.  Suppose that $G$ is $(\eps,\gamma,\Gamma,k,q,\ell)$-stable with respect to $\mu$ and that $\Gamma\ge\eps$.  Define
$
    \beta_\Gamma\eqdef1+48\Gamma.
$
Fix $\delta\in(0,1)$ and $R>\beta_\Gamma$.  
Suppose 
that $\mathcal A_{\mathrm{cert}}$ is a $(k,q,\ell,s,R,\beta_\Gamma)$-sparse PCA certificate routine in the sense of \Cref{def:sparse-pca-certificate-routine}.  Then there exists an algorithm that, given $T$ and access to $\mathcal A_{\mathrm{cert}}$, with probability at least $1-\delta$ returns an estimate $\widehat\mu$ satisfying
\[
    \|\widehat\mu-\mu\|_{2,k}
    \le
    O\!\left(\gamma+\sqrt{\eps\bigl(R-1+\Gamma\bigr)}\right).
\]
Moreover, the running time is
\[
    O\!\left(
        nT_{\mathrm{cert}}\!\left(n,d,k,q,\ell,s,\frac1{24n}\right)
            \log(1/\delta)
    +n^2(d+s^2+\ell)\log(1/\delta)
            +d\log^2(1/\delta)
    \right).
\]
\end{proposition}

\begin{proof}
We first prove that one run of \Cref{alg:robust-sparse-mean} with $\delta_{\rm cert}=1/24$ succeeds with probability at least $5/8$. 
Note that the algorithm makes at most $n$ certificate calls since one point is deleted with probability one. 
Let $\mathcal E_{\rm cert}$ be the event that no certificate call fails, by the union bound since we use $\delta_{\rm cert}/n$ failure probability at each call (Line \ref{line:run}) the union bound gives $\Pr[\mathcal E_{\rm cert}]\ge1-1/24$.

Write the initial sample multiset as $T=G_0\cup B_0$, where $G_0\subseteq G$, $|G_0|\ge(1-\eps)|G|$, and $|B_0|\le\eps|G|$.  At round $t$, write $T^{(t)}=G_t\cup B_t$.  Suppose that before this round at most $3|B_0|$ points of $G_0$ have been removed.  Then $|G_t|\ge(1-4\eps)|G|$ and $\alpha_t\eqdef |B_t|/|T^{(t)}|\le\eps/(1-4\eps)\le2\eps$.
In a nonterminal round define
\[
    \widetilde W=
    \begin{cases}
        W, & \text{for a matrix certificate},\\
        uu^\top, & \text{for a vector certificate}.
    \end{cases}
\]
In both cases $\widetilde W\succeq0$, $\operatorname{tr}(\widetilde W)=1$,
and $\langle M^{(t)},\widetilde W\rangle>1+48\Gamma$.
Let $L_t$ be the cutoff set used by the algorithm and put $H_t=G_t\setminus L_t$.  Because $|T^{(t)}|\le|T|=|G|$ and $|L_t|\le2\eps|T^{(t)}|$,
$|H_t|\ge(1-6\eps)|G|\ge(1-10\eps)|G|$.
Stability therefore applies to both $G_t$ and $H_t$.  For either choice of $\widetilde W$ it gives
\[
    \left|
    \frac1{|K|}\sum_{x\in K}(x-\mu)^\top\widetilde W(x-\mu)-1
    \right|\le\Gamma,
    \qquad K\in\{G_t,H_t\}.
\]

The set $B_t$ must contain at least one point in this round.  Indeed, if $B_t=\emptyset$, centering at the empirical mean can only decrease a positive-semidefinite quadratic second moment, so stability for $G_t$ would give $\langle M^{(t)},\widetilde W\rangle\le1+\Gamma$, a contradiction.  Consequently $|B_t|\ge1$, and $|B_t|\le2\eps|T^{(t)}|$ implies that $L_t$ is nonempty.  The cutoff score-mass inequality, \Cref{lem:score-mass-inequality}, now yields $\sum_{x\in B_t}\tau^{(t)}(x)\ge\sum_{x\in G_t}\tau^{(t)}(x)$.
Thus the hypothesis of the black-box randomized filtering lemma, \Cref{lem:randomized-score-filter}, holds in every relevant nonterminal round of the stopped process.  With probability at least $2/3$, at most $3|B_0|$ points from $G_0$ are removed throughout this process.  On $\mathcal E_{\rm cert}$, the stopped process coincides with the original run.  A union bound therefore shows that, with probability at least $2/3-1/24=5/8$, every certificate call is successful and at most $3|B_0|$ good points are removed.

On this event the algorithm terminates with a nonempty set $T'=G'\cup B'$ satisfying $|G'|\ge(1-4\eps)|G|$ and $\alpha\eqdef|B'|/|T'|\le2\eps$.  Its terminal \textsc{No} certificate gives $\|M_{T'}\|_{\op,k}<R$.  Applying \Cref{lem:stopping-implies-accuracy}, including the lower half of rank-one stability for $G'$, gives $\|\widehat\mu-\mu\|_{2,k}=O\!\left(\gamma+\sqrt{\eps\bigl(R-1+\Gamma\bigr)}\right)$.

The empirical mean costs $O(md)$ on a current set of size $m$.  Matrix-certificate scores cost $O(ms^2)$, vector-certificate scores cost $O(m\ell)$, and the largest $\lfloor2\eps m\rfloor$ scores can be selected in $O(m)$ time.  Summing over at most $n$ rounds gives
\[
    O\!\left(
        nT_{\mathrm{cert}}\!\left(n,d,k,q,\ell,s,\frac1{24n}\right)
        +n^2(d+s^2+\ell)
    \right).
\]
To obtain failure probability $\delta$, run \Cref{alg:robust-sparse-mean} independently $m=\lceil C'\log(1/\delta)\rceil$ times, for a sufficiently large constant $C'>0$, on the original input $T$.  Denote the resulting candidates as $\widehat\mu_1,\ldots,\widehat\mu_m$, and set $\Delta=C\left(\gamma+\sqrt{\eps(R-1+\Gamma)}\right)$,
where $C$ is a sufficiently large universal constant so that each candidate is within $\Delta$ of $\mu$ in $\|\cdot\|_{2,k}$ with probability at least $5/8$. By Hoeffding's inequality with probability at least $1-\delta$ we have that at least  $7m/12$ out of the $m$ candidates are $\Delta$ close to $\mu$ in $\|\cdot\|_{2,k}$.

Scan the candidates and return any $\widehat\mu_i$ such that at least $7m/12$ $\widehat\mu_j$ are within $2\Delta$  in $\|\cdot\|_{2,k}$ of it.
Then for some $\widehat\mu_j$ that $\|\widehat\mu_j-\mu\|_{2,k}\leq \Delta$ we have that $\|\widehat\mu_i-\widehat\mu_j\|_{2,k}\leq 2\Delta$. 
 By the triangle inequality we have that $\|\widehat\mu_i-\mu\|_{2,k}\leq3\Delta$,  with probability at least $1-\delta$.
Each sparse-norm can be computed in $O(d)$ time by selecting the $k$ largest squared coordinates of the difference.  Thus the candidate scan costs $O(dm^2)=O(d\log^2(1/\delta))$. Which concludes the proof of \Cref{thm:robust-sparse-mean-from-pca}.
\end{proof}

\begin{remark}[Sparse-direction error and approximate sparsity]
\label{rem:sparse-direction-error}
The filtering guarantees imply that, even when $\mu$ is not sparse, the
estimator is accurate along every $k$-sparse direction.

Moreover, suppose that
$
    \|\widetilde\mu-\mu\|_{2,k}\le \delta,
$
and define
$
    \widehat\mu:=\topp_k(\widetilde\mu),
    \mu_k:=\topp_k(\mu).
$
By \Cref{lem: ksparsenorm},
\[
\begin{aligned}
    \|\widehat\mu-\mu_k\|_2
    &\le \sqrt{6}\,
        \|\widetilde\mu-\mu_k\|_{2,k} \\
    &\le \sqrt{6}\,\delta
        +\sqrt{6}\,\|\mu-\mu_k\|_2.
\end{aligned}
\]
Consequently,
\[
    \|\widehat\mu-\mu\|_2
    \le
    \sqrt{6}\,\delta
    +(1+\sqrt{6})\|\mu-\mu_k\|_2.
\]
Thus $\widehat\mu$ is close to the best $k$-sparse approximation of
$\mu$, up to the sparse-direction estimation error. In particular, if
$\mu$ is $k$-sparse, then
$
    \|\widehat\mu-\mu\|_2
    \le \sqrt{6}\,
        \|\widetilde\mu-\mu\|_{2,k}.
$
\end{remark}
\subsection{Heavy-Tailed Distributions}
\label{sec:bdd-cov}

In this section, we prove \Cref{thm:informal-robust-mean} for heavy-tailed distributions.
As mentioned in the previous subsection,
we simply need to show that the samples satisfy the relevant stability condition from \Cref{def:stability}.

First, we make a slight simplification by restricting our attention to constant contamination of $\eps=0.01$, which is without loss of generality by the standard median-of-means pre-processing.
\begin{fact}[Median-of-Means Pre-Processing
{\cite[Fact~2.2]{DiaKLP22}}]
\label{fact:mom-preprocessing}
Suppose that there exists an efficient algorithm with the following
guarantee. Given $\sigma \in \mathbb{R}_{+}$ and a $0.01$-corrupted
set of
$
n \gg k^{4}\log d+\log(1/\tau)
$
samples from a distribution $D$ with mean $\mu$ and covariance
$\Sigma$ satisfying
$
\|\Sigma\|_{\op}\leq \sigma^2
$
and
$
\E_{x\sim D}\left[(x_j-\mu_j)^4\right]
    =O(\sigma^4),
    \text{for every }j\in[d],
$
the algorithm returns an estimate $\widehat{\mu}$ such that
$
\|\widehat{\mu}-\mu\|_{2,k}=O(\sigma)
$
with probability at least $1-\tau$.

Then there exists an efficient algorithm with the following guarantee.
Given $\epsilon\in(0,0.01]$ and an $\epsilon$-corrupted set of
$
n\gg
\frac{k^{4}\log d+\log(1/\tau)}{\epsilon}
$
samples from a distribution with mean $\mu$ and covariance $\Sigma$
satisfying
$
\|\Sigma\|_{\op}\leq 1
$
and
$
\E_{x\sim D}\left[(x_j-\mu_j)^4\right]
    =O(1), \text{for every }j\in[d],
$
the algorithm returns a mean estimate $\widehat{\mu}$ satisfying
$
\|\widehat{\mu}-\mu\|_{2,k}=O(\sqrt{\epsilon})
$
with probability at least $1-\tau$.
\end{fact}
The above is the immediate $k^4$-sample variant of the median-of-means preprocessing argument in \cite[Fact~2.2]{DiaKLP22}.
Indeed, the preprocessing forms $\Theta(\eps n)$ aggregate samples by averaging groups of size $\Theta(1/\eps)$.  Each corrupted original sample affects at most one aggregate sample, so only a constant fraction of the aggregate samples are corrupted.  The uncorrupted aggregate distribution has covariance $O(\eps)I$ and coordinatewise fourth moments $O(\eps^2)$, so the constant-corruption algorithm applies with $\sigma=\Theta(\sqrt\eps)$.  Requiring $\Theta(\eps n)\gg k^4\log d+\log(1/\tau)$ gives the displayed sample bound.
Thus, in the remainder of this sub-section, we consider only the constant contamination regime.
For this regime, the relevant stability  guarantee was proved in \cite{DiaKLP22}.

\begin{fact}[Stability after preprocessing
{\cite[Theorem~4.1]{DiaKLP22}}]
\label{thm:stability-after-preprocessing}
Let $S$ be a set of $n$ independent samples from a distribution
$P$ on $\mathbb{R}^d$, and let $T$ be a $0.01$-corruption of $S$.
Fix an integer $\ell\ge1$.
Suppose that $P$ has mean $\mu$ and covariance $\Sigma$ satisfying
$
\|\Sigma\|_{\op}\leq 1
$
and that
$
\E_{x\sim P}\left[(x_i-\mu_i)^4\right]
    =O(1),
  \text{for every }i\in[d].
$
Suppose moreover that $
n=\Omega\left(\ell^2\log d+\log(1/\tau)\right)$.
Then, 
there exists an $O(nd)$-time algorithm that transforms
$T$ into a multiset $\widetilde{T}\subseteq\mathbb{R}^d$ of the same
cardinality. With probability at least $1-\tau$ over the draw of $S$,
simultaneously for every $0.01$-corruption $T$ of $S$, there exist a
subset
$
\widetilde{S}\subseteq\widetilde{T},
|\widetilde{S}|\geq 0.95n,
$
and a vector $\mu'\in\mathbb{R}^d$ such that, simultaneously for every $\widetilde S'\subseteq\widetilde S$ with $|\widetilde S'|\geq0.99|\widetilde S|$,
\[
\|\mu_{\widetilde S'}-\mu'\|_{2,\ell}\le O(1),
\qquad
\sup_{W\in\cX_{\sdp,\ell}}
\left|\left\langle Q_{\widetilde S'}(\mu')-I,W\right\rangle\right|
\le O(1),
\]
and
$
\|\mu'-\mu\|_\infty
    \leq O\left(\frac{1}{\sqrt{\ell}}\right).
$
\end{fact}
Plugging in the above stability condition in the framework from \Cref{sec:reduction-sparse-pca}, we obtain the following tradeoff:
\begin{lemma}[Bounded-covariance sparse mean estimation from certificates]
\label{lem:bounded-covariance-from-sdp-pca}
There exists a sufficiently large universal constant $\beta>0$ such that the following holds.
Let $\ell\ge k$ and $s\ge1$.  Suppose that, for some $R>\beta$,
$\mathcal A_{\mathrm{cert}}$ is a $(k,k,\ell,s,R,\beta)$-sparse PCA
certificate routine in the sense of \Cref{def:sparse-pca-certificate-routine},
with runtime $T_{\mathrm{cert}}(m,d,k,k,\ell,s,\rho)$ on a set of size $m$ and
failure probability $\rho$.

Let $\eps\in(0,0.01]$ and $\delta\in(0,1)$, and let $T$ be an $\eps$-corruption of $n$ independent
samples from a distribution $P$ with mean $\mu$ and covariance $\Sigma$
satisfying $\|\Sigma\|_{\op}\le1$ and $\E_{x\sim P}\left[(x_j-\mu_j)^4\right]=O(1)$ for every $j\in[d]$.
If $n\gg(\ell^2\log d+\log(1/\delta))/\eps$,
then there exists an algorithm returning an estimator $\widehat\mu$ satisfying
$
    \|\widehat\mu-\mu\|_{2,k}
    =O\!\left(\sqrt{R\eps}\right)
$
with probability at least $1-\delta$.  
The runtime of the algorithm is
\[
    O\!\left(
        \eps nT_{\mathrm{cert}}\!\left(O(\eps n),d,k,k,\ell,s,\frac1{24n}\right)
            \log(1/\delta)
        +nd+\eps^2n^2(d+s^2+\ell)\log(1/\delta)
            +d\log^2(1/\delta)
    \right).
\]
\end{lemma}

\begin{proof}
We construct the algorithm by applying \Cref{fact:mom-preprocessing} to the constant-corruption routine obtained from \Cref{thm:stability-after-preprocessing} and \Cref{thm:robust-sparse-mean-from-pca}.  By \Cref{fact:mom-preprocessing}, it  suffices first to consider corruption level $\eps=0.01$.  By
\Cref{thm:stability-after-preprocessing}, after preprocessing, with
probability at least $1-\delta/2$, the transformed sample set $\widetilde T$
contains a subset
$
\widetilde{S}\subseteq\widetilde{T},
|\widetilde{S}|\geq 0.95n,
$
and a vector $\mu'\in\mathbb{R}^d$ such that $\|\mu'-\mu\|_\infty
    \leq O\left({1}/{\sqrt{\ell}}\right)$, and constants $\gamma_0,\Gamma_0=O(1)$ such that, simultaneously for every $\widetilde S'\subseteq\widetilde S$ with $|\widetilde S'|\geq0.99|\widetilde S|$,
\[
\|\mu_{\widetilde S'}-\mu'\|_{2,\ell}\le\gamma_0,
\qquad
\sup_{W\in\cX_{\sdp,\ell}}
\left|\left\langle Q_{\widetilde S'}(\mu')-I,W\right\rangle\right|
\le\Gamma_0.
\]
Now we show that the subset $\widetilde  S$ can be thought of as an $\eps'\eqdef 0.05$ corruption of $\widetilde T$ to apply \Cref{thm:robust-sparse-mean-from-pca}.
 For analysis purposes, form a reference multiset $G$ by adjoining $n-|\widetilde S|$ copies of $\mu'$ to $\widetilde S$.  Then $|G|=n$ and by replacing the added points it is an $\eps'$ corruption of $\widetilde T$. 
 We verify the stability condition of \Cref{def:stability} for $G$. Mainly, for every $W\in\cX_{\sdp,\ell}$ we have
\[
    \langle Q_G(\mu'),W\rangle
    =\frac{|\widetilde S|}{n}
      \langle Q_{\widetilde S}(\mu'),W\rangle
    \le1+\Gamma_0.
\]
For every $G'\subseteq G$ with
$|G'|\ge(1-10\eps')|G|=n/2$, positive semidefiniteness gives
\[
    0\le\langle Q_{G'}(\mu'),W\rangle
    \le\frac{n}{|G'|}\langle Q_G(\mu'),W\rangle
    \le2(1+\Gamma_0).
\]
Consequently, with $\Gamma_1\eqdef 1+2\Gamma_0$,
\[
    \sup_{W\in\cX_{\sdp,\ell}}
    \left|\left\langle Q_{G'}(\mu')-I,W\right\rangle\right|
    \le\Gamma_1.
\]
If $u$ is an $\ell$-sparse unit vector, then $uu^\top\in\cX_{\sdp,\ell}$.  Cauchy--Schwarz therefore yields
\[
\begin{aligned}
    |\langle u,\mu_{G'}-\mu'\rangle|^2
    &\le u^\top Q_{G'}(\mu')u
    \le2(1+\Gamma_0),\\
    \|\mu_{G'}-\mu'\|_{2,\ell}
    &\le\gamma_1\eqdef\sqrt{2(1+\Gamma_0)}.
\end{aligned}
\]
Since $k\le\ell$ and $\cX_{\sdp,k}\subseteq\cX_{\sdp,\ell}$, we have that $G$ is $(\eps',\gamma_1,\Gamma_1,k,k,\ell)$-stable with respect to $\mu'$.  

Since $\eps'<1/18$, set $\beta\eqdef1+48\Gamma_1$ and apply the routine of \Cref{thm:robust-sparse-mean-from-pca} to $\widetilde T$ with the given $(k,k,\ell,s,R,\beta)$-sparse PCA certificate routine and failure probability $\delta/2$.  Conditional on the preprocessing stability event, the routine returns an estimator $\widehat\mu$ such that, with probability at least $1-\delta/2$,
\[
    \|\widehat\mu-\mu'\|_{2,k}
    =O\!\left(\gamma_1+\sqrt{\eps'(R-1+\Gamma_1)}\right)
    =O(\sqrt R).
\]
Moreover,
\[
    \|\mu'-\mu\|_{2,k}
    \le
    \sqrt{k}\,\|\mu'-\mu\|_\infty
    =O\!\left(\sqrt{\frac{k}{\ell}}\right)
    =O(1),
\]
so $\|\widehat\mu-\mu\|_{2,k}=O(\sqrt R)$ for constant corruption.  A union bound over the preprocessing failure and the conditional estimation failure gives overall success probability at least $1-\delta$.

Finally, this is the constant-corruption algorithm required by the median-of-means preprocessing argument, with the constant error multiplied by $\sqrt R$.  Using aggregate sample size $\Theta(\ell^2\log d+\log(1/\delta))$ gives the claimed $O(\sqrt{R\eps})$ guarantee and sample complexity $(\ell^2\log d+\log(1/\delta))/\eps$.
The reduction (\Cref{fact:mom-preprocessing}) forms one collection of $O(\eps n)$ aggregate samples by averaging groups of size $\Theta(1/\eps)$ and rescales them before applying the constant-corruption routine.
Forming the averages (\Cref{fact:mom-preprocessing}) and preprocessing costs (\Cref{thm:stability-after-preprocessing}) $O(nd)$.
Moreover, the application of \Cref{thm:robust-sparse-mean-from-pca} uses a sample set of $O(\eps n)$ points.
Summing all costs proves the final running time.

\end{proof}
In particular, we obtain the following quadratic-time algorithm.
\begin{corollary}[Quadratic-Time Heavy-Tailed Robust Sparse Mean Estimation]
\label{cor:robust-sparse-mean-sdp-tradeoff}
For every $L\ge1$, there exists an algorithm with the following guarantee.
There is a sufficiently large universal constant $C$ such that the following holds.
Let $\eps\in(0,0.01]$, let $\delta\in(0,1)$, and let $T$ be an $\eps$-corruption of $n$ independent
samples from a distribution $P$ on $\R^d$ with mean $\mu$ and covariance
$\Sigma$ satisfying
$\|\Sigma\|_{\op}\le1$ and $
    \E_{x\sim P}\left[(x_j-\mu_j)^4\right]=O(1)
    \quad\text{for every }j\in[d]$.
If
$
    n\gg
    \frac{k^4\log d+\log(1/\delta)}{\eps},
$
then the algorithm returns $\widehat\mu$ such that
\[
    \|\widehat\mu-\mu\|_{2,k}
    =O\!\left(k^{1/(2L)}\sqrt\eps\right)
\]
with probability at least $1-\delta$, running in time
$O(d^2k^{O(L)}\log^2(d)\log^3(1/\delta)/\eps^2)$.
\end{corollary}

\begin{proof}
It remains to construct a sparse-PCA routine satisfying the
assumptions of \Cref{lem:bounded-covariance-from-sdp-pca}.
Let $C'>0$ be a sufficiently large constant. 
We use the algorithm from \Cref{cor:parameterized-radius-sparse-pca} with search-radius parameter $L$, target quadratic form parameter $R=C'\beta k^{{1}/{(2L)}}$, output quadratic-form parameter $\beta$, where $\beta$ is a sufficiently large constant, approximation error $\eta=0.1$, threshold parameter $\tau=1/(C'k)$ and $r=\lceil (C')^3k^2\rceil$.  
For suitable constants $\beta$ and $C'$, these parameters satisfy
the assumptions of \Cref{cor:parameterized-radius-sparse-pca}; in particular,
$(R-\tau k)/k^{1/(2L)}-\tau k\geq\beta+\eta$ and
$r\tau^2\geq\beta^2$.
The PSD assumptions in the SDP conclusion of \Cref{cor:parameterized-radius-sparse-pca} are automatic for the current empirical covariance $M\succeq0$.

Therefore, on input $M$, the sparse-PCA routine either certifies that
$\|M\|_{\op,k}<R$ or returns one of the following two certificates: either a matrix certificate
$W\in\cX_{\sdp,k}$ supported on $k^{O(L)}$ coordinates with
$\langle W,M\rangle\geq \beta$, or a unit vector certificate $u$ with
$\|u\|_0\le Ck^2$ and $u^\top M u\geq \beta$.
Thus this is a $(k,k,Ck^2,k^{O(L)},R,\beta)$ sparse PCA certificate routine in the sense of \Cref{def:sparse-pca-certificate-routine}.

Now since $\beta$ is sufficiently large we can apply \Cref{lem:bounded-covariance-from-sdp-pca} with $\ell=Ck^2$ and get a mean estimation algorithm that achieves error $O\!\left(k^{1/(2L)}\sqrt\eps\right)$. 
Since $\ell^2=C^2k^4$, the sample complexity required by \Cref{lem:bounded-covariance-from-sdp-pca} is $n=O\left((k^4\log(d)+\log(1/\delta))/\eps\right)$.  Substituting this sample size into the runtime bound of that lemma, together with \Cref{cor:parameterized-radius-sparse-pca}, gives $d^2k^{O(L)}\log^2(d)\log^3(1/\delta)/\eps^2$; the additional polynomial factors in $k$ are absorbed by $k^{O(L)}$ because $L\ge1$.  
\end{proof}
Replacing the quadratic-time sparse PCA variant with the subquadratic-time algorithm,
we obtain the following result.
\begin{corollary}[Subquadratic-Time Heavy-Tailed Robust Sparse Mean Estimation]
\label{cor:subquadratic-robust-sparse-mean-tradeoff}
For every integer $L\ge1$, and for every sufficiently large universal constant
$\gamma>0$, there exists an algorithm with the following guarantee.
Let $\eps\in(0,0.01]$, let $\delta\in(0,1)$, and let $T$ be an
$\eps$-corruption of $n$ independent samples from a distribution $P$ on
$\R^d$ with mean $\mu$ and covariance $\Sigma$ satisfying
    $\|\Sigma\|_{\op}\le1$ and $
    \E_{x\sim P}\left[(x_j-\mu_j)^4\right]=O(1)
    \quad\text{for every }j\in[d]$.
If
    $n\gg
    \frac{k^{O(\gamma)}\log d+\log(1/\delta)}{\eps}$,
then the algorithm returns $\widehat\mu$ such that
\[
    \|\widehat\mu-\mu\|_{2,k}
    =O\!\left(k^{1/(2L)}\sqrt\eps\right)
\]
with probability at least $1-\delta$.  Its running time is
\[
O\!\left(
    \left(
        d^{1.72}
        +
        d^{1.62+O(1/\gamma)}
    \right)
    k^{O(L+\gamma)}
    \operatorname{polylog}(d,1/\delta)
    \operatorname{poly}(1/\eps)
\right)\!.
\]
\end{corollary}

\begin{proof}
We verify the sparse PCA certificate routine required by \Cref{lem:bounded-covariance-from-sdp-pca}.  Fix a sufficiently large constant certificate threshold $\beta$, let $\eta=1/10$, and choose a sufficiently large universal constant $C>0$.  Set
\[
    R=C\beta k^{1/(2L)},\qquad
    \rho=\frac{1}{4k},
    \qquad
    \tau=(100k)^{-\gamma}.
\]
For these choices,
\[
    r_{\mathrm{sparse}}
    =\left\lceil \frac{4\beta^2}{\rho^2}\right\rceil
    \le C k^2,
    \qquad
    r_{\mathrm{dense}}
    =\left\lceil \frac{4\beta^2}{\tau^2}\right\rceil
    \le k^{C\gamma},
\]
after increasing $C$ if necessary.  Also, for sufficiently large $\gamma$ we have $\rho>12\tau$.
Since $\beta$ is a sufficiently large fixed constant and
$\rho=1/(4k)$, the parameter $\overline{D}$ in
\Cref{thm:subquadratic-correlation-detection-pca} satisfies
$\overline{D}=\max\{2\rho,\beta\}=\beta=O(1)$.

We now verify the gap condition.  Suppose a current empirical covariance matrix $M\succeq0$ satisfies $\|M\|_{\op,k}\ge R$, and let $v$ be a $k$-sparse unit vector such that $v^\top Mv\ge R$.  Applying \Cref{lem:quantitative-bounded-radius-reduction} with threshold $\rho$ and radius parameter $L$ gives a set $T_0\subseteq\supp(v)$ whose $\rho$-large-correlation graph has radius at most $L$ and
\[ 
    \lambda_{\max}(M_{T_0,T_0})
    \ge
    \frac{R-\rho k}{k^{1/(2L)}}-\rho k
    \ge \beta+\eta, 
\]
where the final inequality holds by choosing $C$ sufficiently large. Let $w_0$ be a unit vector supported on $T_0$ attaining $\lambda_{\max}(M_{T_0,T_0})$. Then $w_0^\top M w_0\ge\beta+\eta$, and $T_0$ is a witness set containing $\supp(w_0)$ with $\rho$-correlation radius at most $L$. Since $M\succeq0$, the local PSD assumptions in \Cref{cor:subquadratic-sdp} are automatic.

Suppose first that $d^{1.1}\ge2d r_{\mathrm{dense}}$. Run the subquadratic sparse PCA routine of \Cref{cor:subquadratic-sdp} with certificate threshold $\beta$, additive slack $\eta$, strong correlation threshold $\rho$, margin threshold $\tau$, and search-radius parameter $L$.  By the preceding paragraph, whenever $\|M\|_{\op,k}\ge R$, the assumptions of \Cref{cor:subquadratic-sdp} are satisfied.  Hence the routine either certifies that $\|M\|_{\op,k}<R$, or returns one of the following two certificates:
\begin{enumerate}[(i)]
    \item a matrix certificate $W\in\cX_{\sdp,k}$ supported on $k^{O(L)}$ coordinates and satisfying $\langle M,W\rangle\ge \beta$;
    \item a unit vector certificate $u$ with $\|u\|_0\le r_{\mathrm{dense}}\le k^{C\gamma}$ and $u^\top M u\ge \beta$.
\end{enumerate}
If instead $d^{1.1}<2d r_{\mathrm{dense}}$, run the quadratic
sparse PCA routine of \Cref{cor:parameterized-radius-sparse-pca} on $M$, with
correlation threshold $\rho$, degree parameter $r_{\mathrm{sparse}}$,
certificate threshold $\beta$, slack $\eta$, and search-radius parameter $L$.
The same gap and PSD arguments give the same two certificate types, with vector
sparsity at most $r_{\mathrm{sparse}}\le r_{\mathrm{dense}}$. Moreover,
\[
    d^2k^{O(L)}
    \le d^{1.72}(2r_{\mathrm{dense}})^{2.8}k^{O(L)}
    \le d^{1.72}k^{O(L+\gamma)},
\]
so this fallback is absorbed by the claimed runtime.
Thus this is a $(k,k,k^{C\gamma},k^{O(L)},R,\beta)$ sparse PCA certificate routine in the sense of \Cref{def:sparse-pca-certificate-routine}.

Applying \Cref{lem:bounded-covariance-from-sdp-pca} with $\ell=k^{C\gamma}$ gives
\[
    \|\widehat\mu-\mu\|_{2,k}
    =O\!\left(\sqrt{R\eps}\right)
    \le O\!\left(k^{1/(2L)}\sqrt\eps\right),  
\]
where the last inequality follows from the displayed choice of $R$ and $k\ge1$.  The sample complexity from \Cref{lem:bounded-covariance-from-sdp-pca} is
\[
    n\gg
    \frac{k^{2C\gamma}\log d+\log(1/\delta)}{\eps}
    =
    \frac{k^{O(\gamma)}\log d+\log(1/\delta)}{\eps}.
\]

For a certificate call on a set of size $m$, substituting $\rho=1/(4k)$ and $\tau=(100k)^{-\gamma}$ into \Cref{cor:subquadratic-sdp} gives runtime
\[
O\!\left(
    \left(
        d^{1.72}
        +
        d^{1.62+O(1/\gamma)}
    \right)
    k^{O(L+\gamma)}
    \operatorname{poly}(m,\log d)
    \log n
\right)\!.
\]
By \Cref{lem:bounded-covariance-from-sdp-pca}, every certificate call has $m=O(\eps n)$, and uses failure probability $1/(24n)$.  There are $O(\eps n\log(1/\delta))$ certificate calls overall.  The remaining filtering, score-computation, and candidate-selection costs are

\[
O\!\left(
    \eps^2n^2\left(d+k^{O(L)}+k^{O(\gamma)}\right)\log(1/\delta)
    +d\log^2(1/\delta)
\right)\!.
\]

Using the target sample size
\[
n=O\!\left(
    \frac{k^{O(\gamma)}\log d+\log(1/\delta)}{\eps}
\right)\!,
\]
all of these contributions, including preprocessing, are absorbed by
\[
O\!\left(
    \left(
        d^{1.72}
        +
        d^{1.62+O(1/\gamma)}
    \right)
    k^{O(L+\gamma)}
    \operatorname{polylog}(d,1/\delta)
    \operatorname{poly}(1/\eps)
\right)\!.
\]
This proves the stated runtime.
\end{proof}

Below, we explain why the preceding argument does not achieve the optimal dependence on $k$ in the sample complexity.

\begin{remark}[Sample complexity of sparse mean estimation for heavy-tailed distributions]
We note that \cite{DiaKLP22} achieved a $k^2$ dependence on the sparsity parameter in the sample complexity using a more computationally intensive algorithm, whereas our argument yields a $k^4$ dependence.

Indeed, our sparse-PCA certificate routine may return a $k^2$-sparse unit vector $u$. To use this certificate in the filtering argument, we must control $\left\langle Q_{\widetilde S'}(\mu')-I,uu^\top\right\rangle$. Although $uu^\top\in\mathcal X_{\sdp,k^2}$, it need not belong to $\mathcal X_{\sdp,k}$. Thus, SDP stability proved only at sparsity level $k$ does not cover this certificate outcome. However, when the certificate is a vector, sample-complexity bounds often scale linearly with the sparsity of that vector.

Unfortunately, the present stability-preprocessing argument does not suffice in this case. Coordinatewise truncation, which is used in \cite{DiaKLP22}, is nonlinear and need not preserve covariance uniformly over directions of larger sparsity; the truncation example in \cite[Appendix~A.4]{DiaKLP22} illustrates this phenomenon. Therefore, the preprocessing used in \cite{DiaKLP22} guarantees covariance control in $\mathcal X_{\sdp,k}$ at the selected sparsity scale, but does not automatically provide the corresponding control in $\mathcal X_{\sdp,k^2}$.

We therefore invoke their stability result at sparsity level $O(k^2)$, which requires $\widetilde O(k^4)$ samples. Recovering the $\widetilde O(k^2)$ bound would require either a different preprocessing procedure or a certificate routine whose vector output has only $O(k)$ sparsity.
\end{remark}

\subsection{Identity-Covariance Subgaussian Distributions}

\label{sec:identity-subgaussian}
In this section, we prove \Cref{thm:informal-robust-mean} for identity-covariance subgaussian distributions.
The key technical difference from the case of heavy-tailed distributions (and moving beyond the $\sqrt{\eps}$ error)
is that 
we need to filter until
the empirical covariance has sparse eigenvalue at most $1+\widetilde O(\eps)$ (as opposed to at most $O(1)$).
Thus, the sparse PCA primitive must detect a $k$-sparse direction of quadratic form
$1+\widetilde\Omega(\eps)$ and return an efficiently usable certificate of quadratic form
$1+\widetilde{\Omega}(\eps)$.

This is an additive approximation problem around the identity scale. However,
\Cref{lem:quantitative-bounded-radius-reduction} does not by itself give the needed additive
approximation guarantee when applied directly to the unshifted empirical covariance. Indeed, a
constant-factor loss in a quantity of order $1+\eps$ may destroy the additive excess above $1$.
The subgaussianity assumption gives the needed additional
control: the covariance of the remaining clean points is close to the
identity in every relevant sparse direction, rather than merely bounded from
above. 
We therefore apply the sparse PCA search to the shifted matrix $A=\Sigma-(1-\kappa)I$, where $\Sigma$ is the empirical covariance around the current empirical mean and $\kappa$ is the target additive covariance scale. After this shift, a direction of variance
$1+\widetilde\Omega(\eps) $ becomes a direction of shifted value $\widetilde\Omega(\eps) + \kappa$, while the
clean shifted covariance level is only $O(\kappa)$. The required approximation is therefore a 
multiplicative  constant at scale $\kappa = \widetilde\Theta(\eps)$, rather than additive at scale one.

The shift is used only to compute the sparse PCA certificate; the
filtering algorithm uses the same score definition as before.  Therefore,
\Cref{thm:robust-sparse-mean-from-pca} still applies.

We start with the following lemma stating that the samples from an isotropic subgaussian distribution satisfy the relevant stability condition:
\begin{restatable}[Identity-covariance subgaussian stability]{lemma}{subgaussianstability}
\label{lem:subgaussian-stability}

Let $G$ contain $n$ independent samples from a mean-$\mu$, identity-covariance subgaussian distribution on $\R^d$.  Fix $k\le\ell$, $\eps\in(0,0.01]$, and $\delta\in(0,1)$.  If
$
    n\ge C{((k^2+\ell)\log d+k^2\log(1/\delta))}/{\eps^2}
$
for a sufficiently large universal constant $C$, then with probability at least $1-\delta$, $G$ is $(\eps,O(\eps\sqrt{\log(1/\eps)}),O(\eps\log(1/\eps)),k,k,\ell)$-stable with respect to $\mu$.
On the same high-probability event, every $G'\subseteq G$ with $|G'|\ge(1-10\eps)|G|$ also satisfies
\[
    \inf_{\|u\|_2=1,\ \|u\|_0\le\ell}
    u^\top\operatorname{Cov}[G']u
    \ge1-O(\eps\log(1/\eps)).
\]
\end{restatable}

We refer the reader to \Cref{app:subgaussian-stability} for the proof.

Combining this stability guarantee with the sparse PCA certification reduction in \Cref{sec:reduction-sparse-pca},
we obtain the following algorithm that runs in quadratic time:
\begin{corollary}[Quadratic-Time Identity-Covariance Subgaussian Robust Sparse Mean Estimation]
\label{cor:identity-subgaussian-sparse-mean-tradeoff}
For every integer $L\ge1$, there is an algorithm with the following guarantee.  Let $T$ be an $\eps$-corruption of $n$ independent samples from a mean-$\mu$, identity-covariance subgaussian distribution on $\R^d$, where $\eps\in(0,0.01]$ and $\delta\in(0,1)$.  If $n\gg k^2\left(\log d+\log(1/\delta)\right)/\eps^2$, then with probability at least $1-\delta$ the algorithm returns $\widehat\mu$ satisfying $\|\widehat\mu-\mu\|_{2,k}\le O\!\left(\eps\sqrt{\log(1/\eps)}\,k^{1/(4L)}\right)$. Its running time is $d^2k^{O(L)}\polylog(d,1/\delta)\poly(1/\eps)$.
\end{corollary}

\begin{proof}
Fix a sufficiently large universal constant $C_1$. Apply \Cref{lem:subgaussian-stability} at the given value of $\eps$ with $\ell=C_2k^2$ and failure probability $\delta/2$, where $C_2$ is a sufficiently large universal constant, and let $\gamma=O\!\left(\eps\sqrt{\log(1/\eps)}\right)$ and $\Gamma=O\!\left(\eps\log(1/\eps)\right)$ be the mean- and covariance-stability bounds supplied by the lemma. Enlarge $\Gamma$ if necessary so that $\Gamma\ge\eps$, and set $\xi\eqdef\Gamma$. Choose $\kappa=C_1\xi$ and set
\[
    b\eqdef\kappa+48\xi,
    \qquad
    \eta\eqdef b/10,
    \qquad
    r\eqdef C_3(b+\eta) k^{1/(2L)},
    \qquad
    R\eqdef1-\kappa+r,
    \qquad
    \tau\eqdef C_4\frac{\xi}{k}.
\]
For $C_1,C_2,C_3$ sufficiently large and $C_4$ sufficiently small,
\[
    \frac{r-\tau k}{k^{1/(2L)}}-\tau k\ge b+\eta,
    \qquad
    C_2k^2\tau^2>b^2.
\]
At each filtering iteration, let $\Sigma$ denote the empirical covariance around the current empirical mean and define $A\eqdef \Sigma-(1-\kappa)I$. Run the radius-$L$ sparse PCA routine of \Cref{cor:parameterized-radius-sparse-pca} on $A$ with certificate threshold $b$, approximation slack $\eta$, correlation threshold $\tau$, and sparsity parameter $C_2k^2$. The lower stability bound in \Cref{lem:subgaussian-stability} gives the required $C_2k^2$-PSD condition for $A$. Thus, if $\|\Sigma\|_{\op,k}\ge R$, then $\|A\|_{\op,k}\ge r$, and \Cref{cor:parameterized-radius-sparse-pca} returns either $W\in\cX_{\sdp,k}$ or a $C_2k^2$-sparse unit vector $u$ with shifted value at least $b$.  In either case its unshifted value is at least
$1-\kappa+b=1+48\xi$.
Conversely, a \textsc{No} answer certifies $\|\Sigma\|_{\op,k}<R$.  Hence \Cref{cor:parameterized-radius-sparse-pca} implements an ordinary $(k,k,C_2k^2,k^{O(L)},R,1+48\xi)$ certificate routine.

The reduction in \Cref{thm:robust-sparse-mean-from-pca} assumes precisely such a sparse PCA certificate routine, so we apply it with mean-stability parameter $\gamma$ and covariance-stability parameter $\Gamma$ and failure probability $\delta/2$.  Conditional on the stability event, with probability at least $1-\delta/2$, it gives
\[
\begin{aligned}
    \|\widehat\mu-\mu\|_{2,k}
    &\le O\!\left(\gamma+\sqrt{\eps(R-1+\Gamma)}\right)\\
    &\le O\!\left(\eps\sqrt{\log(1/\eps)}\,k^{1/(4L)}\right).
\end{aligned}
\]  A union bound gives overall success probability at least $1-\delta$.  The sample bound follows from the stability lemma, and the runtime follows from the radius-$L$ certificate routine and \Cref{thm:robust-sparse-mean-from-pca}.  
\end{proof}
Similarly, we obtain the following sub-quadratic time variant.\begin{corollary}[Subquadratic-Time Identity-Covariance Subgaussian Robust Sparse Mean Estimation]
\label{cor:subquadratic-identity-subgaussian-sparse-mean}
For every integer $L\ge1$ and any constant $a$ greater than a sufficiently large universal constant, there is an algorithm with the following guarantee. Let $T$ be an $\eps$-corruption of $n$ independent samples from a mean-$\mu$, identity-covariance subgaussian distribution on $\R^d$, where $1\le k\le d$, $\eps\in(0,0.01]$, and $\delta\in(0,1)$. If
\[
    n\gg
    \frac{(k/\eps)^{O(a)}\log d+k^2\log(1/\delta)}{\eps^2},
\]
then, with probability at least $1-\delta$, the algorithm returns $\widehat\mu$ with
$
    \|\widehat\mu-\mu\|_{2,k}
    \le O\!\left(\eps\sqrt{\log(1/\eps)}\,k^{1/(4L)}\right).
$
Its running time is
\[
O\!\left(
    \left(
        d^{1.72}
        +
        d^{1.62+O(1/a)}
    \right)
    (k/\eps)^{O(a)}k^{O(L)}
    \polylog(d,1/\delta)\poly(1/\eps)
\right).
\]
\end{corollary}

\begin{proof}
Apply \Cref{lem:subgaussian-stability} {with failure probability $\delta/2$ and}
$
    \ell=(k/\eps)^{O(a)}
$
and let $\gamma=O\!\left(\eps\sqrt{\log(1/\eps)}\right)$ and
$\Gamma=O\!\left(\eps\log(1/\eps)\right)$ be the mean- and
covariance-stability bounds supplied by the lemma. Enlarge $\Gamma$ if
necessary so that $\Gamma\ge\eps$, and set $\xi\eqdef\Gamma$.
Choose the hidden
constant so that $\ell$ is at least $k$, $r_{\mathrm{dense}}+1$, and
$r_{\mathrm{sparse}}+1$ for the parameters below.  Set
\[
    \kappa=C_0\xi,
    \qquad
    \lambda=1-\kappa,
    \qquad
    b=\kappa+48\xi,
    \qquad
    \eta=b/10,
    \qquad
    r=C_1(b+\eta)k^{1/(2L)},
    \qquad
    R=1-\kappa+r,
\]
and take
\[
    \rho=\frac{\xi}{4k},
    \qquad
    \tau=\left(\frac{\xi}{100k}\right)^a.
\]
For the shifted covariance-detection routine, the parameter
$\overline D$ in \Cref{thm:shifted-subquadratic-correlation-detection-pca}
satisfies
\[
    \overline D
    =\max\{2\rho,b+\lambda\}
    =\max\left\{\frac{\xi}{2k},1+48\xi\right\}
    =O(1).
\]
For sufficiently large $a,C_1$, we have $\rho>12\tau$ and
\[
    \frac{r-\rho k}{k^{1/(2L)}}-\rho k\ge b+\eta.
\]
Moreover, the degree parameters of
\Cref{cor:shifted-subquadratic-sdp} satisfy
\[
    r_{\mathrm{sparse}}
    =O\!\left(\frac{b^2}{\rho^2}\right)=O(k^2),
    \qquad
    r_{\mathrm{dense}}
    =O\!\left(\frac{b^2}{\tau^2}\right)
    =(k/\xi)^{O(a)}.
\]
At each filtering iteration, let $\Sigma$ denote the empirical covariance
around the current empirical mean and define
$A\eqdef \Sigma-\lambda I=\Sigma-(1-\kappa)I$. Choose $C_0$ sufficiently large. The lower
stability bound in \Cref{lem:subgaussian-stability} then implies the local-PSD
condition for $A$ at every sparsity level required above.

Suppose first that
$d^{1.1}\ge2d r_{\mathrm{dense}}$.   Whenever $\|\Sigma\|_{\op,k}\ge R$, we have
$\|A\|_{\op,k}\ge r$.  The quantitative bounded-radius reduction and
\Cref{cor:shifted-subquadratic-sdp}, with certificate threshold $b$ and shift
$\lambda$, therefore return either a matrix certificate
$W\in\cX_{\sdp,k}$ supported on $k^{O(L)}$ coordinates or an
$\ell$-sparse unit vector certificate $u$, in both cases with
$A$-value at least $b$.  By the trace-one and unit-norm identities in that
corollary, the corresponding value for $\Sigma$ is at least
$\lambda+b=1-\kappa+b=1+48\xi$.

If instead $d^{1.1}<2d r_{\mathrm{dense}}$, run the quadratic
sparse PCA routine of \Cref{cor:parameterized-radius-sparse-pca} directly on
$A$, with correlation threshold $\rho$, degree parameter
$r_{\mathrm{sparse}}$, certificate threshold $b$, and slack $\eta$.  It
produces the same two types of certificates in time
$d^2k^{O(L)}\poly(1/\xi)$.  Moreover,
\[
    d^2k^{O(L)}\poly(1/\xi)
    \le d^{1.72}(2r_{\mathrm{dense}})^{2.8}
       k^{O(L)}\poly(1/\xi)
    =d^{1.72}(k/\xi)^{O(a)}k^{O(L)}
    \le d^{1.72}(k/\eps)^{O(a)}k^{O(L)},
\]
so this fallback is absorbed by the claimed runtime.

In either case, a \textsc{No} answer certifies
$\|\Sigma\|_{\op,k}<R$.  Thus the procedure implements the certificate routine
required by \Cref{thm:robust-sparse-mean-from-pca}.  We apply that theorem
with mean-stability parameter $\gamma$ and covariance-stability parameter
$\Gamma$ and failure probability $\delta/2$.  Conditional on the stability event, with probability at least $1-\delta/2$, this gives
\[
    \|\widehat\mu-\mu\|_{2,k}
    \le O\!\left(\gamma+\sqrt{\eps(R-1+\Gamma)}\right)
    \le O\!\left(\eps\sqrt{\log(1/\eps)}\,k^{1/(4L)}\right).
\]
A union bound over stability and the conditional estimation guarantee gives overall success probability at least $1-\delta$.  The sample bound follows from \Cref{lem:subgaussian-stability}.
In the subquadratic case, $\overline D=O(1)$ and
\[
    \frac{\log(4\overline D/\rho)}
         {\log(\overline D/(3\tau))}
    =O(1/a).
\]
Also, $\overline D/\tau=(k/\xi)^{O(a)}\le(k/\eps)^{O(a)}$.  Substituting these bounds into
\Cref{cor:shifted-subquadratic-sdp}, and including the filtering and
score-computation costs from \Cref{thm:robust-sparse-mean-from-pca} gives the stated runtime.
\end{proof}

\section{Conclusion}

We introduced a relaxation of sparse PCA that augments the standard SDP relaxation by allowing the algorithm to return a $f(k)$-sparse vector, where $f(k)\gg k$. This additional flexibility yields quadratic-time algorithms in the highly sparse regime while preserving the sample-complexity guarantees for well-behaved families of distributions. More broadly, our results suggest that allowing a controlled increase in output sparsity may lead to faster algorithms for other sparse-estimation problems. The main open question is whether one can obtain an algorithm with running time $d^2\cdot\poly(k)$.

\paragraph{AI disclosure} 

The authors used GPT-5.5 as an interactive tool during parts of the development of this work and GPT-5.6 to improve its exposition.
The project began with the goal of designing an efficient robust sparse mean estimation algorithm for general bounded-covariance distributions. In an initial GPT-5.5-assisted line of investigation, the authors attempted to construct a counterexample to the Restarted Truncated Power Method (RTPM) by showing that its iterates could escape the support of the target vector. Although these attempts were unsuccessful, they led the authors to discover that certain elements of the target support are captured by individual RTPM steps (\Cref{lem:discover}) and eventually the $d^2 k^{O(k)}$ runtime algorithm and $k^{O(k)} \log d$ sample complexity. In subsequent attempts to improve this result, also assisted by GPT-5.5, the authors derived \Cref{lem:quantitative-bounded-radius-reduction}, the $d^2 k^{O(\log k)}$ runtime algorithm and $k^{O(\log k)} \log d$ sample complexity.

The remaining main results were developed independently by the authors. These include the reduction from sparse mean estimation to relaxed sparse PCA (\Cref{thm:robust-sparse-mean-from-pca}), the subquadratic algorithm (\Cref{thm:informal-subquadratic-sdp}), the observation that the SDP yields better sample complexity which led to the notion of $\|\cdot\|_{\mathrm{relaxed},k}$, and the robust mean estimation applications (\Cref{thm:informal-robust-mean}). The authors independently verified all mathematical arguments, results, citations, and conclusions and take full responsibility for the content of the manuscript.

\printbibliography
\appendix
\section*{Supplementary Material}
The supplementary material is organized as follows.
\Cref{app:related-work} reviews the related literature,
\Cref{app:omitted} collects the facts and lemmas used in the main analysis,
\Cref{app:raw-covariance-detection} proves the covariance-detection
lemma, \Cref{app:shifted-covariance-detection} develops its shifted sparse PCA
variant, \Cref{app:subgaussian-stability} proves the subgaussian stability
lemma,
and \Cref{app:rtpm-counterexample} contains the proof of our counterexample
for the restarted truncated power method \cite{KSTZ26}.

\section{Related work}\label{app:related-work}

\paragraph{Sparse PCA.}
The term sparse PCA encompasses several related but distinct tasks:
worst-case maximization of a quadratic form over sparse unit vectors,
estimation of a sparse leading eigenspace from samples, and detection of a
sparse spike. 
For general positive semidefinite inputs, the
cardinality-constrained objective is NP-hard and admits the standard
$\ell_1$-constrained SDP relaxation
\cite{dAsGJL2007,chan2016approximability}.  
A separate statistical literature assumes independent observations
from a population with a sparse leading vector or principal subspace.
Various methods such as diagonal screening, semidefinite relaxations, covariance thresholding, and
truncated-power methods have been analyzed for support or eigenspace recovery
\cite{johnstone2009consistency,amini2008high,CaiMW14,DM16,YZ13, Novikov23}. 
These statistical results estimate a distinguished
population component from samples and are therefore distinct from worst-case
approximation of a supplied matrix.
Still, these problems also exhibit statistical--computational gaps
\cite{BerRig13,BreBHLS21}.
Overall, these results concern exact-sparsity optimization
or model-based recovery, whereas we study a fixed-matrix bicriteria gap
problem.

A direct iterative comparison is \cite{KSTZ26}, whose restarted
truncated-power method assumes independent subgaussian samples, a sparse
population leading eigenvector separated by an eigengap, and fresh sample
blocks across iterations.  We assume neither a distinguished sparse leading
eigenvector nor an eigengap; instead, the input is a fixed matrix satisfying
the stated local positive-semidefiniteness conditions and admitting a
high-value sparse direction.  Our counterexample concerns the deterministic
fixed-matrix analogue that reuses the same matrix at every iteration, and
therefore does not contradict the statistical guarantee of \cite{KSTZ26}.

A close structural comparison is the block-diagonalization
framework of \cite{PZZ25}.  Their method thresholds and permutes the matrix,
then invokes a sparse-PCA solver on each connected block; the runtime gain
depends on the largest block being small, and the guarantee incurs additive
error.  We also threshold the matrix, but our argument needs only a
bounded-radius component induced by the unknown sparse witness.  This yields
a worst-case, one-sided multiplicative-gap certificate through either an
enlarged-support vector or an SDP witness.

\paragraph{Robust sparse estimation.}
High-dimensional robust estimation has seen substantial algorithmic progress since the first polynomial-time estimators~\cite{DiaKKLMS16-focs,LaiRV16}. 
Our work is part of a research thread that seeks to reduce the computational
overhead of robustness.  This program has produced faster algorithms for
robust mean estimation
\cite{CheDG19,DonHL19,DepLec22,LeiLVZ20,DiaKKLT22-cluster,DiaKPP22-streaming,DiaKPP23-huber-optimal},
covariance estimation \cite{CheDGW19}, principal component analysis
\cite{JamLT20,DiaKPP23-pca,JamKLPPT24}, list-decoding
\cite{CheMY20,DiaKKLT22-cluster}, and linear regression
\cite{CheATJFB20,DiaKPP23-huber-optimal}.  These algorithms address dense models and thus
require $\Omega(d)$ samples.

Focusing on robust sparse estimation, \cite{BalDLS17} gave the first computationally-efficient framework for robust estimation of sparse functionals. Subsequent work developed filtering, convex-programming, and nonconvex approaches for robust sparse mean estimation \cite{DiaKKPS19,CheDKGGS21,ZhuJS20}.
Other works have relaxed distributional assumptions~\cite{DiaKKPP22-colt,DiaKLP22,DiaHPT25}
or improved error under weaker contamination models~\cite{DiaKKPP24-icml}.

Broadly speaking, the existing results on robust sparse mean estimation can be categorized as follows:
(i) SDP-based algorithms for broad distribution families (e.g., heavy-tailed distributions with unknown covariance) that run in $\Omega(d^4)$ time~\cite{DiaKLP22} and (ii) faster algorithms for highly structured distributions (e.g., identity-covariance Gaussians) that run in quadratic time~\cite{DiaKKPS19} or in subquadratic time~\cite{Pensia24-subquad}. 
The algorithms in (ii) exploit a Gaussian-specific
property (to be precise, bounded variance of quadratic forms) that need not hold for  general identity-covariance subgaussian
distributions; see \cite[Section 5]{Pensia24-subquad}.  
Our work narrows this divide by extending  quadratic
and subquadratic algorithms beyond this highly-structured  setting.

\section{Omitted Facts}\label{app:omitted}

Here we state the standard convergence guarantee for the ordinary power method, which we use to test the largest eigenvalues of the candidate support matrices.
\begin{lemma}[Random-start ordinary power method]
\label{lem:power-method-threshold}
Let $\Sigma\in\mathbb{R}^{d\times d}$ be positive semidefinite, and let
$
0<\lambda_2<\lambda_1$.
Assume that
$
\lambda_{\max}(\Sigma)\geq \lambda_1.
$
Let $g\sim\mathcal{N}(0,I_d)$, define
$
v_0=\frac{g}{\|g\|_2},
$
and run $q$ steps of the normalized power method:
\[
v_t=\frac{\Sigma v_{t-1}}{\|\Sigma v_{t-1}\|_2},
\qquad t=1,\ldots,q.
\]
Fix $\delta\in(0,1)$. If
\[
q
\geq
\max\left\{
0,\,
\left\lceil
\frac{
\log\left(
\frac{16d\lambda_2}
{\pi\delta^3(\lambda_1-\lambda_2)}
\right)
}{
2\log(\lambda_1/\lambda_2)
}
\right\rceil
\right\},
\]
then
\[
\mathbb{P}\left(
v_q^\top\Sigma v_q\geq\lambda_2
\right)
\geq 1-\delta.
\]

\end{lemma}

\begin{proof}
Let
$
\mu_1\geq\mu_2\geq\cdots\geq\mu_d\geq0
$
be the eigenvalues of $\Sigma$, with corresponding orthonormal
eigenvectors $u_1,\ldots,u_d$. Thus
\[
\Sigma=\sum_{i=1}^d\mu_i u_i u_i^\top,
\qquad
\mu_1=\lambda_{\max}(\Sigma)\geq\lambda_1.
\]
Write the Gaussian initialization in the eigenbasis of $\Sigma$:
$
g=\sum_{i=1}^d \alpha_i u_i.
$
By rotational invariance of the standard Gaussian distribution,
$
\alpha_1,\ldots,\alpha_d
$
are independent standard Gaussian random variables.

The intermediate normalizations in the power method do not change the
direction of the iterate. Hence, almost surely,
$
v_q
=
\frac{\Sigma^q g}{\|\Sigma^q g\|_2}.
$
Indeed, since $\mu_1>0$ and $\alpha_1\neq0$ almost surely,
$\Sigma^qg\neq0$ almost surely.

Define the Rayleigh quotient
$
R_q:=v_q^\top\Sigma v_q.
$
Using the eigenbasis expansion,
\[
R_q
=
\frac{
\sum_{i=1}^d \mu_i^{2q+1}\alpha_i^2
}{
\sum_{i=1}^d \mu_i^{2q}\alpha_i^2
}.
\]
Let
$
D_q:=\sum_{i=1}^d\mu_i^{2q}\alpha_i^2.
$
Then $D_q>0$ almost surely, and
\begin{align*}
D_q(R_q-\lambda_2)
&=
\sum_{i=1}^d
(\mu_i-\lambda_2)\mu_i^{2q}\alpha_i^2.
\end{align*}

We lower-bound the positive contribution from the top eigendirection and
upper-bound all negative contributions. Since
$\mu_1\geq\lambda_1>\lambda_2$,
$(\mu_1-\lambda_2)\mu_1^{2q}\alpha_1^2\geq
(\lambda_1-\lambda_2)\lambda_1^{2q}\alpha_1^2$.
Moreover, if $\mu_i<\lambda_2$, then
\[
(\lambda_2-\mu_i)\mu_i^{2q}
\leq
\lambda_2\lambda_2^{2q}
=
\lambda_2^{2q+1}.
\]
Discarding all other nonnegative terms, we obtain
\begin{align}
D_q(R_q-\lambda_2)
&\geq
(\lambda_1-\lambda_2)\lambda_1^{2q}\alpha_1^2
-
\sum_{\{i:\mu_i<\lambda_2\}}
(\lambda_2-\mu_i)\mu_i^{2q}\alpha_i^2
\nonumber\\
&\geq
(\lambda_1-\lambda_2)\lambda_1^{2q}\alpha_1^2
-
\lambda_2^{2q+1}\sum_{i=1}^d\alpha_i^2
\nonumber\\
&=
(\lambda_1-\lambda_2)\lambda_1^{2q}\alpha_1^2
-
\lambda_2^{2q+1}\|g\|_2^2.
\label{eq:rayleigh-threshold-lower-bound}
\end{align}
Consider the event%
\[
\mathcal{E}
:=
\left\{
\alpha_1^2\geq\frac{\pi\delta^2}{8}
\right\}
\cap
\left\{
\|g\|_2^2\leq\frac{2d}{\delta}
\right\}.
\]
We first show that
$
\mathbb{P}(\mathcal{E})\geq1-\delta.
$
Let $Z\sim\mathcal{N}(0,1)$. Since the standard Gaussian density is at
most $1/\sqrt{2\pi}$, for every $a>0$,
$
\mathbb{P}(|Z|<a)
\leq
\frac{2a}{\sqrt{2\pi}}.
$
Taking
$
a=\delta\sqrt{\frac{\pi}{8}},
$
we obtain
\[
\mathbb{P}\left(
\alpha_1^2<\frac{\pi\delta^2}{8}
\right)
=
\mathbb{P}\left(
|\alpha_1|<\delta\sqrt{\frac{\pi}{8}}
\right)
\leq\frac{\delta}{2}.
\]

Also, $
\mathbb{E}\|g\|_2^2=d.
$
Therefore, by Markov's inequality,%
\[
\mathbb{P}\left(
\|g\|_2^2>\frac{2d}{\delta}
\right)
\leq
\frac{\delta}{2}.
\]
A union bound now gives
$
\mathbb{P}(\mathcal{E})
\geq1-\delta.
$
On the event $\mathcal{E}$, inequality
\eqref{eq:rayleigh-threshold-lower-bound} gives
\begin{align*}
D_q(R_q-\lambda_2)
&\geq
\frac{\pi\delta^2}{8}
(\lambda_1-\lambda_2)\lambda_1^{2q}
-
\frac{2d}{\delta}\lambda_2^{2q+1}.
\end{align*}
The right-hand side is nonnegative whenever
\[
\frac{\pi\delta^2}{8}
(\lambda_1-\lambda_2)\lambda_1^{2q}
\geq
\frac{2d}{\delta}\lambda_2^{2q+1}.
\]
Equivalently,
\[
\left(\frac{\lambda_1}{\lambda_2}\right)^{2q}
\geq
\frac{16d\lambda_2}
{\pi\delta^3(\lambda_1-\lambda_2)}.
\]
The assumed lower bound on $q$ guarantees this inequality. Consequently, on
$\mathcal{E}$, $D_q(R_q-\lambda_2)\geq0$. Since $D_q>0$ almost surely,
it follows that $R_q=v_q^\top\Sigma v_q\geq\lambda_2$. Finally,
$\mathbb{P}(\mathcal{E})\geq1-\delta$, and hence
$\mathbb{P}(v_q^\top\Sigma v_q\geq\lambda_2)\geq1-\delta$.
\end{proof}

\begin{corollary}[Complexity for a fixed relative threshold]
\label{cor:power-method-fixed-relative-threshold}
Under the assumptions of Lemma~\ref{lem:power-method-threshold}, suppose
that $\lambda_2\leq(1-\gamma)\lambda_1$ for some
$\gamma\in(0,1)$. Then it suffices to take
$q=O\!\left(\frac{\log(d/(\delta\gamma))}{\gamma}\right)$.  In
particular, when $\gamma$ and $\delta$ are fixed constants,
$q=O(\log d)$.
If $\Sigma$ is stored as a dense matrix, each power iteration costs
$O(d^2)$ arithmetic operations, so the total running time is $O(d^2\log d)$.
\end{corollary}

\begin{proof}
Since $\lambda_2/\lambda_1\leq1-\gamma$,
$\log(\lambda_1/\lambda_2)\geq\log(1/(1-\gamma))\geq\gamma$.
Also,
$\lambda_2/(\lambda_1-\lambda_2)
\leq(1-\gamma)/\gamma\leq1/\gamma$.
Substituting these inequalities into the iteration bound in
Lemma~\ref{lem:power-method-threshold} proves the claim.
\end{proof}
We next prove a standard concentration result for the empirical covariance matrix of subgaussian random vectors; see, e.g., \cite{Vershynin18}.

\begin{fact}[Concentration of the empirical covariance matrix of subgaussian random  vectors]
\label{fact:sparse-quadratic-concentration}
There exists a universal constant $C>0$ such that the following holds. Let
$x_1,\ldots,x_n\in\R^d$ be i.i.d.\ samples from a mean-zero distribution with
covariance matrix $\Sigma$, and let $
    \widehat\Sigma=\frac1n\sum_{i=1}^n x_i x_i^\top.
$
Let the subgaussian constant
$
    K
    =
    \max\left\{1,\sup_{\|u\|_2=1}
    \|\langle u,x_1\rangle\|_{\psi_2}\right\}.
$
Fix $1\le s\le d$ and $\eta,\delta\in(0,1)$. If
\[
    n
    \ge
    C K^4\frac{s\log(ed/s)+\log(2/\delta)}{\eta^2},
\]
then, with probability at least $1-\delta$,
\[
    \sup_{\|u\|_2=1,\ \|u\|_0\le s}
    \left|u^\top(\widehat\Sigma-\Sigma)u\right|
    \le\eta.
\]
\end{fact}

\begin{proof}
First fix a unit vector $u$ and define
$
    Z_i=\langle u,x_i\rangle$,
    and 
   $ Y_i=Z_i^2-\E Z_i^2$.
By the definition of $K$, $\|Z_i\|_{\psi_2}\le K$.  The square of a
subgaussian random variable is subexponential, and centering changes its
$\psi_1$ norm by at most an absolute factor.  Therefore,
$
    \|Y_i\|_{\psi_1}\le C_0K^2
$
for a universal constant $C_0$.  Moreover,
\[
    \frac1n\sum_{i=1}^nY_i
    =u^\top(\widehat\Sigma-\Sigma)u.
\]
Bernstein's inequality for centered subexponential random variables therefore
gives, for every $t>0$,
\[
    \Pr\!\left(
        \left|u^\top(\widehat\Sigma-\Sigma)u\right|>t
    \right)
    \le
    2\exp\!\left(
        -c n\min\left\{\frac{t^2}{K^4},\frac{t}{K^2}\right\}
    \right).
\]
Since $t\in(0,1)$ and $K\ge1$, this implies
\[
    \Pr\!\left(
        \left|u^\top(\widehat\Sigma-\Sigma)u\right|>t
    \right)
    \le
    2\exp\!\left(-c n\frac{t^2}{K^4}\right).
\]

Next fix a support $S\subseteq[d]$ with $|S|=s$.  The unit vectors supported
on $S$ form a copy of the Euclidean sphere $\mathbb S^{s-1}$.  It has a
$1/4$-net $\mathcal N_S$ satisfying %
\[
    |\mathcal N_S|
    \le
    \left(1+\frac{2}{1/4}\right)^s
    =9^s.
\]
There are at most
$
    \binom ds
    \le
    \left(\frac{ed}{s}\right)^s
$
such supports.  (Vectors supported on fewer than $s$ coordinates are covered
by padding their support to one of size $s$.)  Applying the fixed-direction
bound with $t=\eta/2$ and taking a union bound over every point of every net
shows that the probability of any violation is at most
\[
    2\left(\frac{9ed}{s}\right)^s
    \exp\!\left(-c n\frac{\eta^2}{4K^4}\right).
\]
The stated lower bound on $n$ makes this quantity at most $\delta$, after
adjusting the universal constant $C$.

It remains to pass from each net to every vector on the corresponding sphere.
Put $A=\widehat\Sigma-\Sigma$.  For a fixed support $S$, let
\[
    M_S
    =
    \sup_{\substack{\|u\|_2=1\\\supp(u)\subseteq S}}
    |u^\top Au|
    =\|A_{S,S}\|_{\mathrm{op}}.
\]
The standard quadratic-form net estimate gives, for any $\varepsilon$-net
$\mathcal N_S$ with $\varepsilon<1/2$,
\[
    M_S
    \le
    \frac{1}{1-2\varepsilon}
    \max_{v\in\mathcal N_S}|v^\top Av|.
\]
Indeed, for a unit vector $u$ and a net point $v$ with
$\|u-v\|_2\le\varepsilon$, symmetry of $A_{S,S}$ gives
\[
    |u^\top Au|
    \le
    |v^\top Av|+2\varepsilon M_S;
\]
taking the supremum over $u$ and rearranging proves the estimate.  With
$\varepsilon=1/4$, it follows that
\[
    M_S
    \le
    2\max_{v\in\mathcal N_S}|v^\top Av|.
\]
Thus, if every net point satisfies $|v^\top Av|\le\eta/2$, then every unit
vector supported on $S$ satisfies $|u^\top Au|\le\eta$.  Since the union bound
covered all supports, the conclusion holds simultaneously for every
$s$-sparse unit vector.
\end{proof}

We also use the following elementary consequence of entrywise covariance concentration.

\begin{fact}[Concentration for bounded-$\ell_1$ matrices]
\label{fact:l1-matrix-concentration}
There exists a sufficiently large constant $C>0$ such that the following holds. Let $x_1,\ldots,x_n\in\R^d$ be i.i.d.\ samples from a mean-zero subgaussian distribution with covariance matrix $\Sigma$, and let $\widehat\Sigma=\frac1n\sum_{i=1}^n x_i x_i^\top$. Fix $K\ge1$ and $\eta,\delta\in(0,1)$. If
\[
    n\ge C K^2\frac{\log(d/\delta)}{\eta^2},
\]
then, with probability at least $1-\delta$, simultaneously for every matrix $X\in\R^{d\times d}$ satisfying $\|X\|_1\le K$,
\[
    \left|\langle\widehat\Sigma-\Sigma,X\rangle\right|\le\eta.
\]
\end{fact}

\begin{proof}
Standard entrywise covariance concentration and a union bound give
\[
    \|\widehat\Sigma-\Sigma\|_\infty\le\frac{\eta}{K}
\]
with probability at least $1-\delta$ under the stated sample bound. On this event, H\"older's inequality yields, for every $X$ with $\|X\|_1\le K$,
\[
    \left|\langle\widehat\Sigma-\Sigma,X\rangle\right|
    \le
    \|\widehat\Sigma-\Sigma\|_\infty\|X\|_1
    \le\eta.
\]
\end{proof}

We will also need the following well-known fact about sparse norms.

\begin{fact}[Lemma 3.24 \cite{DiaKan22-book}]\label{lem: ksparsenorm}
For any $x,y \in \mathbb{R}^d$ with $y$ $k$-sparse, we have
\[
\|\topp_k(x)-y\|_2
\le
\sqrt{6}\,\|x-y\|_{2,k}.
\]
\end{fact}
We now state several standard facts concerning filtering-based mean estimation algorithms. Similar lemmas are proved in Section 2 of \cite{DiaKan22-book}.

\begin{lemma}[Cutoff score mass inequality]
\label{lem:score-mass-inequality}
Let $T=G\cup B$, let $N=|T|$, and let $0<\alpha\eqdef|B|/N\le2\eps$, where $\eps\le1/18$.  Let $m$ be the mean of $T$, let $W\succeq0$ satisfy $\operatorname{tr}(W)=1$, and define the raw scores $g(x)\eqdef(x-m)^\top W(x-m)$.
Let $L$ consist of the $\lfloor2\eps N\rfloor$ largest raw scores, and define $\tau(x)=g(x)\mathbf1\{x\in L\}$.  Put $H=G\setminus L$.  Suppose that $\Gamma\ge\eps$ and that the 
two sets $K\in\{G,H\}$ satisfy
\[
    \left|
    \frac1{|K|}\sum_{x\in K}(x-\mu)^\top W(x-\mu)-1
    \right|\le\Gamma.
\]
If, for the empirical covariance $M_T$ around $m$,
$
    \langle M_T,W\rangle>1+48\Gamma,
$
then
\[
    \sum_{x\in B}\tau(x)
    \ge
    \sum_{x\in G}\tau(x).
\]
\end{lemma}

\begin{proof}
We first upper-bound the cutoff score mass contributed by $G$ and then lower-bound the cutoff score mass contributed by $B$.
Write $V=\langle M_T,W\rangle$, $a=|G\cap L|$, and
\[
    S_\mu\eqdef\sum_{x\in G\cap L}(x-\mu)^\top W(x-\mu).
\]
Subtracting the stability bounds for $G$ and $H=G\setminus L$ gives
\[
\begin{aligned}
    S_\mu
    &\le |G|(1+\Gamma)-|H|(1-\Gamma)\\
    &\le a+2\Gamma N
    \le4\Gamma N,
\end{aligned}
\]
where we used $a\le|L|\le2\eps N$ and $\Gamma\ge\eps$.

Let $m_G,m_B$ be the two empirical means, put $z=m-\mu$, and set $d=z^\top Wz$.  The variance decomposition gives
\[
    (m-m_G)^\top W(m-m_G)
    \le\frac{\alpha}{1-\alpha}V.
\]
Jensen's inequality and stability for $G$ therefore imply
\[
    d
    \le2(m_G-\mu)^\top W(m_G-\mu)
       +2(m-m_G)^\top W(m-m_G)
    \le2(1+\Gamma)+\frac{2\alpha}{1-\alpha}V.
\]
Since $g(x)\le2(x-\mu)^\top W(x-\mu)+2d$, the total cutoff score of the points in $G$, denoted by $C$, satisfies

\[
\begin{aligned}
    C
    &\eqdef\sum_{x\in G}\tau(x)\\
    &\le2S_\mu+2ad\\
    &\le(16+8\eps)\Gamma N
       +\frac{16\eps^2}{1-2\eps}VN\\
    &\le17\Gamma N+18\eps^2VN.
\end{aligned}
\]
The second inequality uses $a\le2\eps N$, $\alpha\le2\eps, \Gamma\ge\eps$ and the final inequality follows from $\eps\le1/18$.

We now lower-bound the cutoff score mass contributed by $B$.  If
\[
    C_G\eqdef\frac1{|G|}\sum_{x\in G}(x-m_G)^\top W(x-m_G),
    \qquad
    \Delta\eqdef(m_B-m_G)^\top W(m_B-m_G),
\]
then $C_G\le1+\Gamma$ and
\[
\begin{aligned}
    \sum_{x\in G}g(x)
    &=|G|\bigl(C_G+\alpha^2\Delta\bigr)\\
    &\le(1-\alpha)N(1+\Gamma)+\alpha NV.
\end{aligned}
\]
Consequently,
\[
    \sum_{x\in B}g(x)
    \ge(1-\alpha)N\bigl(V-1-\Gamma\bigr).
\]
Because $|B|\le2\eps N$ and $|B|$ is integral, $|B|\le\lfloor2\eps N\rfloor=|L|$.  The definition of $L$ therefore gives
\[
    \sum_{x\in T}\tau(x)
    =\sum_{x\in L}g(x)
    \ge\sum_{x\in B}g(x).
\]
For $V>1+48\Gamma$, $\alpha\le2\eps$, $\eps\le1/18$, and $\Gamma\ge\eps$, we have $1-2\eps-36\eps^2>0$.  Hence
\[
\begin{aligned}
    &(1-\alpha)(V-1-\Gamma)-34\Gamma-36\eps^2V\\
    &\quad\ge(1-2\eps-36\eps^2)V
       -(1-2\eps)(1+\Gamma)-34\Gamma\\
    &\quad>(13-94\eps-1728\eps^2)\Gamma-36\eps^2\\
    &\quad\ge(13-130\eps-1728\eps^2)\Gamma\\
    &\quad\ge\frac49\Gamma>0.
\end{aligned}
\]
Therefore, $\sum_{x\in B}\tau(x)\ge C=\sum_{x\in G}\tau(x)$, as claimed.
\end{proof}

\begin{lemma}[Stopping implies sparse-norm accuracy]
\label{lem:stopping-implies-accuracy}
Let $T=G\cup B$, let $m=\mu_T$, and let $M_T$ be the empirical covariance around $m$.  Put $\alpha=|B|/|T|\le2\eps$, where $\eps\le1/18$ and $\Gamma\ge\eps$.  Suppose that $\|\mu_G-\mu\|_{2,k}\le\gamma$ and, for every unit $k$-sparse vector $v$, $v^\top Q_G(\mu)v\ge1-\Gamma$.
Suppose the stopping condition $\|M_T\|_{\op,k}<R$ holds.  Then
\[
    \|m-\mu\|_{2,k}
    =O\!\left(\gamma+\sqrt{\eps\bigl(R-1+\Gamma\bigr)}\right).
\]
\end{lemma}

\begin{proof}
The case $\alpha=0$ follows from the mean-stability assumption.  Otherwise write $m_G=\mu_G$ and let $v$ be a unit $k$-sparse vector attaining $\|m-\mu\|_{2,k}$.  Mean stability gives
\[
    |\langle v,m-m_G\rangle|
    \ge\|m-\mu\|_{2,k}-\gamma.
\]
Moreover,
\[
    \operatorname{Var}_{G}(\langle v,x\rangle)
    =v^\top Q_{G}(\mu)v-\langle v,m_G-\mu\rangle^2
    \ge1-\Gamma-\gamma^2.
\]
Variance decomposition therefore yields
\[
\begin{aligned}
    R
    &>v^\top M_Tv\\
    &\ge(1-\alpha)(1-\Gamma-\gamma^2)
       +\frac{1-\alpha}{\alpha}\langle v,m-m_G\rangle^2.
\end{aligned}
\]
Rearranging and using $\alpha\le2\eps$, $\eps\le1/18$, and $\Gamma\ge\eps$ gives
\[
    \|m-\mu\|_{2,k}
    \le\gamma+
    \sqrt{\frac{\alpha}{1-\alpha}
    \left(R-(1-\alpha)(1-\Gamma-\gamma^2)\right)}
    =O\!\left(\gamma+\sqrt{\eps(R-1+\Gamma)}\right).
\]
\end{proof}

\begin{lemma}[Randomized score filtering (see e.g. \cite{DiaKan22-book} Theorem 2.17)]
\label{lem:randomized-score-filter}
Consider an adaptive filtering process initialized at $T^{(0)}=G_0\cup B_0$.  At each nonterminal round $t$, the filtering process has a current decomposition $T^{(t)}=G_t\cup B_t$, chooses scores $\tau_t:T^{(t)}\to\R_{\ge0}$ with $\tau_{t,\max}\eqdef\max_{x\in T^{(t)}}\tau_t(x)>0$, and independently deletes each $x\in T^{(t)}$ with probability $\tau_t(x)/\tau_{t,\max}$.  Assume that whenever the filtering process has not stopped and at most $3|B_0|$ samples from $G_0$ have been removed before round $t$, the scores satisfy
\[
    \sum_{x\in B_t}\tau_t(x)
    \ge
    \sum_{x\in G_t}\tau_t(x).
\]
Then, with probability at least $2/3$, at most $3|B_0|$ samples from $G_0$ are removed over the entire filtering process.  Moreover, because every nonterminal filtering round deletes a maximum-score sample with probability one, the filtering process has at most $|T^{(0)}|$ filtering rounds.
\end{lemma}

\section{Fast Covariance Detection}
\label{app:raw-covariance-detection}
In this section, we prove the following lemma by reducing it to previously established results on correlation detection.

\rawcovariancedetection*

We first explain the reduction.  Correlation detection is
formulated as the task of  finding large  normalized inner products among vectors (see \Cref{fact:boolean-correlation-detection}), while the preceding lemma asks for
the pairs of coordinates whose  empirical covariance is large.  
First by looking at each coordinate of the centered data points as a vector we have that we need  the large unnormalized inner products of those vectors.
In order to do that we equalize the norms of
the coordinate vectors by padding each of them in a distinct orthogonal
direction.  This increases their dimension from $n$ to $n+d$, so applying the
 correlation detection algorithm directly to the padded vectors would introduce
an unacceptable dependence on $d$ through the vector dimension. However,  the standard approach used to get correlation detection for real valued vectors from correlation detection from boolean vectors (see \cite[Appendix A.4]{Pensia24-subquad})  avoids this
problem by applying a gaussian projection to the vectors making their dimension effectively $\log(d)$. Also because the padding matrix is diagonal, its projection can be
computed implicitly in nearly linear time in $d$.

We record the two ingredients used in the proof.  The first is the
Boolean correlation-detection theorem of \cite{Valiant15}, in the notation
needed here.

\begin{fact}[Boolean correlation detection {\cite[Theorem~2.1]{Valiant15}}]
\label{fact:boolean-correlation-detection}
Let $y_1,\ldots,y_d\in\{-1,1\}^m$, and let
$1\ge\rho_0>\tau_0>0$.  Suppose that there are at most $s$ pairs
$\{i,j\}$ such that
\[
    \frac{|\langle y_i,y_j\rangle|}
         {\|y_i\|_2\|y_j\|_2}>\tau_0.
\]
There is a randomized algorithm that, with probability $1-o(1)$, outputs
all pairs satisfying
\[
    \frac{|\langle y_i,y_j\rangle|}
         {\|y_i\|_2\|y_j\|_2}\ge\rho_0.
\]
Its running time is
\[
    \left(
        s m d^{0.62}
        +d^{1.62+2.4\frac{\log(1/\rho_0)}{\log(1/\tau_0)}}
    \right)
    \operatorname{poly}(\log d,1/\tau_0).
\]
The success probability can be amplified to $1-\delta$ with an additional
$O(\log(1/\delta))$ factor.
\end{fact}

The second ingredient is the Gaussian-sign projection used by
\cite[Lemma~A.4]{Pensia24-subquad}, following
\cite[Lemma~4.1]{Valiant15}.

\begin{fact}[Gaussian-sign correlation projection]
\label{fact:gaussian-sign-correlation-projection}
Let $v_1,\ldots,v_d\in\R^N$ be nonzero vectors, let
$\gamma,\delta\in(0,1)$, and let $G\in\R^{N\times m}$ have independent
standard Gaussian entries.  For $m\ge C\gamma^{-2}\log(d/\delta)$, define $y_i=\operatorname{sgn}(v_i^\top G)\in\{-1,1\}^m$.
Then, with probability at least $1-\delta$, simultaneously for every
$i\ne j$,
\[
\left|
    \frac{\langle y_i,y_j\rangle}{m}
    -\frac{2}{\pi}\arcsin\left(
        \frac{\langle v_i,v_j\rangle}
             {\|v_i\|_2\|v_j\|_2}
    \right)
\right|
\le\gamma.
\]
\end{fact}
Now we are ready to prove \Cref{fact:cordetect}.
\begin{proof}[Proof of \Cref{fact:cordetect}]
Let $\mu_T$ be the empirical mean of $T$, and let
$X\in\R^{d\times n}$ be the centered data matrix whose $t$-th column is
$z_t-\mu_T$.  Thus,
$
    \Sigma=\frac1nXX^\top.
$
Write $X_i$ for the $i$-th row of $X$.  For each $i\in[d]$, set
$
    a_i\eqdef\sqrt{n(D-\Sigma_{ii})},
$
which is real because $D\ge \Sigma_{ii}$, and let
$P\eqdef\operatorname{diag}(a_1,\ldots,a_d)$.  Consider the padded matrix
\[
    \widetilde X\eqdef[\,X\mid P\,]\in\R^{d\times(n+d)}.
\]
Its $i$-th row is $\widetilde X_i=(X_i,a_i e_i)$.  Consequently,
\[
    \|\widetilde X_i\|_2^2
    =\|X_i\|_2^2+a_i^2
    =n\Sigma_{ii}+n(D-\Sigma_{ii})
    =nD,
\]
while, for $i\ne j$, the padding directions are orthogonal and hence
\[
    \langle\widetilde X_i,\widetilde X_j\rangle
    =\langle X_i,X_j\rangle
    =n\Sigma_{ij}.
\]
It follows that the normalized correlation between the padded rows is
exactly
\[
    \frac{|\langle\widetilde X_i,\widetilde X_j\rangle|}
         {\|\widetilde X_i\|_2\|\widetilde X_j\|_2}
    =\frac{|\Sigma_{ij}|}{D}.
\]
Set $\widehat\rho\eqdef\rho/D$ and
$\widehat\tau\eqdef\tau/D$.
By the definition of $D$, we have $\widehat\rho\le1/2$, and the assumption
$\rho>12\tau$ gives $\widehat\rho>12\widehat\tau$.  Apply
\Cref{fact:gaussian-sign-correlation-projection} to the rows of
$\widetilde X$ with $\gamma=\widehat\tau$ and failure probability
$\delta/2$.  Thus we may take
$m=O\!\left(\frac{D^2}{\tau^2}\log(d/\delta)\right)$.
For every pair satisfying $|\Sigma_{ij}|<\tau$, the corresponding normalized
correlation has magnitude less than $\widehat\tau$.  Using
$|\arcsin(x)|\le2|x|$ for $|x|\le1$, its projected Boolean correlation is
at most
\[
    \frac{4}{\pi}\widehat\tau+\widehat\tau
    \le3\widehat\tau.
\]
Therefore, all but at most $s$ Boolean pairs have correlation at most
$3\widehat\tau$.  On the other hand, if $|\Sigma_{ij}|\ge\rho$, then its
projected Boolean correlation is at least
\[
    \frac{2}{\pi}\arcsin(\widehat\rho)-\widehat\tau
    \ge\frac{\widehat\rho}{2}-\widehat\tau
    \ge\frac{\widehat\rho}{4}.
\]
The last inequality follows from
$\widehat\rho>12\widehat\tau$.  We may therefore invoke
\Cref{fact:boolean-correlation-detection} with strong threshold
$\widehat\rho/4$ and margin threshold $3\widehat\tau$, since
$\widehat\rho/4>3\widehat\tau$.  After standard amplification, it finds
all pairs satisfying $|\Sigma_{ij}|\ge\rho$ with failure probability at most
$\delta/2$.  We first verify the Boolean correlation of every reported
candidate and discard candidates below $\widehat\rho/4$.  We then evaluate
$\Sigma_{ij}$ for each remaining pair and discard any pair with
$|\Sigma_{ij}|<\rho$.  On the projection event, every remaining pair has
$|\Sigma_{ij}|\ge\tau$, so there are at most $s$ such covariance evaluations and
their cost is absorbed by the claimed running time.

It remains to verify that the padding does not create a quadratic-time
preprocessing step.  The Gaussian matrix used above has
$n+d$ rows; split it as
\[
    G=\begin{bmatrix}G_1\\G_2\end{bmatrix},
    \qquad
    G_1\in\R^{n\times m},
    \quad
    G_2\in\R^{d\times m}.
\]
Then $\widetilde XG=XG_1+PG_2$.
The first term can be computed in $O(ndm)$ time.  Since $P$ is diagonal,
the second term is obtained by scaling the $i$-th row of $G_2$ by $a_i$
and costs only $O(dm)$ time.  In particular, neither $P$ nor
$\widetilde X$ needs to be materialized as a dense matrix.

Finally, substituting
\[
    \rho_0=\frac{\widehat\rho}{4}=\frac{\rho}{4D},
    \qquad
    \tau_0=3\widehat\tau=\frac{3\tau}{D}
\]
into \Cref{fact:boolean-correlation-detection} gives the exponent
\[
    1.62+2.4\frac{\log(4D/\rho)}{\log(D/(3\tau))}.
\]
The projection dimension $m$, the $O(ndm+dm)$ preprocessing time, and
standard success amplification are absorbed by
\[
    \operatorname{poly}\left(n,\log d,\frac{D}{\tau}\right)
    \log(1/\delta).
\]
Taking a union bound over the projection and Boolean-detection events proves
the lemma.
\end{proof}

\section{Subquadratic Sparse PCA for Shifted Matrices}
\label{app:shifted-covariance-detection}

We now extend algorithm in \Cref{sec:sub-quadratic}  to a diagonally shifted
matrix.  Since covariance-detection (\Cref{fact:cordetect})  must be applied to an empirical
covariance matrix. We use covariance-detection on the original matrix, whereas the remaining sparse PCA routines  require only the
appropriate sparse PSD assumptions.  We therefore use the empirical covariance
matrix to construct the off-diagonal graph and use the shifted matrix for all the remaining computations.

\begin{algorithm}[h]
\centering
\fbox{\parbox{6in}{
{\bf Input:}
A data set $T\subseteq\R^d$ with empirical covariance matrix
$\Sigma$, a shift $\lambda\ge0$, search-radius parameter $L$, certificate
threshold $\beta>0$, strong correlation threshold $\rho>0$, margin threshold
$\tau>0$, and failure probability $\delta$.\\
{\bf Output:}
Either a sparse unit vector $u\in\R^d$ with
$u^\top(\Sigma-\lambda I)u>\beta$, or a family $\mathcal B$ of candidate
supports.

\begin{enumerate}[leftmargin=*]
    \item Set $A\gets\Sigma-\lambda I$.  Query all diagonal
    entries of $\Sigma$.  If $A_{ii}=\Sigma_{ii}-\lambda>\beta$ for some
    $i\in[d]$, return $e_i$.  Otherwise, set
    $D\gets\max\{2\rho,\max_{i\in[d]}\Sigma_{ii}\}$.

    \item Set
    $r_{\mathrm{dense}}\gets\lceil4\beta^2/\tau^2\rceil$ and
    $r_{\mathrm{sparse}}\gets\lceil4\beta^2/\rho^2\rceil$.

    \item Run \Cref{alg:dense-case-random-row-search} on $A$ with
    threshold $\tau$, degree parameter $r_{\mathrm{dense}}$, certificate
    threshold $\beta$, edge threshold $d^{1.1}$, and failure probability
    $\delta/2$.  If it outputs a vector, return that vector.

    \item Run the covariance-detection algorithm from
    \Cref{fact:cordetect} on $T$ with parameters
    $(\delta/2,\rho,\tau,D)$.  Let $E_\rho$ be its returned list and set
    $G_\rho\gets([d],E_\rho)$.

    \item Run \Cref{alg:graph-branching-search} on
    $(A,G_\rho,\rho,\beta,r_{\mathrm{sparse}},L)$ and return its output.
\end{enumerate}
}}
\caption{Shifted Correlation-Detection Branching Search.}
\label{alg:shifted-cd-branching-search}
\end{algorithm}

\begin{theorem}[Shifted subquadratic support recovery]
\label{thm:shifted-subquadratic-correlation-detection-pca}
Fix $\delta\in(0,1)$, $\lambda\ge0$, a certificate threshold
$\beta>0$, an approximation slack $\eta>0$, a strong correlation threshold
$\rho\in(0,1)$, and a margin threshold $\tau\in(0,1)$.  Set
\[
    r_{\mathrm{dense}}
    =\left\lceil\frac{4\beta^2}{\tau^2}\right\rceil,
    \qquad
    r_{\mathrm{sparse}}
    =\left\lceil\frac{4\beta^2}{\rho^2}\right\rceil,
    \qquad
    \overline D=\max\{2\rho,\beta+\lambda\}.
\]
Assume that $\rho>12\tau$ and
$d^{1.1}\ge2d r_{\mathrm{dense}}$.  Let $T\subseteq\R^d$ be a data set of
size $n$ with empirical covariance matrix $\Sigma$, and define
$A=\Sigma-\lambda I$.  Assume that $A$ is
$\max\{r_{\mathrm{dense}}+1,r_{\mathrm{sparse}}+1\}$-PSD.

Suppose there exist a unit vector $v\in\R^d$ and a set $E\subseteq[d]$ such
that $\|v\|_0\le k$, $v^\top A v\ge\beta+\eta$,
$\supp(v)\subseteq E$, and $E$ has $\rho$-correlation radius at most $L$
with respect to $A$.  Then, with probability at least $1-\delta$,
\Cref{alg:shifted-cd-branching-search} returns one of the following:
\begin{enumerate}[(i)]
    \item a unit vector $u$ such that
    $\|u\|_0\le r_{\mathrm{dense}}$ and $u^\top A u>\beta$;
    \item a family $\mathcal B$ of at most $d(L+1)$ candidate supports,
    each of size at most $(r_{\mathrm{sparse}}+1)^L$, such that
    $E\subseteq B$ for some $B\in\mathcal B$.
\end{enumerate}
Its running time is
\[
O\left(
\left(
\left(
    d^{1.72}
    +
    d^{1.62+2.4\frac{\log(4\overline D/\rho)}
                            {\log(\overline D/(3\tau))}}
\right)
\poly(n,\log d,\overline D/\tau)
+d r_{\mathrm{sparse}}(r_{\mathrm{sparse}}+1)^L
\right)
\log(1/\delta)
\right).
\]
\end{theorem}

\begin{proof}
Write
$
    E_\theta^A
    =\{\{i,j\}:i\ne j,\ |A_{ij}|\ge\theta\},
    E_\theta^\Sigma
    =\{\{i,j\}:i\ne j,\ |\Sigma_{ij}|\ge\theta\}.
$
Because the shift changes only the diagonal,
$E_\theta^A=E_\theta^\Sigma$ for every $\theta>0$.

The algorithm first scans the diagonal.  If $A_{ii}>\beta$, then $e_i$ is
a $1$-sparse unit vector with $e_i^\top A e_i>\beta$.  Otherwise,
$A_{ii}\le\beta$ for every $i$, and hence
$\Sigma_{ii}=A_{ii}+\lambda\le\beta+\lambda$.  Consequently,
the parameter used by covariance detection satisfies
$D=\max\{2\rho,\max_i\Sigma_{ii}\}
\le\max\{2\rho,\beta+\lambda\}=\overline D$.

Next run the density check from \Cref{lem:densesearch} on $A$ at threshold $\tau$.  Its hypotheses hold
because $A$ is $(r_{\mathrm{dense}}+1)$-PSD,
$r_{\mathrm{dense}}\tau^2>\beta^2$, and
$d^{1.1}\ge2d r_{\mathrm{dense}}$.  Thus, with probability at least
$1-\delta/2$, it either returns an $r_{\mathrm{dense}}$-sparse vector of
$A$-quadratic form greater than $\beta$, or establishes
$|E_\tau^A|<d^{1.1}$.  Condition on the latter event.  Since
$E_\tau^A=E_\tau^\Sigma$, the premise of \Cref{fact:cordetect} holds for
$\Sigma$ with sparsity bound $d^{1.1}$.  Covariance detection therefore
returns every pair in $E_\rho^\Sigma=E_\rho^A$ with probability at least
$1-\delta/2$.

The returned graph is exactly the graph required by
\Cref{lem:sparse-search} for the matrix $A$.  Applying that lemma with degree
parameter $r_{\mathrm{sparse}}$ gives either an
$r_{\mathrm{sparse}}$-sparse vector of $A$-quadratic form greater than
$\beta$, or the claimed family of candidate supports.  A union bound gives
the claimed success probability.

For the runtime, the diagonal scan, density check, covariance detection, and
branching search have the same costs as in the proof of
\Cref{thm:subquadratic-correlation-detection-pca}.  The detector is applied
to $\Sigma$ with parameter $D$.  Since $D\le\overline D$, the monotonicity
calculation in that proof allows us to replace $D$ by $\overline D$ in both
the exponent and the polynomial factor.  Combining the resulting detector
bound with the branching cost proves the displayed runtime.
\end{proof}

\begin{corollary}[Shifted subquadratic sparse PCA via the semidefinite relaxation]
\label{cor:shifted-subquadratic-sdp}
Under the assumptions of
\Cref{thm:shifted-subquadratic-correlation-detection-pca}, run
\Cref{alg:shifted-cd-branching-search}.  If it returns a vector, return that
vector.  Otherwise, compute a feasible additive-$\eta$ approximate solution
to $\operatorname{SDP}_k(A_{B,B})$ for every $B\in\mathcal B$, where
$A=\Sigma-\lambda I$, pad each
solution with zeros outside $B\times B$, and return the padded solution with
the largest objective value.  Then, with probability at least $1-\delta$,
the procedure returns one of the following:
\begin{enumerate}[(i)]
    \item a unit vector $u$ such that
    $\|u\|_0\le r_{\mathrm{dense}}$, $u^\top A u>\beta$, and hence
    $u^\top\Sigma u>\beta+\lambda$;
    \item a matrix $X\in\cX_{\sdp,k}$ supported on at most
    $(r_{\mathrm{sparse}}+1)^L$ coordinates such that
    $\langle A,X\rangle\ge\beta$, and hence
    $\langle\Sigma,X\rangle\ge\beta+\lambda$.
\end{enumerate}
Using a polynomial-time SDP solver, its running time is
\[
O\left(
\left(
    d^{1.72}
    +
    d^{1.62+2.4\frac{\log(4\overline D/\rho)}
                            {\log(\overline D/(3\tau))}}
\right)
\poly(n,\log d,\overline D/\tau)\log(1/\delta)
+d(r_{\mathrm{sparse}}+1)^{O(L)}/\eta^{O(1)}
\right),
\]
where the constant hidden in the final exponent depends on the SDP solver.
\end{corollary}

\begin{proof}
The direct-vector conclusion follows from
\Cref{thm:shifted-subquadratic-correlation-detection-pca}.  Otherwise, some
$B\in\mathcal B$ contains $\supp(v)$.  The matrix $vv^\top$ is feasible for
$\operatorname{SDP}_k(A_{B,B})$ and has objective value at least
$\beta+\eta$, so the corresponding additive-$\eta$ approximate solution has
$A$-objective value at least $\beta$.  Padding preserves feasibility,
objective value, and support size.  Finally, every returned vector has unit
norm and every returned matrix has trace one, so
$u^\top\Sigma u=u^\top A u+\lambda$ and
$\langle\Sigma,X\rangle=\langle A,X\rangle+\lambda$.
This proves the two unshifted-value conclusions and the corollary.
\end{proof}

\section{Proof of the Identity-Covariance Subgaussian Stability Lemma}
\label{app:subgaussian-stability}
Here we prove that samples from an identity-covariance subgaussian
distribution satisfy the stability conditions in \Cref{def:stability}.
For the mean and sparse-direction covariance conditions, we adapt the
standard dense-direction stability argument and take a union bound over
sparse supports. In fact, the same argument yields both mean and covariance
stability. It remains to establish stability over the class of positive
semidefinite matrices with bounded entrywise $\ell_1$ norm. Such a stability
result is known for standard Gaussian samples \cite{BalDLS17}. To transfer it
to general subgaussian distributions, we use the following representation of
a subgaussian vector as a sum of three standard Gaussian vectors. 
\begin{fact}[Three-Gaussian representation of a subgaussian vector \cite{hua2026talagrand}]
\label{fact:three-gaussian-representation}
Fix $K<\infty$. There exists a constant $c_K>0$, depending only on $K$, with the following property. For every dimension $d$ and every centered random vector $Z\in\R^d$ satisfying
\(
    \sup_{\|v\|_2=1}\|\langle v,Z\rangle\|_{\psi_2}\le K,
\)
there are random vectors $G_1,G_2,G_3$, such that each $G_j$ is marginally distributed as $N(0,I_d)$ and
\[
    c_K Z\stackrel{d}=G_1+G_2+G_3.
\]
\end{fact}

\subgaussianstability*

\begin{proof}
Write $G=\{X_1,\ldots,X_n\}$, put $Z_i=X_i-\mu$, and set
\[
    \alpha=10\eps,
    \qquad
    t=\log(12/\delta).
\]
Throughout the proof, constants may depend on the universal subgaussian
constant. Since the distribution is isotropic, for every unit vector $v$,
$\langle v,Z_i\rangle$ has mean zero, variance one, and uniformly bounded
$\psi_2$ norm.

For $1\le s\le d$, let
\[
    \mathcal U_s
    =\{v\in\R^d:\|v\|_2=1,\ \|v\|_0\le s\},
    \qquad
    a_s=s\log(ed/s)+t.
\]
For each support of size $s$, fix a $1/4$-net of its unit sphere, and let
$\mathcal N_s$ be the union of these nets. Then
\[
    |\mathcal N_s|
    \le 9^s\binom ds
    \le \exp\!\left(C s\log(ed/s)\right).
\]
The standard net argument loses only a universal constant when passing from
$\mathcal N_s$ to $\mathcal U_s$, both for linear forms and for symmetric
quadratic forms.

We first prove the sparse mean and rank-one covariance conclusions directly.
In both arguments, a deletion set is fixed before concentration is applied;
a union bound over all deterministic deletion sets then makes the resulting
estimate simultaneous over every data-dependent choice of a large subset.

\begin{claim}[Sparse mean stability]
\label{cl:subgaussian-mean-stability}
Simultaneously for every $G'\subseteq G$ with
$|G'|\ge(1-\alpha)n$,
\[
    \|\mu_{G'}-\mu\|_{2,\ell}
    = O(\eps\sqrt{\log(1/\eps)}).
\]
\end{claim}

\begin{proof}
Fix a deterministic set $D\subseteq[n]$ with $|D|\le\alpha n$ and a fixed
$v\in\mathcal N_\ell$. Since the indices in $D$ are fixed, the random
variables $\{\langle v,Z_i\rangle:i\in D\}$ are independent, centered, and
uniformly subgaussian. Hence, for every $r>0$,
\[
    \Pr\!\left(
        \left|\frac1n\sum_{i\in D}\langle v,Z_i\rangle\right|>r
    \right)
    \le
    2\exp\!\left(-c\frac{n^2r^2}{|D|}\right)
    \le
    2\exp\!\left(-c\frac{nr^2}{\alpha}\right),
\]
where the empty-set case is trivial. Moreover,
\[
    \sum_{j=0}^{\lfloor\alpha n\rfloor}\binom nj
    \le
    \exp\!\left(C\alpha n\log(e/\alpha)\right).
\]
Taking a union bound over $v\in\mathcal N_\ell$ and all such $D$, and then
using the net argument, shows that with probability at least $1-\delta/6$,
\begin{equation}
    \sup_{\substack{D\subseteq[n]\\|D|\le\alpha n}}
    \sup_{v\in\mathcal U_\ell}
    \left|\frac1n\sum_{i\in D}\langle v,Z_i\rangle\right|
    \le
    C\left(
        \alpha\sqrt{\log(e/\alpha)}
        +\sqrt{\frac{\alpha a_\ell}{n}}
    \right).
    \label{eq:deleted-linear-concentration}
\end{equation}
For the full sample, applying the same fixed-direction subgaussian
inequality without the subset union bound gives, with probability at least
$1-\delta/6$,
\begin{equation}
    \sup_{v\in\mathcal U_\ell}
    \left|\frac1n\sum_{i=1}^n\langle v,Z_i\rangle\right|
    \le C\sqrt{\frac{a_\ell}{n}}.
    \label{eq:full-linear-concentration}
\end{equation}

Now fix an arbitrary $G'\subseteq G$ with $m=|G'|\ge(1-\alpha)n$ and write
$D=[n]\setminus G'$. Since $n/m\le(1-\alpha)^{-1}\le2$, for every
$v\in\mathcal U_\ell$,
\[
    \frac1m\sum_{i\in G'}\langle v,Z_i\rangle
    =\frac nm\left(
        \frac1n\sum_{i=1}^n\langle v,Z_i\rangle
        -\frac1n\sum_{i\in D}\langle v,Z_i\rangle
    \right).
\]
Combining \Cref{eq:deleted-linear-concentration,eq:full-linear-concentration}
gives
\[
    \|\mu_{G'}-\mu\|_{2,\ell}
    \le
    C\left(
        \sqrt{\frac{a_\ell}{n}}
        +\alpha\sqrt{\log(e/\alpha)}
        +\sqrt{\frac{\alpha a_\ell}{n}}
    \right).
\]
The sample-size assumption implies $a_\ell/n\le C\eps^2$. Since
$\alpha=10\eps$, the right-hand side is at most
$C\eps\sqrt{\log(e/\eps)}=O(\eps\sqrt{\log(1/\eps)})$.
\end{proof}

\begin{claim}[Rank-one $\ell$-sparse covariance stability]
\label{cl:subgaussian-rank-one-stability}
Simultaneously for every $G'\subseteq G$ with
$|G'|\ge(1-\alpha)n$,
\[
    \sup_{u\in\mathcal U_\ell}
    \left|u^\top\bigl(Q_{G'}(\mu)-I\bigr)u\right| =O(\eps\log(1/\eps)).
\]
\end{claim}

\begin{proof}
For $u\in\mathcal U_\ell$, define $Y_i(u)\eqdef\langle u,Z_i\rangle^2-1$.
For every fixed $u$, the variables $Y_i(u)$ are independent, centered, and
have uniformly bounded $\psi_1$ norm. Fix a deterministic
$D\subseteq[n]$ with $|D|\le\alpha n$ and $u\in\mathcal N_\ell$.
Bernstein's inequality gives, for every $r>0$,
\[
    \Pr\!\left(
        \left|\frac1n\sum_{i\in D}Y_i(u)\right|>r
    \right)
    \le
    2\exp\!\left(
        -c n\min\left\{\frac{r^2}{\alpha},r\right\}
    \right).
\]
Taking a union bound over $u\in\mathcal N_\ell$ and all deletion
sets of size at most $\alpha n$ and then applying the quadratic-form net
argument shows that, with probability at least $1-\delta/6$,
\begin{equation}
    \sup_{\substack{D\subseteq[n]\\|D|\le\alpha n}}
    \sup_{u\in\mathcal U_\ell}
    \left|\frac1n\sum_{i\in D}
        \bigl(\langle u,Z_i\rangle^2-1\bigr)\right|
    \le
    C\left(
        \alpha\log(e/\alpha)
        +\sqrt{\frac{\alpha a_\ell}{n}}
        +\frac{a_\ell}{n}
    \right).
    \label{eq:deleted-quadratic-concentration}
\end{equation}
For the full sample, Bernstein's inequality and a union bound only over the
net give, with probability at least $1-\delta/6$,
\begin{equation}
    \sup_{u\in\mathcal U_\ell}
    \left|\frac1n\sum_{i=1}^n
        \bigl(\langle u,Z_i\rangle^2-1\bigr)\right|
    \le C\left(
        \sqrt{\frac{a_\ell}{n}}+\frac{a_\ell}{n}
    \right).
    \label{eq:full-quadratic-concentration}
\end{equation}

Fix an arbitrary $G'\subseteq G$ with $m=|G'|\ge(1-\alpha)n$ and put
$D=[n]\setminus G'$. The centered complement identity gives, for every
$u\in\mathcal U_\ell$,
\[
    u^\top\bigl(Q_{G'}(\mu)-I\bigr)u
    =\frac nm\left[
        \frac1n\sum_{i=1}^n
        \bigl(\langle u,Z_i\rangle^2-1\bigr)
        -\frac1n\sum_{i\in D}
        \bigl(\langle u,Z_i\rangle^2-1\bigr)
    \right].
\]
Using \Cref{eq:deleted-quadratic-concentration,eq:full-quadratic-concentration}
and $n/m\le2$, we obtain
\[
    \sup_{u\in\mathcal U_\ell}
    \left|u^\top\bigl(Q_{G'}(\mu)-I\bigr)u\right|
    \le
    C\left(
        \sqrt{\frac{a_\ell}{n}}+\frac{a_\ell}{n}
        +\alpha\log(e/\alpha)
        +\sqrt{\frac{\alpha a_\ell}{n}}
    \right).
\]
Since $a_\ell/n\le C\eps^2$ and $\alpha=10\eps$, this is at most
$C\eps\log(1/\eps)$.
\end{proof}

\begin{claim}[Lower $\ell$-sparse empirical covariance]
\label{cl:subgaussian-lower-covariance}
Simultaneously for every $G'\subseteq G$ with
$|G'|\ge(1-\alpha)n$,
\[
    \inf_{u\in\mathcal U_\ell}
    u^\top\operatorname{Cov}[G']u
    \ge1-O(\eps\log(1/\eps)).
\]
\end{claim}

\begin{proof}
Let $b=\mu_{G'}-\mu$. Centering at the empirical mean gives the identity
\[
    \operatorname{Cov}[G']
    =Q_{G'}(\mu)-bb^\top.
\]
Therefore, for every $u\in\mathcal U_\ell$, Claims
\ref{cl:subgaussian-mean-stability} and
\ref{cl:subgaussian-rank-one-stability} give
\begin{align*}
    u^\top\operatorname{Cov}[G']u
    &=u^\top Q_{G'}(\mu)u-\langle u,b\rangle^2\\
    &\ge1-C\eps\log(e/\eps)
       -C\eps^2\log(e/\eps)\\
    &\ge1-C'\eps\log(e/\eps).
\end{align*}
This proves the claim.
\end{proof}

\begin{claim}[$\cX_{\sdp,k}$-covariance stability]
\label{cl:subgaussian-Xk-stability}
Simultaneously for every $G'\subseteq G$ with
$|G'|\ge(1-\alpha)n$,
\[
    \sup_{W\in\cX_{\sdp,k}}
    \left|\left\langle Q_{G'}(\mu)-I,W\right\rangle\right|
    = O(\eps\log(1/\eps)).
\]
\end{claim}

\begin{proof}
We prove the full-sample and deleted-sample estimates separately. The
full-sample estimate is obtained directly for the original subgaussian
vectors. The deleted-sample estimate is transferred from Gaussian samples
using \Cref{fact:three-gaussian-representation}.

First set $\widehat\Sigma=\frac1n\sum_{i=1}^n Z_iZ_i^\top$.
For every pair $a,b\in[d]$, the centered random variable
$Z_{i,a}Z_{i,b}-\delta_{ab}$ has $\psi_1$ norm bounded by a constant that
depends only on the subgaussian-norm bound. Bernstein's inequality and a
union bound over the $d^2$ pairs give, with probability at least
$1-\delta/6$,
\[
    \|\widehat\Sigma-I\|_\infty
    \le
    C\left(
        \sqrt{\frac{\log(12d^2/\delta)}{n}}
        +\frac{\log(12d^2/\delta)}{n}
    \right).
\]
Since every $W\in\cX_{\sdp,k}$ satisfies $\|W\|_1\le k$, on this event
\begin{equation}
    \sup_{W\in\cX_{\sdp,k}}
    \left|\frac1n\sum_{i=1}^n Z_i^\top WZ_i-1\right|
    \le \eta,
    \qquad
    \eta\eqdef
    C\left(
        k\sqrt{\frac{\log(12d^2/\delta)}{n}}
        +\frac{k\log(12d^2/\delta)}{n}
    \right).
    \label{eq:full-Xk-three-gaussian}
\end{equation}
The sample-size assumption implies $\eta\le C\eps$.

For the deleted-sample estimate, we invoke the Gaussian
$\cX_{\sdp,k}$-stability result in Theorem~G.2 of \cite{BalDLS17}.
If
$g_1,\ldots,g_n$ are i.i.d.\ $N(0,I_d)$, then with probability at least
$1-\delta/18$, simultaneously for every $H\subseteq[n]$ with
$|H|\ge(1-\alpha)n$,
\begin{equation}
    \sup_{W\in\cX_{\sdp,k}}
    \left|
        \left\langle
            \frac1{|H|}\sum_{i\in H}g_ig_i^\top-I,W
        \right\rangle
    \right|
    \le \eta_G,
    \qquad
    \eta_G\le C\alpha\log(e/\alpha).
    \label{eq:li-gaussian-Xk-stability}
\end{equation}
Since \eqref{eq:li-gaussian-Xk-stability} holds simultaneously for every
\(H \subseteq [n]\) with \(|H| \ge (1-\alpha)n\), we may apply it to
\(H=[n]\) and to \(H=[n]\setminus D\), for every deletion set
\(D \subseteq [n]\) with \(|D| \le \alpha n\). Therefore,
\begin{align*}
    \frac1n\sum_{i\in D}g_i^\top Wg_i
    &=\frac1n\sum_{i=1}^n g_i^\top Wg_i
      -(1-|D|/n)\frac1{|H|}\sum_{i\in H}g_i^\top Wg_i\\
    &\le (1+\eta_G)-(1-\alpha)(1-\eta_G)\\
    &\le \alpha+2\eta_G
     \le C\alpha\log(e/\alpha).
\end{align*}

We now return to the original vectors $Z_i$. By
\Cref{fact:three-gaussian-representation}, take independent copies of the
entire three-Gaussian coupling and define
\[
    \widetilde Z_i
    =c_K^{-1}(G_{1,i}+G_{2,i}+G_{3,i}),
    \qquad i\in[n].
\]
Then $\widetilde Z_1,\ldots,\widetilde Z_n$ have exactly the same joint law
as $Z_1,\ldots,Z_n$. We may therefore prove the desired event for this
coupled copy and, to simplify notation, write $Z_i=\widetilde Z_i$ below.
For each fixed $j\in\{1,2,3\}$, the vectors
$G_{j,1},\ldots,G_{j,n}$ are i.i.d.\ $N(0,I_d)$; no independence between
different values of $j$ is needed. 
Applying the preceding
consequence of \eqref{eq:li-gaussian-Xk-stability} to each
$j\in\{1,2,3\}$ and taking a union bound shows, with probability at least
$1-\delta/6$, that all three Gaussian sequences satisfy the deleted-mass
bound simultaneously.

For every $W\in\cX_{\sdp,k}$, positive semidefiniteness gives the
deterministic inequality
\[
\begin{aligned}
    Z_i^\top WZ_i
    &=\frac1{c_K^2}
      \left\|W^{1/2}(G_{1,i}+G_{2,i}+G_{3,i})\right\|_2^2\\
    &\le\frac3{c_K^2}\sum_{j=1}^3G_{j,i}^\top WG_{j,i}.
\end{aligned}
\]
Consequently, on the preceding event, simultaneously for every
$D\subseteq[n]$ with $|D|\le\alpha n$,
\begin{equation}
    \sup_{W\in\cX_{\sdp,k}}
    \frac1n\sum_{i\in D}Z_i^\top WZ_i
    \le
    C_K\alpha\log(e/\alpha)
    \le C_K\eps\log(e/\eps),
    \label{eq:deleted-Xk-three-gaussian}
\end{equation}

Finally, fix an arbitrary $G'\subseteq G$ with $m=|G'|\ge(1-\alpha)n$ and
write $D=G\setminus G'$. For $W\in\cX_{\sdp,k}$, let
\[
    F_W=\frac1n\sum_{i=1}^nZ_i^\top WZ_i,
    \qquad
    R_W=\frac1n\sum_{i\in D}Z_i^\top WZ_i.
\]
By \eqref{eq:full-Xk-three-gaussian}, $|F_W-1|\le\eta$, and by
\eqref{eq:deleted-Xk-three-gaussian},
$0\le R_W\le\rho$, where
$\rho=C_K\alpha\log(e/\alpha)$. Writing
$r_D=|D|/n\le\alpha$, the centered complement identity gives
\[
    \left\langle Q_{G'}(\mu)-I,W\right\rangle
    =\frac{(F_W-1)-(R_W-r_D)}{1-r_D}.
\]
Since $R_W\ge0$, the numerator is at most $\eta+r_D\le\eta+\alpha$.
Since $R_W-r_D\le R_W\le\rho$, the numerator is at least $-\eta-\rho$. Therefore,
using $1-r_D\ge1-\alpha$,
\[
    \left|
        \left\langle Q_{G'}(\mu)-I,W\right\rangle
    \right|
    \le
    \frac{\eta+\alpha+\rho}{1-\alpha}
    \le C_K\eps\log(e/\eps),
\]
where the last inequality uses $\alpha=10\eps\le1/10$.
The bound is uniform over $W\in\cX_{\sdp,k}$ and $G'$, proving the claim.
\end{proof}

The intersection of the concentration events above has probability at least
$1-\delta$. Claims \ref{cl:subgaussian-mean-stability},
\ref{cl:subgaussian-rank-one-stability}, and
\ref{cl:subgaussian-Xk-stability} are precisely the three conditions in the
definition of
$(\eps,\gamma,\Gamma,k,k,\ell)$-stability, where
$\gamma=O\!\left(\eps\sqrt{\log(1/\eps)}\right)$ and
$\Gamma=O\!\left(\eps\log(1/\eps)\right)$, while
Claim \ref{cl:subgaussian-lower-covariance} is the additional
lower-covariance conclusion. This completes the proof.
\end{proof}

\section{A Counterexample to the Restarted Truncated Power Method }
\label{app:rtpm-counterexample}
In this section, we present an example demonstrating that, for a general PSD matrix, the Restarted Truncated Power Method, as defined in \cite{KSTZ26}, cannot approximate the optimal quadratic form within a factor of $2$, even with an infinite number of iterations.

Given a truncation parameter $r$, for each $\alpha\in[d]$ define
\begin{equation}
    x_\alpha^{(0)}=e_\alpha,
    \qquad
    x_\alpha^{(t+1)}
    =\frac{\topp_r(\Sigma x_\alpha^{(t)})}
    {\|\topp_r(\Sigma x_\alpha^{(t)})\|_2}.
    \label{eq:deterministic-rtpm}
\end{equation}
Ties in $\topp_r$ may be resolved arbitrarily. The next lemma 
gives an instance for which, even when $r=Ck^2$, every iterate has quadratic
form less than $1$, whereas a $k$-sparse unit vector has quadratic form at
least $2$.

\paragraph{Construction Intuition.}
For a positive semidefinite matrix, the quadratic form along the truncated
power method is non-decreasing. Indeed, if $x$ is an $r$-sparse unit vector,
then
\[
    \|\topp_r(\Sigma x)\|_2
    =\max_{\substack{\|u\|_2=1\\ \|u\|_0\leq r}}u^\top\Sigma x
    \geq x^\top\Sigma x,
\]
and Cauchy--Schwarz in the seminorm induced by $\Sigma$ shows that the next
normalized iterate has quadratic form at least $x^\top\Sigma x$. Thus, to
construct a counterexample, we must trap every trajectory on a low-value
support without allowing any iterate to attain a quadratic form comparable
to that of the target.

Moreover, \Cref{lem:discover} implies that, when the method starts at a
coordinate $e_i$ in the support of a target vector, a low-value first step
must retain the nontrivially correlated target coordinates. Our construction
therefore lets the first iterate retain the entire target support $[k]$, but
then forces the trajectory to leave that support before it can exploit the
large quadratic form there. We couple each signal coordinate $i\in[k]$ to a
set $F_i$, so that the first iterate is supported on $[k]\cup F_i$. The
set $F_i$ is in turn coupled to a set $G_i$, and the coupling is
chosen so that the next multiplication cancels every target coordinate except
possibly $i$. The trajectory is then trapped on $G_i$ together with a single
coordinate from $\{i\}\cup F_i$; all such supports have low value and do not
allow the full target support to be recovered.

A restart from a coordinate in $G_i$ enters the same low-value trap directly.
A restart from $F_i$, however, could reveal $[k]$.  We therefore
couple every coordinate $a\in F_i$ to a set of coordinates $P_a$. These sets trap all
restarts from $F$ while themselves having low quadratic form.

\begin{theorem}[Counterexample to the restarted truncated power method]
\label{lem:rtpm-counterexample}
For every fixed $C\geq1$, there exists $k_0(C)$ such that the following holds
for every $k\geq k_0(C)$. Let $r=\lceil Ck^2\rceil$. There exist a
dimension $d$, a matrix $\Sigma\in\R^{d\times d}$ with $\Sigma\succeq0$, and
a $k$-sparse unit vector $v\in\R^d$ such that
$
    v^\top\Sigma v=2.
$
Moreover, for every $\alpha\in[d]$ and every $t\geq0$, the iterates
$x_\alpha^{(t)}$ defined in \eqref{eq:deterministic-rtpm} satisfy
\[
    (x_\alpha^{(t)})^\top\Sigma x_\alpha^{(t)}<1.
\]
\end{theorem}

\begin{proof}
We first construct the counterexample matrix  $\Sigma$.

\paragraph{Coordinate sets.}
Set
$
    m=r-k,
    n=r-1.
$
For every $i\in[k]$, introduce disjoint sets $F_i$ and $G_i$ of sizes $m$
and $n$, respectively. For every coordinate
$a\in F:=\bigcup_iF_i$, introduce a set $P_a$ of size $n$. All of
these sets are pairwise disjoint. Write
\[
    f_i=\frac{\Ind_{F_i}}{\sqrt m},
    \qquad
    g_i=\frac{\Ind_{G_i}}{\sqrt n},
    \qquad
    p_a=\frac{\Ind_{P_a}}{\sqrt n}.
\]
Because their underlying coordinate sets are pairwise disjoint, these normalized indicator vectors are mutually orthonormal.

\paragraph{Target and cancellation blocks.}
Consider the following $\R^{k\times k}$ matrix
\[
   A\eqdef\sigma(I-P)+2P, \qquad    \sigma\eqdef\frac{1}{24\sqrt C},
    \qquad
    P\eqdef\frac1k\Ind_{[k]}\Ind_{[k]}^\top,
\]
and define
\[
    h_k:=2\sigma^2+\frac{4-\sigma^2}{k},
    \qquad
    W:=\frac{h_kI-A^2}{\sigma}.
\]
Thus, by construction,
\begin{equation}
    A^2+\sigma W=h_kI.
    \label{eq:rtpm-cancellation-identity}
\end{equation}
Since $A^2=\sigma^2(I-P)+4P$, the entries of $W$ are
\[
    W_{ii}=\sigma,
    \qquad
    W_{ij}=q:=-\frac{4-\sigma^2}{k\sigma}
    \quad(i\neq j).
\]
Define $B:\R^F\to\R^k$ by
\[
    B=\sum_{i=1}^k(We_i)f_i^\top.
\]
In particular, $Bf_i=We_i$. On the coordinates $[k]\cup F$, put
\begin{equation}
    \Sigma^{SF}
    =
    \begin{pmatrix}
        A & B\\
        B^\top & B^\top A^{-1}B
    \end{pmatrix}.
    \label{eq:rtpm-SF-block}
\end{equation}

\paragraph{Additional couplings}
Fix
\[
    \lambda=\frac34,
    \qquad
    b=\frac14,
    \qquad
    \eta=\frac1{12}.
\]
Define
\begin{align}
    \Sigma^{FG}
    &:=\sum_{i=1}^k
    \left(
        \frac{b}{\sqrt\lambda}f_i+\sqrt\lambda g_i
    \right)
    \left(
        \frac{b}{\sqrt\lambda}f_i+\sqrt\lambda g_i
    \right)^\top,
    \label{eq:rtpm-FG-block}\\
    \Sigma^{FP}
    &:=\sum_{a\in F}
    \left(
        \frac{\eta}{\sqrt\lambda}e_a+\sqrt\lambda p_a
    \right)
    \left(
        \frac{\eta}{\sqrt\lambda}e_a+\sqrt\lambda p_a
    \right)^\top.
    \label{eq:rtpm-FP-block}
\end{align}
All matrices are extended by zero outside their displayed coordinate
spaces, and we set $\Sigma\eqdef\Sigma^{SF}+\Sigma^{FG}+\Sigma^{FP}$.
The decoy blocks are sums of rank-one positive semidefinite matrices.
Moreover,
\[
    \Sigma^{SF}
    =
    \begin{pmatrix}
        A^{1/2}\\
        B^\top A^{-1/2}
    \end{pmatrix}
    \begin{pmatrix}
        A^{1/2}\\
        B^\top A^{-1/2}
    \end{pmatrix}^{\!\top},
\]
hence $\Sigma\succeq0$.

\begin{claim}[The sparse witness has value $2$]
\label{cl:rtpm-witness}
The vector
$v=\Ind_{[k]}/\sqrt{k}$
is $k$-sparse and satisfies $v^\top\Sigma v=2$.
\end{claim}

\begin{proof}
We have $Pv=v$, and therefore $Av=2v$. Since $v$ has no $F$ component
and all other blocks vanish on $v$,
 $   v^\top\Sigma v=v^\top Av=2$.
\end{proof}

\begin{claim}[The first target step]
\label{cl:rtpm-first-support}
For every $i\in[k]$,
\[
    e_i^\top\Sigma e_i
    =A_{ii}
    =\sigma+\frac{2-\sigma}{k}<\frac1{10}
\]
for all sufficiently large $k$, and
\[
    \topp_r(\Sigma e_i)=Ae_i+\sigma f_i.
\]
In particular,
\[
    \supp(\topp_r(\Sigma e_i))=[k]\cup F_i.
\]
\end{claim}

\begin{proof}
The formula for $e_i^\top\Sigma e_i$ follows from
$A=\sigma I+(2-\sigma)P$ and $P_{ii}=1/k$. Since
$\sigma\leq1/24$, it is below $1/10$ for all sufficiently large $k$.
By the block form of $\Sigma^{SF}$, we have
\[
\Sigma e_i = Ae_i + B^\top e_i.
\]
Since
\[
B^\top e_i=\sum_{j=1}^k W_{ji}f_j
\]
and \(W_{ii}=\sigma\), \(W_{ji}=q\) for \(j\neq i\), it follows that
\[
\Sigma e_i
=
Ae_i+\sigma f_i+q\sum_{j\neq i}f_j.
\]
The coordinate $i$ has value $\sigma+(2-\sigma)/k$, every other
coordinate in $[k]$ has magnitude $(2-\sigma)/k$, every coordinate in
$F_i$ has magnitude $\sigma/\sqrt m$, and every coordinate in $F_j$,
$j\neq i$, has magnitude $|q|/\sqrt m$. Since $|q|<\sigma$ for all
sufficiently large $k$ and
\[
    \frac{|q|}{\sqrt m}=O_C(k^{-2})
    <\frac{2-\sigma}{k},
\]
every coordinate in $[k]\cup F_i$ is strictly larger in magnitude than
every coordinate outside this set. Since $|[k]\cup F_i|=k+m=r$, the
top-$r$ truncation keeps exactly $[k]\cup F_i$.
\end{proof}

\begin{claim}[The first truncated iterate has small value]
\label{cl:rtpm-first-value}
Let
\[
    y_i^{(1)}:=Ae_i+\sigma f_i.
\]
For all sufficiently large $k$,
\[
    \frac{(y_i^{(1)})^\top\Sigma y_i^{(1)}}
    {\|y_i^{(1)}\|_2^2}<\frac15.
\]
\end{claim}

\begin{proof}
Set $d_0:=\frac{\eta^2+b^2}{\lambda}=\frac5{54}$.
Orthogonality of $Ae_i$ and $f_i$ gives
\[
    \|y_i^{(1)}\|_2^2
    =e_i^\top A^2e_i+\sigma^2
    =2\sigma^2+\frac{4-\sigma^2}{k}
    =h_k.
\]
Let
\[
    D:=
    \begin{pmatrix}
        A^{1/2}\\
        B^\top A^{-1/2}
    \end{pmatrix},
\]
so that $\Sigma^{SF}=DD^\top$. Using $Bf_i=We_i$ and
\eqref{eq:rtpm-cancellation-identity}, we have
\[
    D^\top y_i^{(1)}
    =A^{3/2}e_i+\sigma A^{-1/2}We_i
    =A^{-1/2}(A^2+\sigma W)e_i
    =h_kA^{-1/2}e_i.
\]
The two decoy blocks contribute
\[
    (y_i^{(1)})^\top
    (\Sigma^{FG}+\Sigma^{FP})y_i^{(1)}
    =\sigma^2\frac{b^2+\eta^2}{\lambda}
    =\sigma^2d_0.
\]
Therefore,
\[
    \frac{(y_i^{(1)})^\top\Sigma y_i^{(1)}}
    {\|y_i^{(1)}\|_2^2}
    =h_ke_i^\top A^{-1}e_i+\frac{\sigma^2d_0}{h_k}.
\]
Since
\[
    e_i^\top A^{-1}e_i
    =\left(1-\frac1k\right)\frac1\sigma+\frac1{2k},
\]
the last display converges, as $k\to\infty$, to
$2\sigma+d_0/2$. Finally,
\[
    2\sigma+\frac{d_0}{2}
    \leq\frac1{12}+\frac5{108}
    =\frac7{54}<\frac15.
\]
The strict gap proves the claim for all sufficiently large $k$.
\end{proof}

\begin{claim}[The signal trajectory enters a low-value trap]
\label{cl:rtpm-signal-trap}
For every $i\in[k]$,
\[
    \topp_r(\Sigma y_i^{(1)})
    =h_ke_i+b\sigma g_i.
\]
Its normalized version has quadratic form below $4/5$. Every subsequent
iterate remains supported on a set of the form
\[
    G_i\cup\{a\},
    \qquad
    a\in\{i\}\cup F_i,
\]
and has quadratic form below $4/5$.
\end{claim}

\begin{proof}
By \eqref{eq:rtpm-cancellation-identity}, the restriction of
$\Sigma y_i^{(1)}$ to $[k]$ is
\[
    A^2e_i+\sigma We_i=h_ke_i.
\]
Its $G$ component is exactly $b\sigma g_i$. The $F$-component of $\Sigma y_i^{(1)}$ is
\[
(\Sigma y_i^{(1)})_F
=
\sum_{j=1}^k
\left(
\left[WA+\sigma\left(WA^{-1}W+d_0I\right)\right]e_i
\right)_j f_j.
\]
Since $A$ and $W$ commute, \eqref{eq:rtpm-cancellation-identity} gives
\[
WA+\sigma WA^{-1}W
=
WA^{-1}(A^2+\sigma W)
=
h_kWA^{-1}.
\]
Therefore,
\[
(\Sigma y_i^{(1)})_F
=
\sum_{j=1}^k
\left(
\left[h_kWA^{-1}+\sigma d_0I\right]e_i
\right)_j f_j.
\]
The diagonal entry of this matrix is
\[
    2\sigma^2+\sigma d_0+O_C(k^{-1}),
\]
and each off-diagonal entry is $O_C(k^{-1})$. Hence every $F_i$
coordinate has magnitude
\[
    \frac{2\sigma^2+\sigma d_0+O_C(k^{-1})}{\sqrt m},
\]
all other $F$ coordinates are smaller, and every coordinate generated in a
set $P_a$ has magnitude $O_C(k^{-2})$. Since
\[
    b-(2\sigma+d_0)\geq\frac2{27}>0,
\]
every coordinate of $b\sigma g_i$ is larger than every $F$ or $P$
coordinate for all sufficiently large $k$. Since $|G_i|=r-1$, truncation
keeps exactly $G_i$ and $i$.

There is no direct $[k]$--$G$ block, and $g_i^\top\Sigma g_i=\lambda$.
Thus the Rayleigh quotient of $h_ke_i+b\sigma g_i$ is a convex
combination of $A_{ii}$ and $\lambda=3/4$, and is below $4/5$.

It remains to prove invariance. First consider
\[
    x=\alpha e_i+\beta g_i,
    \qquad
    \alpha,\beta\geq0,
    \qquad
    \frac{\alpha}{\beta}\leq\frac{h_k}{b\sigma}.
\]
Every $G_i$ coordinate of $\Sigma x$ has magnitude
$\lambda\beta/\sqrt n$. Every coordinate in $[k]\setminus\{i\}$ has
magnitude $A_{ij}\alpha$, and every $F_i$ coordinate has magnitude
$(\sigma\alpha+b\beta)/\sqrt m$. The ratios of the latter two quantities
to a $G_i$ coordinate are at most
\[
    \frac{4\sigma\sqrt C}{b\lambda}+O_C(k^{-1})
    =\frac89+O_C(k^{-1})<1
\]
and
\[
    \frac{b+2\sigma^2/b}{\lambda}+O_C(k^{-1})
    \leq\frac{19}{54}+O_C(k^{-1})<1,
\]
respectively. Thus all of $G_i$ is retained and the remaining coordinate
is either $i$ or one coordinate of $F_i$. If $i$ is retained, the ratio
$\alpha/\beta$ is multiplied by $A_{ii}/\lambda<1$. If some
$a\in F_i$ is retained, the new vector has the form
\[
    \gamma e_a+\beta g_i,
    \qquad
    \frac{\sqrt m\,\gamma}{\beta}\leq1.
\]

Now consider
\[
    x=\gamma e_a+\beta g_i,
    \qquad
    a\in F_i,
    \qquad
    c:=\frac{\sqrt m\,\gamma}{\beta}\leq1.
\]
Define
\begin{equation}
    d_F(k):=e_a^\top\Sigma e_a
    =\frac{\eta^2}{\lambda}
      +\frac{(WA^{-1}W)_{ii}+b^2/\lambda}{m}
    =\frac1{108}+O_C(k^{-2}).
    \label{eq:rtpm-dF}
\end{equation}
After multiplication, the aggregate coefficient on $g_i$ is
\[
    \lambda\beta+\frac{b\gamma}{\sqrt m},
\]
and the coefficient on $e_a$ is
\[
    \frac{b\beta}{\sqrt m}+d_F(k)\gamma.
\]
All other coordinates are smaller for sufficiently large $k$. Since
\[
    b+\frac1{108}=\frac7{27}<\lambda,
\]
the top-$r$ truncation retains $G_i$ and $a$. The updated scaled ratio is
\[
    c'
    =\frac{b+d_F(k)c}{\lambda+bc/m}
    \leq
    \frac{b+1/108+O_C(k^{-2})}{\lambda}
    <1.
\]
This family is therefore invariant. On its support the Rayleigh quotient
is at most the largest eigenvalue of
\[
    \begin{pmatrix}
        d_F(k) & b/\sqrt m\\
        b/\sqrt m & \lambda
    \end{pmatrix},
\]
which converges to $\lambda=3/4$ and is below $4/5$ for all sufficiently
large $k$.
\end{proof}

\begin{claim}[Restarts in $F$ and the sets $P_a$]
\label{cl:rtpm-FP-restarts}
For every $a\in F$ and every coordinate $u\in P_a$, the truncated-power
trajectory initialized at $e_a$ or $e_u$ enters the invariant cone
\[
    \mathcal C_a
    :=
    \{\alpha e_a+\beta p_a:\alpha,\beta\geq 0\}
    \subseteq \operatorname{span}\{e_a,p_a\}.
\]
Every vector in these trajectories has quadratic form below $4/5$ for
all sufficiently large $k$.
\end{claim}

\begin{proof}
Fix $a\in F_i$. In $\Sigma e_a$, the coordinate $a$ has value $d_F(k)$ (see \Cref{eq:rtpm-dF}),
and every coordinate in $P_a$ has magnitude $\eta/\sqrt n$. The largest
leakage into $[k]$ has magnitude $\sigma/\sqrt m$, while all other $F$
and $G$ leakages are $O_C(k^{-2})$. Since
\[
    \frac{\eta}{\sqrt n}>\frac{\sigma}{\sqrt m}
\]
for all sufficiently large $k$, the first truncation keeps exactly
$\{a\}\cup P_a$.

More generally, for $x=\alpha e_a+\beta p_a\in\mathcal C_a$,
\[
    \topp_r(\Sigma x)
    =
    \bigl(d_F(k)\alpha+\eta\beta\bigr)e_a
    +
    \bigl(\eta\alpha+\lambda\beta\bigr)p_a,
\]
because these coordinates dominate all leakage outside
$\{a\}\cup P_a$ for sufficiently large $k$. Thus $\mathcal C_a$ is
invariant under the truncated-power map.

The compression of $\Sigma$ to $\operatorname{span}\{e_a,p_a\}$ is
represented by
\[
    K_F(k)
    =
    \begin{pmatrix}
        d_F(k) & \eta\\
        \eta & \lambda
    \end{pmatrix}.
\]
Using \eqref{eq:rtpm-dF},
\[
    \lambda_{\max}(K_F(k))
    =\lambda+\frac{\eta^2}{\lambda}+O_C(k^{-2})
    =\frac{41}{54}+O_C(k^{-2})
    <\frac45
\]
for sufficiently large $k$. If $u\in P_a$, then
\[
    \Sigma e_u
    =\frac{\eta}{\sqrt n}e_a
     +\frac{\lambda}{\sqrt n}p_a,
\]
so this restart immediately enters $\mathcal C_a$. Finally,
\[
    e_a^\top\Sigma e_a=d_F(k)<\frac45,
    \qquad
    e_u^\top\Sigma e_u=\frac{\lambda}{n}<\frac45,
\]
so the initial basis vectors also have quadratic form below $4/5$.
\end{proof}
\begin{claim}[Restarts in $G$ enter the same trap]
\label{cl:rtpm-G-restarts}
For every $z\in G_i$, the truncated-power trajectory initialized at $e_z$
enters the invariant family from \Cref{cl:rtpm-signal-trap} after one
truncated multiplication. Every vector in this trajectory has quadratic
form below $4/5$.
\end{claim}

\begin{proof}
For $z\in G_i$,
\[
    \Sigma e_z
    =\frac{b}{\sqrt n}f_i+\frac{\lambda}{\sqrt n}g_i.
\]
Thus every coordinate in $G_i$ has magnitude $\lambda/n$, whereas every
coordinate in $F_i$ has magnitude $b/\sqrt{mn}$. Since
\[
    \frac{\lambda}{n}>\frac{b}{\sqrt{mn}}
    \quad\Longleftrightarrow\quad
    \lambda\sqrt m>b\sqrt n,
\]
the top-$r$ truncation keeps all $n=r-1$ coordinates of $G_i$ and one
arbitrary coordinate $a\in F_i$ for all sufficiently large $k$.
Therefore,
\[
    \topp_r(\Sigma e_z)
    =\frac{\lambda}{\sqrt n}g_i
     +\frac{b}{\sqrt{mn}}e_a.
\]
Up to normalization, this belongs to the second invariant family in
\Cref{cl:rtpm-signal-trap}, since its scaled ratio is
\[
    \frac{\sqrt m\,(b/\sqrt{mn})}{\lambda/\sqrt n}
    =\frac b\lambda
    =\frac13<1.
\]
Hence every subsequent iterate has quadratic form below $4/5$ by
\Cref{cl:rtpm-signal-trap}. Finally,
\[
    e_z^\top\Sigma e_z=\frac{\lambda}{n}<\frac45,
\]
so the initial vector also has quadratic form below $4/5$.
\end{proof}
The coordinate classes
\[
    [k],
    \qquad
    F,
    \qquad
    G,
    \qquad
    \bigcup_{a\in F}P_a
\]
exhaust the ambient space. By \Cref{cl:rtpm-first-support},
\Cref{cl:rtpm-first-value}, and \Cref{cl:rtpm-signal-trap}, every trajectory
initialized in $[k]$ has quadratic form below $4/5$. By
\Cref{cl:rtpm-FP-restarts}, the same holds for every trajectory initialized
in $F$ or a set $P_a$, and by \Cref{cl:rtpm-G-restarts} it holds
for every trajectory initialized in $G$. Thus, for every $\alpha\in[d]$
and $t\geq0$,
\[
    (x_\alpha^{(t)})^\top\Sigma x_\alpha^{(t)}
    <\frac45<1.
\]
Together with \Cref{cl:rtpm-witness}, this proves the lemma.
\end{proof}

\begin{remark}[Why the branching search succeeds]
\label{rem:rtpm-branching-succeeds}
By \Cref{cl:rtpm-first-support}, the support of the first truncated
multiplication from any signal coordinate contains all of $[k]$. Therefore,
a method that branches on this support and then computes the largest
eigenvalue on the resulting candidate support finds value at least
$v^\top\Sigma v=2$.
\end{remark} 

\end{document}